\documentclass[10pt, reqno]{amsart}
\usepackage{amsmath, amsfonts, amssymb}
\usepackage{xcolor}
\usepackage{mathrsfs}
\usepackage{lscape}
\usepackage{graphics}
\usepackage{amsthm}

\usepackage{tikz}
\usetikzlibrary{patterns,arrows.meta}

\allowdisplaybreaks[4] 

\title[R13 equations with Onsager boudnary conditions]{Linear regularized 13-moment equations with Onsager boundary conditions for general gas molecules}

\author{Zhenning Cai}
\address[Zhenning Cai]{Department of Mathematics, National University of Singapore,
  Level 4, Block S17, 10 Lower Kent Ridge Road, Singapore 119076}
\email{matcz@nus.edu.sg}

\author{Manuel Torrilhon}
\address[Manuel Torrilhon]{Applied and Computational Mathematics, RWTH Aachen University, Schinkelstrasse 2, 52062 Aachen, Germany}
\email{mt@acom.rwth-aachen.de}

\author{Siyao Yang}
\address[Siyao Yang]{Department of Mathematics, National University of Singapore,
  Level 4, Block S17, 10 Lower Kent Ridge Road, Singapore 119076}
\email{matsiya@nus.edu.sg}


\thanks{Zhenning Cai's work was supported by the Academic Research Fund of the Ministry of Education of Singapore under Grant No. A-0004592-00-00.}

\newcommand\dd{\mathrm{d}}
\newcommand\bxi{\boldsymbol{\xi}}
\newcommand\bx{\boldsymbol{x}}
\newcommand\bu{\boldsymbol{u}}
\newcommand\bn{\boldsymbol{\mathrm{n}}}

\newcommand{\ms}[1]{\mathscr{#1}}
\newcommand{\kn}{\mathrm{Kn}}

\newtheorem{theorem}{Theorem}
\newtheorem{proposition}[theorem]{Proposition}

\theoremstyle{remark}
\newtheorem*{remark}{Remark}

\begin{document}

\maketitle

\begin{abstract}
    We develop the steady-state regularized 13-moment equations in the linear regime for rarefied gas dynamics with general collision models. For small Knudsen numbers, the model is accurate up to the super-Burnett order, and the  resulting system of moment equations is shown to have a symmetric structure. We also propose Onsager boundary conditions for the moment equations that guarantees the stability of the equations. The validity of our model is verified by benchmark examples for the one-dimensional channel flows.
\end{abstract}

\section{Introduction} \label{sec:intro}
The modeling and computation of rarefied gas dynamics has been a classical research topic in the history of fluid mechanics. On one hand, significant progress has been made in the development of efficient solvers for the Boltzmann equation \cite{Dimarco2018, Hu2020fast, Liu2020, Su2020, Pareschi2022}; On the other hand, many researchers still try to avoid the high computational cost and find extensions of classical fluid models such as Euler equations and Navier-Stokes equations and hope that the new models are capable of describing the motion of moderately rarefied gases. The moment method introduced by Grad \cite{Grad1949} is one of the important approaches in this direction. While Grad's original method suffers from a number of deficiencies such as loss of hyperbolicity and convergence \cite{Muller, Cai2014, Cai2019}, many new ideas have been proposed in the recent years to improve its robustness and make moment methods more widely applicable \cite{Rana2018, Cai2020regularized, Bohmer2020, Sarna2021, Fox2022, Pichard2022}. In this work, we will study the regularized version of Grad's 13-moment equations, which are originally proposed in \cite{Struchtrup2003} for Maxwell molecules and extended to general gases in \cite{Struchtrup2013, Cai2020regularized}. In the literature, this model is called R13 equations for short. The R13 equations for Maxwell molecules have been verified for a variety of problems \cite{Taheri2009a, Ivanov2013, Timokhin2017, Claydon2017, DeFraja2022}. Recently, attentions have been drawn to the study of Onsager boundary conditions, which guarantees the stability of simulations of boundary value problems \cite{Beckmann2018, Sarna2017}. Here we will consider the formulation of Onsager boundary conditions for general R13 equations with arbitrary elastic collision models. The linearized Boltzmann equation will be taken as the base model, so that rigorous theory of stabilization can be established for the moment equations.

In the linearized setting, one typical form of the moment equations is as follows:
\begin{equation} \label{eq:me}
\mathbf{A}_0 \frac{\partial \bu}{\partial t} + \sum_{j=1}^3 \mathbf{A}_j \frac{\partial \bu}{\partial x_j} = \mathbf{L} \bu,
\end{equation}
where $\mathbf{A}_0$ is symmetric positive semidefinite, and $\mathbf{L}$ is symmetric negative semidefinite. The matrices $\mathbf{A}_j$, $j=1,2,3$ are symmetric. Here we allow $\mathbf{A}_0$ to have zero eigenvalues so that the form \eqref{eq:me} can also cover parabolic equations such as the R13 equations. For problems with unbounded domains or periodic boundary conditions, one can show that
\begin{displaymath}
\frac{\mathrm{d}}{\mathrm{d}t} \int \bu^T \mathbf{A}_0 \bu \,\mathrm{d}\bx = \int \bu^T \mathbf{L} \bu \,\mathrm{d}\bx \leqslant 0,
\end{displaymath}
indicating the $L^2$ stability. For problems on bounded domains, additional conditions on the boundary conditions are suggested in \cite{Sarna2017,Torrilhon2017,Bunger2023} to preserve the $L^2$ stability. Assume that $\boldsymbol{n} = (n_1, n_2, n_3)^T$ is the outer unit normal vector on the boundary point. The conditions are based on the following structures of the matrices $\mathbf{A}_j$:
\begin{equation} \label{eq:PAP}
\mathbf{P} (n_1 \mathbf{A}_1 + n_2 \mathbf{A}_2 + n_3 \mathbf{A}_3) \mathbf{P}^T = \begin{pmatrix}
0 & \mathbf{A}_{\mathrm{oe}} \\ \mathbf{A}_{\mathrm{eo}} & 0
\end{pmatrix},
\end{equation}
where $\mathbf{P}$ is an orthogonal matrix converting the moments $\bu$ to another set of moments, each of which is either odd or even in the normal direction, and we choose the permutation such that
\begin{displaymath}
\mathbf{P} \bu = \begin{pmatrix} \bu_{\mathrm{odd}} \\ \bu_{\mathrm{even}} \end{pmatrix}.
\end{displaymath}
The moments in $\bu_{\mathrm{odd}}$ include quantities that changes sign when the frame of reference changes by flipping the normal vector $\boldsymbol{n}$, and the moments in $\bu_{\mathrm{even}}$ remain unchanged under this transformation. By the symmetry of $\mathbf{A}_j$, we have $\mathbf{A}_{\mathrm{oe}} = \mathbf{A}_{\mathrm{eo}}^T$. The structure \eqref{eq:PAP} comes from the fact that the normal flux of an odd moment is an even moment, and the normal flux of an even moment is an odd moment. With this structure, if $\mathbf{A}_{\mathrm{oe}}$ has full row rank, the $L^2$ stable boundary conditions have the form
\begin{equation} \label{eq:stableBC}
\bu_{\mathrm{odd}} = \mathbf{Q} \mathbf{A}_{\mathrm{oe}} (\bu_{\mathrm{even}} - \boldsymbol{g}_{\mathrm{ext}}),
\end{equation}
where $\mathbf{Q}$ is a positive semidefinite matrix, $\boldsymbol{g}_{\mathrm{ext}}$ refers to the external source coming from the boundary of the domain.
Such boundary conditions are known as Onsager boundary conditions.
The particular form of boundary conditions that has odd moments on the left-hand side originates from Grad's work \cite{Grad1949}, which ensures the continuity of boundary conditions with respect to the accommodation coefficient. However, for a specific moment system, choosing $\bu_{\mathrm{odd}}$ to be all the odd moments may result in a rank-deficient $\mathbf{A}_{\mathrm{oe}}$, so that \eqref{eq:stableBC} will provide too many boundary conditions. In this work, we will encounter such a situation during our derivation, requiring us to adjust the selections of $\bu_{\mathrm{odd}}$ and $\bu_{\mathrm{even}}$ to restore the surjective property of $\mathbf{A}_{\mathrm{oe}}$.
The structure of Onsager boundary conditions is useful not only for time-dependent problems.
For steady-state problems (time derivative removed in \eqref{eq:me}), Onsager boundary conditions can provide a symmetric weak form, which helps develop the theory of well-posedness and the finite element methods \cite{Theisen2021}.

For the regularized 13-moment equations for Maxwell molecules, the boundary conditions with this particular structure have been devised in \cite{Torrilhon2017, Theisen2021}. However, the technique to derive the Onsager boundary conditions cannot be directly generalized to the R13 equations for more general molecules. In this work, we will reconsider the derivation of the linearized steady-state R13 equations for general collision models, from which we will show clearly how the structure \eqref{eq:PAP} is built into the derivation of moment equations, and thus the boundary conditions in the form \eqref{eq:stableBC} can be naturally obtained. 

In the following section, we will review the moment method for the linear Boltzmann equation and the asymptotic expansions of the moments when the Knudsen number is small. Our main results are presented in Section \ref{sec:main results}, where the explicit forms of the linear R13 equations and the Onsager boundary conditions are provided. The derivation of the R13 equations and boundary conditions are respectively given in Section \ref{sec:R13} and \ref{sec:bc}. In Section \ref{sec:results 1d}, we verify the accuracy of our model by one-dimensional channel problems. A brief conclusion is given in Section \ref{sec:conclusion}.

\section{Review of the moment equations and asymptotic properties of moments}
We consider the steady-state linear Boltzmann equation
\begin{equation}\label{boltzmann eq}
    \xi_j \frac{\partial f}{\partial x_j} = \frac{1}{\kn} \mathcal{L} [f]
\end{equation}
where $f(\bx,\bxi)$ is the distribution function, $\bx = (x_1,x_2,x_3)$ denotes the position and $\bxi = (\xi_1,\xi_2,\xi_3)$ stands for the velocity of the gas molecules. We apply Einstein's summation convention throughout this work, meaning that when an index appears twice in the same term, the expression represents the sum of this term with this index running from 1 to 3. For example, in \eqref{boltzmann eq}, $$\xi_j  \frac{\partial f}{\partial x_j} =\xi_1  \frac{\partial f}{\partial x_1} + \xi_2  \frac{\partial f}{\partial x_2} + \xi_3  \frac{\partial f}{\partial x_3}.$$ For indices whose ranges are not from 1 to 3, the summation symbol will be written explicitly. On the right-hand side of \eqref{boltzmann eq}, the constant $\kn$ is the Knudsen number characterizing how rarefied the gas is, and $\mathcal{L} [f]$ is the linearized Boltzmann collision operator. Below we will introduce the general moment equations for the linearized Boltzmann equation and the asymptotic expansion of the moments in the case of a small Knudsen number.

\subsection{Series expansion and moment equations}
Following \cite{Torrilhon2018}, we expand the distribution function into an infinite series:
\begin{equation}\label{f series}
    f(\bx,\bxi)= \sum_{l=0}^{+\infty}\sum_{m=0}^{+\infty} \frac{(2l+1)!!}{l!}w_{i_1\cdots i_l}^m(\bx) \psi^{m}_{i_1 \cdots i_l}(\bxi).
\end{equation}
Here, the basis functions $\psi_{i_1\cdots i_l}^m$ are defined by
\begin{equation}\label{psi}
    \psi^{m}_{i_1,\cdots,i_l}(\bxi)=\frac{1}{(2\pi)^{3/2}}\bar{L}_m^{(l+1/2)}\left( \frac{|\bxi|^2}{2} \right)\xi_{\langle i_1} \cdots \xi_{i_l \rangle} \exp \left( -\frac{|\bxi|^2}{2} \right),
\end{equation}
where $\xi_{\langle i_1} \cdots \xi_{i_l \rangle}$ is the trace-free part of the tensor $\xi_{ i_1} \cdots \xi_{i_l}$ (see \cite[Appendix A]{Struchtrup2005macroscopic}), and $\bar{L}_n^{(l+1/2)}$ is the normalized Laguerre polynomial
\begin{equation}\label{laguerre}
\bar{L}_n^{(l+1/2)}(x) = \sqrt{\frac{\sqrt{\pi}}{2^{l+1} n! \Gamma(n+l+3/2)}} x^{-(l+1/2)} \left( \frac{\mathrm{d}}{\mathrm{d}x} - 1 \right)^n x^{n+l+1/2}.
\end{equation}
In literature, people usually define the Maxwellian
    \begin{displaymath}
     f_M(\bxi) = \frac{1}{(2\pi)^{3/2}} \exp \left( -\frac{|\bxi|^2}{2} \right),
    \end{displaymath}
so that $\psi^{m}_{i_1,\cdots,i_l}(\bxi)=\bar{L}_m^{(l+1/2)}(|\bxi|^2/2)\,\xi_{\langle i_1} \cdots \xi_{i_l \rangle} f_M(\bxi)$. Compared with the classical series expansion by Grad \cite{Grad1949}, the Maxwellian $f_M$ is a global equilibrium state due to our linearized setting. This expansion requires us to assume that the distribution function is defined in the following Hilbert space:
\begin{displaymath}
L^2(\mathbb{R}^3, [f_M(\bxi)]^{-1} \,\mathrm{d}\bxi) := \left\{ f \bigg\vert \int_{\mathbb{R}^3} \frac{[f(\bxi)]^2}{f_M(\bxi)} \,\mathrm{d}\bxi < +\infty\right\},
\end{displaymath}
so that all the moments of the distribution function can be properly defined. The inner product of this Hilbert space is
\begin{displaymath}
   \langle g_1,g_2\rangle := \int_{\mathbb{R}^3} \frac{g_1(\bxi)g_2(\bxi)}{f_M(\bxi)} \dd \bxi,
\end{displaymath}
so that we can express the coefficients $w_{i_1\cdots i_l}^m$ as the moments of the distribution function:
\begin{equation}\label{moment}
w_{i_1\cdots i_l}^n(\bx) = \langle f, \psi^{m}_{i_1\cdots i_l} \rangle.
\end{equation}
In particular, we would like to highlight the relationship between these coefficients and the quantities in Grad's 13-moment equations:
\begin{equation} \label{physics quantities}
\rho = w^0, \quad v_i = \sqrt{3} w_i^0, \quad \theta = -\sqrt{\frac{2}{3}} w^1, \quad \sigma_{ij} = \sqrt{15} w_{ij}^0, \quad q_i = -\sqrt{\frac{15}{2}} w_i^1,
\end{equation}
where $\rho, v_i, \theta, \sigma_{ij}$ and $q_i$ denote the density, velocity, temperature, stress tensor and heat flux, respectively.

Due to the rotational invariance of the collision, the linear operator $\mathcal{L}$ satisfies 
\begin{equation} \label{eq:rot_inv}
    \mathcal{L}[ \psi^{n}_{i_1,\cdots,i_l}] = \sum_{m=0}^{+\infty} a_{lmn} \psi^{m}_{i_1,\cdots,i_l},
\end{equation}
where the coefficient $a_{lmn}$ satisfies $a_{lmn} = a_{lnm}$ for all nonnegative integers $l,m$ and $n$, and
\begin{equation} \label{eq:almn}
a_{lmn} = \frac{\langle \psi_{i_1\cdots i_l}^m, \mathcal{L}\psi_{i_1\cdots i_l}^n \rangle}{\langle \psi_{i_1\cdots i_l}^m, \psi_{i_1\cdots i_l}^m \rangle}.
\end{equation}
Note that we have chosen the basis functions such that $\langle \psi_{i_1\cdots i_l}^n, \psi_{i_1\cdots i_l}^n \rangle$ depends only on $i_1, \cdots, i_l$, leading to the symmetry of $a_{lmn}$.
Due to the conservation of mass, momentum and energy, it holds that
\begin{equation} \label{eq:a_zero_value}
a_{00n} = a_{0n0} = a_{01n} = a_{0n1} = a_{10n} = a_{1n0} = 0.
\end{equation}
For inverse-power-law models, these coefficients have been calculated in \cite{Cai2015}, where it has also been pointed out that the collision operator $\mathcal{L}$ is usually an unbounded operator acting on a subset of $L^2(\mathbb{R}^3, [f_M(\bxi)]^{-1} \,\mathrm{d}\bxi)$.

The equations of the moments $w_{i_1\cdots i_l}^n$ can be obtained by taking the inner product of $\psi_{i_1\cdots i_l}^n$ and both sides of the Boltzmann equation \eqref{boltzmann eq}. The resulting moment equations are
\begin{equation} \label{eq:mnt_eq}
T_{i_1 \cdots i_l}^n =
   \frac{1}{\kn} \sum_{n'=0}^{+\infty} a_{lnn'} w_{i_1\cdots i_l}^{n'}
\end{equation}
where 
\begin{equation}\label{eq:T}
 \begin{split}
      T_{i_1\cdots i_l}^n:= &  \left( \sqrt{2(n+l)+3} \frac{\partial w_{i_1\cdots i_l j}^n}{\partial x_j} - \sqrt{2n} \frac{\partial w_{i_1\cdots i_l j}^{n-1}}{\partial x_j} \right)\\
& \quad +   \frac{l}{2l+1} \left( \sqrt{2(n+l)+1} \frac{\partial w_{\langle i_1\cdots i_{l-1}}^n}{\partial x_{i_l \rangle}} - \sqrt{2(n+1)} \frac{\partial w_{\langle i_1\cdots i_{l-1}}^{n+1}}{\partial x_{i_l\rangle}} \right).
 \end{split}
\end{equation}
The derivation of these equations will be briefly introduced in Section \ref{app:w mnt eq} of the supplementary material.

\subsection{Asymptotic expansion of moments}
\label{sec:asym}
Assuming that $\kn$ is a small parameter, we consider the asymptotic expansions of $w_{i_1\cdots i_l}^n$:
\begin{equation} \label{eq:w_asymp}
w_{i_1\cdots i_l}^n =  w_{i_1\cdots i_l}^{n|0} + \kn w_{i_1\cdots i_l}^{n|1} + \kn^2 w_{i_1\cdots i_l}^{n|2} + \kn^3 w_{i_1\cdots i_l}^{n|3} + \cdots.
\end{equation}
The classical Chapman-Enskog expansion can be applied to express each term using the density, momentum, energy and their derivatives. Here, instead of performing the Chapman-Enskog expansion, we would like to find the orders of magnitude of each moment and the relationship between the terms in each order.
To this aim, we introduce the coefficients $b_{lnn'}^{(n_0)}$ to denote the inverses of $a_{lmn}$, which satisfy
\begin{equation}\label{sum ab}
\sum_{n'=n_0}^{+\infty} a_{lnn'} b_{ln_1 n'}^{(n_0)} = \delta_{nn_1}.
\end{equation}
Note that the coefficients $b_{0nn'}^{(0)}$, $b_{0nn'}^{(1)}$ and $b_{1nn'}^{(0)}$ do not exist due to \eqref{eq:a_zero_value}.
By asymptotic analysis, we are able to identify the magnitude of each moment $w_{i_1\cdots i_l}^n$ and find the linear dependency between $w_{i_1\cdots i_l}^{n|k}$ for different $n$'s. This method is known as the order of magnitude approach \cite{Struchtrup2004}.
Below we list the moments by order up to $O(\kn^3)$ and some results of the linear relationship to be used later in this work. The derivation can be found in Section \ref{app:proof of lemma} in the supplementary material.
\begin{itemize}
    \item[\textbf{\emph{(O0)}}] $O(1)$ moments:  $w^0$, $w^1$, $w^0_i$.
    \item[\textbf{\emph{(O1)}}] $O(\kn)$ moments: $\{w^n_i\}_{n=1}^\infty$, $\{w^n_{ij}\}_{n=0}^\infty$. The leading order terms of these moments satisfy 
\begin{align}
    w_i^{n|1} & = \frac{b_{11n}^{(1)}}{b_{111}^{(1)}} w_i^{1|1} \text{~for~} n \geqslant 1, \label{eq:w_i^n1} \\
    w_{ij}^{n|1} & = \frac{b_{20n}^{(0)}}{b_{200}^{(0)}} w_{ij}^{0|1} \text{~for~} n \geqslant 0, \label{eq:w_ij^n1}
\end{align}    
and the second order terms satisfy  
\begin{align}
w_i^{n|2} & = \frac{\gamma^{(1),n}_{1}}{\gamma^{(1),2}_{1}} w_i^{2|2} \text{~for~} n \geqslant 2   \text{~with~} 
\gamma^{(1),n}_{1} = \sum_{n'=2}^{\infty}
  \frac{b_{1nn'}^{(2)} (\sqrt{2n'+5} b_{20n'}^{(0)} - \sqrt{2n'} b_{20,n'-1}^{(0)})}{b_{200}^{(0)}},
 \label{eq:w_i^n2} \\
 w_{ij}^{n|2} & = \frac{\gamma^{(1),n}_2}{\gamma^{(1),1}_{2}} w_{ij}^{1|2} \text{~for~} n \geqslant 1   \text{~with~} 
\gamma^{(1),n}_2 = \frac{2}{5} \sum_{n'=1}^{\infty}
  \frac{b_{2nn'}^{(1)} (\sqrt{2n'+5} b_{11n'}^{(1)} - \sqrt{2(n'+1)} b_{11,n'+1}^{(1)})}{b_{111}^{(1)}}.  \label{eq:w_ij^n2}
\end{align}

\item[\textbf{\emph{(O2)}}] $O(\kn^2)$ moments: $\{w^n\}_{n=2}^\infty$, $\{w^n_{ijk}\}_{n=0}^\infty$. Their leading order terms satisfy 
\begin{align}
    & w^{n|2} =    \frac{\gamma^{(2),n}_0}{\gamma^{(2),2}_0}w^{2|2} \text{~with~}   \gamma^{(2),n}_0 =  \sum_{n'=2}^{+\infty} \frac{b_{0nn'}^{(2)} (\sqrt{2n'+3} b_{11n'}^{(1)} - \sqrt{2n'} b_{11,n'-1}^{(1)})}{b_{111}^{(1)}},  \label{w^n} \\
    & w_{ijk}^{n|2} = \frac{\gamma^{(2),n}_3}{\gamma^{(2),0}_3} w_{ijk}^{0|2}
  \text{~with~}  \gamma^{(2),n}_3 =  \frac{3}{7} \sum_{n'=0}^{+\infty} \frac{b_{3nn'}^{(0)}}{b_{200}^{(0)}} \left(\sqrt{2n'+7} b_{20n'}^{(0)} - \sqrt{2(n'+1)} b_{20,n'+1}^{(0)}\right) .\label{wijk}
\end{align}
\item[\textbf{\emph{(O3)}}] $O(\kn^3)$ moments: $\{w^n_{ijkl}\}_{n=0}^\infty$.
\item[\textbf{\emph{(O4)}}] $o(\kn^3)$ moments: all other $w_{i_1\cdots i_l}^n$ which are not listed above.
\end{itemize}
These results show that only the conserved moments are zeroth-order moments, which agrees with the results from the Chapman-Enskog expansion. Although there are infinite first-order moments, the leading-order terms depend only on the stress tensor ($w_{ij}^0$) and the heat flux ($w_i^1$). The purpose of R13 equations is to formulate equations using only these representative moments up to the first order, and ``regularization terms'' are added to increase its order of accuracy to cover super-Burnett equations. This requires us to express all second-order terms $w_{i_1\cdots i_l}^{n|2}$ using the thirteen moments appearing in the equations, and such a procedure has been done in the literature \cite{Struchtrup2013, Cai2020regularized}. However, the approach therein does not clearly show how the stable boundary conditions should be derived. Although attempts have been made to study boundary value problems in \cite{Hu2020}, the boundary conditions do not have the structure \eqref{eq:stableBC} as required in \cite{Torrilhon2017}. In this paper, we will re-derive the regularized 13-moment equations from another point of view, and equip the model with reasonable boundary conditions with the desired structure. Before that, we will first present our final models in the next section for the readers who are not interested in the derivation.

\section{Linear R13 equations and Onsager boundary conditions}\label{sec:main results}
In this section, we present the steady-state linear regularized 13-moment equations for general gas molecules, and provide the Onsager boundary conditions satisfying the conditions in \cite{Torrilhon2017}. The equations and boundary conditions will be presented using the physical variables $\rho, v_i, \theta, \sigma_{ij}, q_i$, which are equivalent to the coefficients $w^0$, $w^1$, $w^0_i$, $w^1_i$, $w^0_{ij}$ according to \eqref{physics quantities}.

\subsection{Linearized R13 moment equations}
\label{sec:symR13}
In Section \ref{sec:R13}, We have derived the following equation system of the 13 moments including $\rho, v_i, \theta, \sigma_{ij}, q_i$:
\begin{itemize}
   \item Equations of mass conservation, energy conservation and momentum conservation: 
      \begin{align}
     \frac{\partial v_j}{\partial x_j} & =   0, \label{r13eq:density}\\
   \frac{\partial v_j}{\partial x_j} +  \frac{\partial q_j}{\partial x_j} & = 0,\label{r13eq:temperature}\\
    \frac{\partial \rho}{\partial x_i} +\frac{\partial \theta}{\partial x_i} +\frac{\partial \sigma_{ij}}{\partial x_j} & = 0 .\label{r13eq:velocity}
    \end{align}
    \item Equations of heat flux and stress tensor:
    \begin{align}
      \frac{\partial \theta}{\partial x_i}+ \beta_4 \frac{\partial \sigma_{ij}}{\partial x_j} + \frac{2}{15c^{(1),1}_1} \kn\frac{\partial}{\partial x_j}\left(\beta_0 \frac{\partial q_j}{\partial x_i} + \beta_2 \frac{\partial q_{\langle i}}{\partial x_{  j\rangle}} \right) & =\frac{2}{15(c^{(1),1}_1)^2}\frac{1}{\kn}\mathscr{L}^{(11)}_1 q_i,\label{r13eq:heat flux}\\
        \frac{\partial v_{\langle i}}{\partial x_{j \rangle}} + \beta_4 \frac{\partial q_{\langle i}}{\partial x_{j \rangle}} +\frac{1}{15c^{(1),0}_2}\kn \frac{\partial}{\partial x_k} \left( \beta_1  \frac{\partial \sigma_{k\langle i}}{\partial x_{j\rangle}} + \beta_3 \frac{\partial \sigma_{\langle ij}}{\partial x_{k \rangle}}   \right) & = \frac{1}{15(c^{(1),0}_2)^2}\frac{1}{\kn} \mathscr{L}^{(11)}_2 \sigma_{ij}\label{r13eq:stress tensor}
    \end{align}
\end{itemize}
where 
\begin{gather*}
   \beta_0 =  \frac{(A_{46})^2}{c^{(1),1}_1 \mathscr{L}^{(22)}_0}, \qquad  \beta_1  =  \frac{(c^{(1),1}_1 A_{57} - c^{(2),1}_1A_{45})^2}{c^{(1),1}_1 c^{(1),0}_2(c^{(1),1}_1\mathscr{L}^{(22)}_1 - c^{(2),1}_1\mathscr{L}^{(12)}_1) },\\
   \beta_2 =  \frac{(c^{(1),0}_2A_{48} - c^{(2),0}_2A_{45})^2}{c^{(1),1}_1 c^{(1),0}_2(c^{(1),0}_2\mathscr{L}^{(22)}_2 - c^{(2),0}_2 \mathscr{L}^{(12)}_2) }, \quad \beta_3 =  \frac{(A_{59})^2}{c^{(1),0}_2 \mathscr{L}^{(22)}_3}, \quad \beta_4 = - \frac{\sqrt{2} A_{45}}{15c^{(1),1}_1c^{(1),0}_2}.
\end{gather*}
 The expression of the coefficients  $c^{(p),n}_l$ and $A_{ij}$ are given in Section \ref{app:coefficients} of the supplementary material and $\mathscr{L}^{(mn)}_l$ is formulated as \eqref{Lmnl}. We remark that for Maxwell molecules, $\beta_1$ above takes the form $\frac{0}{0}$ and is thus not well-defined. In this case, $\beta_1$ is set to be $0$. One can easily observe the symmetric structure of the system above, where the complicated second-order derivatives in the last two equations are on the diagonal. We have shown that such moment system has the super-Burnett order. 

\subsection{Onsager boundary conditions} \label{sec:obc}
The R13 equations are equipped with an Onsager-type boundary conditions which read 
\begin{align}
    &v_n =  0, \label{bc phys 1}\\
    &q_n =  \frac{2\chi}{2-\chi}\left[\lambda_{11}(\theta - \theta^W) + \lambda_{12}\sigma_{nn} + \kn \lambda_{13} \frac{\partial q_j}{ \partial x_j} + \kn \lambda_{14} \frac{\partial q_{\langle n}}{\partial x_{n \rangle}} \right],\label{bc phys 2}\\
    &\sigma_{t_i n} =   \frac{2\chi}{2-\chi}\left[ \lambda_{21} (v_{t_i} - v_{t_i}^W) + \lambda_{22}q_{t_i} + \kn \lambda_{23} \frac{\partial \sigma_{t_i j}}{\partial x_j} + \kn\lambda_{24} \frac{\partial \sigma_{\langle nn}}{\partial x_{t_i  \rangle}} \right], \quad i =1,2,\label{bc phys 3}\\
    &\kn \frac{\partial q_{\langle t_i}}{\partial x_{n \rangle}} =  \frac{2\chi}{2-\chi}\left[ \lambda_{31} (v_{t_i} - v_{t_i}^W) + \lambda_{32}q_{t_i} + \kn \lambda_{33} \frac{\partial \sigma_{t_i j}}{\partial x_j} + \kn\lambda_{34} \frac{\partial \sigma_{\langle nn}}{\partial x_{t_i \rangle}  } \right], \quad i =1,2,\label{bc phys 4}\\
    &\kn\left( \frac{\partial \sigma_{\langle nn}}{\partial x_{n \rangle}} + \lambda_{45} \frac{\partial \sigma_{ n t_j}}{\partial x_{t_j}}\right)  
    =   \frac{2\chi}{2-\chi}\left[ \lambda_{41}(\theta - \theta^W) + \lambda_{42}\sigma_{nn} + \kn \lambda_{43} \frac{\partial q_j}{ \partial x_j} + \kn \lambda_{44} \frac{\partial q_{\langle n}}{\partial x_{n \rangle}} \right], \label{bc phys 5}\\
     &\kn\left(\frac{\partial \sigma_{\langle t_it_i}}{\partial x_{n \rangle}  } + \frac{1}{2} \frac{\partial \sigma_{\langle nn}}{\partial x_{n \rangle}  }\right) =   \frac{2\chi}{2-\chi}\left(\lambda_{51} \sigma_{t_it_i} + \lambda_{52} \sigma_{nn}\right),\quad i =1,2,\label{bc phys 6}\\
     &\kn \frac{\partial \sigma_{\langle t_1t_2}}{\partial x_{n \rangle}  } =  \frac{2\chi}{2-\chi} \lambda_{61}  \sigma_{t_1t_2}. \label{bc phys 7}
\end{align}
The expressions of coefficient $\lambda_{ij}$ can be found in Section \ref{app:phys bc} of the supplementary material.

\section{Derivation of R13 equations}
\label{sec:R13}
We will now present the derivation of the R13 equations given in Section \ref{sec:symR13}. Our derivation will use a method different from previous papers \cite{Struchtrup2013, Cai2020regularized}, so that it is clear why the structure \eqref{eq:PAP} exist in the final system.
Since the derivation of moment equations often involves complicated notations and calculations, in order to better explain the main idea of our derivation, we will first write the equations using operators on function spaces instead of the moments (Section \ref{sec:abstract}), and then explain how to convert the abstract form to the explicit moment equations (Section \ref{sec:R13_explicit}).


\subsection{Reformulation of the distribution function}
\label{sec:expansion of distribution function}
In Section \ref{sec:asym}, we have seen that when using $\psi_{i_1 \cdots i_l}^m f_M$ as basis functions, there are infinite $O(\kn^d)$ coefficients for any $d \geq 1$, which is inconvenient for the derivation of moment equations. In this section, we will look for new basis functions such that in the expansions, only finite coefficients have the order $O(\kn^k)$ for any $k$. In other words, we seek the following orthogonal decomposition of the function space:
\begin{displaymath}
L^2(\mathbb{R}^3, [f_M(\bxi)]^{-1} \,\mathrm{d}\bxi) = \mathbb{V}^{(0)} \oplus \mathbb{V}^{(1)} \oplus \mathbb{V}^{(2)} \oplus \cdots
\end{displaymath}
such that each $\mathbb{V}^{(k)}$ is a finite dimensional space, and it holds that
\begin{equation} \label{eq:Pk}
\mathcal{P}^{(k)} f \sim O(\kn^k), \qquad
  \forall k \in \mathbb{N},
\end{equation}
where $\mathcal{P}^{(k)}$ is the projection operator from $L^2(\mathbb{R}^3, [f_M(\bxi)]^{-1} \,\mathrm{d}\bxi)$ onto $\mathbb{V}^{(k)}$. Such a decomposition allows us to consider the projection of $f$ onto a finite dimensional space when we want to achieve a reduced model up to a given order of accuracy.

For the purpose of deriving R13 equations, we just need to use the function spaces from $\mathbb{V}^{(0)}$ to $\mathbb{V}^{(3)}$. They will be discussed in the following subsections.

\subsubsection{The zeroth-order function space}
The function space $\mathbb{V}^{(0)}$ can be easily observed from the Chapman-Enskog expansion. It should be spanned by the basis functions corresponding to the conserved moments, which means
\begin{equation} \label{eq:V0}
\mathbb{V}^{(0)} = \operatorname{span} \{ \psi^0, \psi^1, \psi_i^0 \mid i = 1,2,3\}.
\end{equation}
It is clear that
\begin{displaymath}
\dim \mathbb{V}^{(0)} = 5.
\end{displaymath}

\subsubsection{The first-order function space}
Our idea to find the first-order function space is to first construct the orthogonal complement of $\mathbb{V}^{(0)} \oplus \mathbb{V}^{(1)}$, and then find $\mathbb{V}^{(1)}$ by orthogonality. This orthogonal complement will include the part of $f$ that has order higher than or equal to $O(\kn^2)$. According to \textbf{(O2)}-\textbf{(O4)} in Section \ref{sec:asym}, we know that all the following functions should be members of $(\mathbb{V}^{(0)} \oplus \mathbb{V}^{(1)})^{\perp}$:
\begin{equation} \label{eq:V01 complement 1}
\psi^n \text{ for } n \geq 2, \qquad
\psi_{i_1 \cdots i_l}^n \text{ for all } l \geq 3 \text{ and } n \geq 0.
\end{equation}
In addition, the relation \eqref{eq:w_i^n1} yields
\begin{equation*}
  w_i^{n} - \frac{b_{11n}^{(1)}}{b_{111}^{(1)}} w_i^{1} = \left\langle f,  \psi_i^{n} - \frac{b_{11n}^{(1)}}{b_{111}^{(1)}} \psi_i^{1}  \right \rangle \sim O(\kn^2) \text{~for~} n \geq 2,
\end{equation*}
meaning that
\begin{equation} \label{eq:V01 complement 2}
\psi_i^{n} - \frac{b_{11n}^{(1)}}{b_{111}^{(1)}} \psi_i^{1} \in (\mathbb{V}^{(0)} \oplus \mathbb{V}^{(1)})^{\perp}, \qquad \forall i = 1,2,3, \quad n \geq 2.
\end{equation}
Similarly, we can use \eqref{eq:w_ij^n1} to obtain
\begin{equation} \label{eq:V01 complement 3}
\psi_{ij}^{n} - \frac{b_{20n}^{(0)}}{b_{200}^{(0)}} \psi_{ij}^{0} \in (\mathbb{V}^{(0)} \oplus \mathbb{V}^{(1)})^{\perp}, \qquad \forall i = 1,2,3, \quad n \geq 1.
\end{equation}
We can now conclude that $(\mathbb{V}^{(0)} \oplus \mathbb{V}^{(1)})^{\perp}$ is the subspace of $L^2(\mathbb{R}^3, [f_M(\bxi)]^{-1}\,\mathrm{d}\bxi)$ spanned by all the functions in \eqref{eq:V01 complement 1}\eqref{eq:V01 complement 2}\eqref{eq:V01 complement 3}. Consequently, we can find $\mathbb{V}^{(1)}$ in the form of
\begin{equation} \label{eq:V1}
\mathbb{V}^{(1)} = \operatorname{span} \left\{ \phi_i^{(1)}, \phi_{ij}^{(1)} \,\Big|\, i,j = 1,2,3 \right\},
\end{equation}
where
\begin{displaymath}
\phi_i^{(1)} = \sum_{n=1}^{+\infty} c_1^{(1),n} \psi_i^n, \qquad
\phi_{ij}^{(1)} = \sum_{n=0}^{+\infty} c_2^{(1),n} \psi_{ij}^n.
\end{displaymath}
The coefficients $c_1^{(1),n}$ and $c_2^{(1),n}$ are determined by solving 
\begin{equation} \label{determine coefficient}
    \left\langle \psi_i^{n} - \frac{b_{11n}^{(1)}}{b_{111}^{(1)}} \psi_i^{1},\phi_i^{(1)} \right\rangle  =0 \quad  \text{~and~} \quad 
    \left\langle \psi_{ij}^{n} - \frac{b_{20n}^{(0)}}{b_{200}^{(0)}} \psi_{ij}^{0} ,\phi_{ij}^{(1)}  \right\rangle   =0 
\end{equation}
respectively, and are then scaled such that
\begin{equation}\label{scaling}
  \sum_{n=1}^{+\infty} \left|c_1^{(1),n}\right|^2 = \sum_{n=0}^{+\infty} \left|c_2^{(1),n}\right|^2 = 1.
\end{equation}
Since $\psi_{ij}^n$ is a trace-free tensor, we have
\begin{displaymath}
\dim \mathbb{V}^{(1)} = 8.
\end{displaymath}

\subsubsection{The second-order function space}
The second-order function spaces can be found in a similar way. Using \textbf{(O2)}--\textbf{(O4)}, we notice that the following functions are members of $(\mathbb{V}^{(0)} \oplus\mathbb{V}^{(1)} \oplus\mathbb{V}^{(2)})^{\perp}$:
\begin{gather*}
\psi^n - \frac{\gamma_n}{\gamma_2} \psi^2 \text{ for } n \geq 3, \qquad
\left( \psi_{i}^{n} - \frac{b^{(1)}_{11n}}{b^{(1)}_{111}} \psi^1_i \right) - \frac{\gamma^{(1),n}_{1}}{\gamma^{(1),2}_{1}}  \left( \psi_{i}^{2} - \frac{b^{(1)}_{112}}{b^{(1)}_{111}} \psi^1_i \right) \text{ for } n \geq 3, \\
\left( \psi_{ij}^{n} - \frac{b^{(0)}_{20n}}{b^{(0)}_{200}} \psi^0_{ij} \right) - \frac{\gamma^{(1),n}_{2}}{\gamma^{(1),1}_{2}}  \left( \psi_{ij}^{1} - \frac{b^{(0)}_{201}}{b^{(0)}_{200}} \psi^0_{ij} \right) \text{ for } n \geq 2, \qquad
\psi_{ijk}^n - \frac{\gamma^{(2),n}_3}{\gamma^{(2),0}_3} \psi_{ijk}^0 \text{ for } n \geq 1, \\
\psi_{i_1\cdots i_l}^n \text{ for all } l \geq 4 \text{ and } n \geq 0.
\end{gather*}
Let $(\mathbb{V}^{(0)} \oplus\mathbb{V}^{(1)} \oplus\mathbb{V}^{(2)})^{\perp}$ be the linear span of all these functions. We can find $\mathbb{V}^{(2)}$ in the following form:
\begin{equation} \label{eq:V2}
\mathbb{V}^{(2)} = \operatorname{span} \left\{
  \phi^{(2)}, \phi_i^{(2)}, \phi_{ij}^{(2)}, \phi_{ijk}^{(2)} \,\Big\vert\, i,j,k = 1,2,3
\right\},
\end{equation}
where
\begin{gather*}
\phi^{(2)} = \sum_{n=2}^{+\infty} c_0^{(2),n} \psi^n, \qquad
\phi_i^{(2)} = \sum_{n=2}^{+\infty} c_1^{(2),n} \left(\psi_{i}^{n} - \frac{b^{(1)}_{11n}}{b^{(1)}_{111}} \psi^1_i\right), \\
\phi_{ij}^{(2)} = \sum_{n=1}^{+\infty} c_2^{(2),n} \left(\psi_{ij}^{n} - \frac{b^{(0)}_{20n}}{b^{(0)}_{200}} \psi^0_{ij}\right), \qquad
\phi_{ijk}^{(2)} = \sum_{n=0}^{+\infty} c_3^{(2),n} \psi_{ijk}^n.
\end{gather*}
For simplicity, we let
\begin{displaymath}
c_1^{(2),1} = -\sum_{n=2}^{+\infty} c_1^{(2),n} \frac{b_{11n}^{(1)}}{b_{111}^{(1)}}, \qquad
c_2^{(2),0} = -\sum_{n=1}^{+\infty} c_2^{(2),n} \frac{b_{20n}^{(0)}}{b_{200}^{(0)}},
\end{displaymath}
so that
\begin{equation}
\phi_i^{(2)} = \sum_{n=1}^{+\infty} c_1^{(2),n} \psi_{i}^{n}, \qquad
\phi_{ij}^{(2)} = \sum_{n=0}^{+\infty} c_2^{(2),n} \psi_{ij}^{n}.
\end{equation}
These coefficients are again determined by the orthogonality similar to \eqref{determine coefficient} and then scaled similarly as \eqref{scaling}. The expressions of $c^{(1),n}_l$ and $c^{(2),n}_l$ can be found in Section \ref{app:cdl} of the supplementary material. Also, it is not difficult to find that
\begin{displaymath}
\dim \mathbb{V}^{(2)} = 16.
\end{displaymath}

\subsubsection{The third-order function space}
The space $\mathbb{V}^{(3)}$ can again be derived using the same strategy. The result will have the form
\begin{displaymath}
\mathbb{V}^{(3)} = \operatorname{span} \left\{
  \phi^{(3)}, \phi_i^{(3)}, \phi_{ij}^{(3)}, \phi_{ijk}^{(3)}, \phi_{ijkl}^{(3)} \,\Big\vert\, i,j,k,l = 1,2,3
\right\},
\end{displaymath}
and
\begin{displaymath}
\dim \mathbb{V}^{(3)} = 25.
\end{displaymath}
For our purpose, the precise forms of these functions will not be used.
\begin{remark}
For Maxwell molecules, due to the special structure $b_{lnn'}^{n_0} = 0$ for all $n \neq n'$, the function space $\mathbb{V}^{(2)}$ will be slightly different. Instead of \eqref{eq:V2}, we will have
\begin{displaymath}
\mathbb{V}^{(2)} = \operatorname{span} \left\{
  \phi^{(2)}, \phi_{ij}^{(2)}, \phi_{ijk}^{(2)} \,\Big\vert\, i,j,k = 1,2,3
\right\},
\end{displaymath}
and thus $\dim \mathbb{V}^{(2)} = 13$. In fact, we also have $\phi^{(2)} = \psi^2$, $\phi_{ij}^{(2)} = \psi_{ij}^1$ and $\phi_{ijk}^{(2)} = \psi_{ijk}$, which can significantly simplify the derivation. In this paper, we will mainly focus on the R13 equations for general molecules. One can find the equations for Maxwell molecules in many references such as \cite{Struchtrup2005macroscopic, Theisen2021}.
\end{remark}

\subsection{The abstract form of regularized 13-moment equations} \label{sec:abstract}
In order to derive R13 moment equations that are accurate up to the super-Burnett order, it suffices to work only in the function space
\begin{displaymath}
\mathbb{V} := \mathbb{V}^{(0)} \oplus \mathbb{V}^{(1)} \oplus \mathbb{V}^{(2)} \oplus \mathbb{V}^{(3)}.
\end{displaymath}
In other words, we consider the following approximation of the Boltzmann equation \eqref{boltzmann eq}:
\begin{displaymath}
\frac{\partial}{\partial x_j} (\mathcal{P}_{\mathbb{V}} \xi_j \mathcal{P}_{\mathbb{V}} f) = \frac{1}{\kn} \mathcal{P}_{\mathbb{V}} \mathcal{L}\mathcal{P}_{\mathbb{V}} f.
\end{displaymath}
Here $\mathcal{P}_{\mathbb{V}}$ denotes the projection operator from $L^2(\mathbb{R}^3, [f_M(\bxi)]^{-1} \,\mathrm{d}\bxi)$ onto $\mathbb{V}$. One can easily verify that both $\mathcal{P}_{\mathbb{V}} \xi_j \mathcal{P}_{\mathbb{V}}$ and $\mathcal{P}_{\mathbb{V}} \mathcal{L} \mathcal{P}_{\mathbb{V}}$ are self-adjoint operators.

To separate different orders in the distribution function $f$, we further write $\mathcal{P}_{\mathbb{V}} f$ as
\begin{displaymath}
\mathcal{P}_{\mathbb{V}} f = f^{(0)} + f^{(1)} + f^{(2)} + f^{(3)},
\end{displaymath}
where $f^{(k)} = \mathcal{P}^{(k)} f$ (see \eqref{eq:Pk}). Thus, the projected Boltzmann equation can be written in the following form:
\begin{equation} \label{eq:pBoltz}
\sum_{j=1}^3 \frac{\partial}{\partial x_j} \left(
\begin{bmatrix}
 \mathcal{A}_j^{(00)} &  \mathcal{A}_j^{(01)} &  \mathcal{A}_j^{(02)} &  \mathcal{A}_j^{(03)} \\
  \mathcal{A}_j^{(10)} &  \mathcal{A}_j^{(11)} &  \mathcal{A}_j^{(12)}  &  \mathcal{A}_j^{(13)}  \\
    \mathcal{A}_j^{(20)} &  \mathcal{A}_j^{(21)} &  \mathcal{A}_j^{(22)}  &  \mathcal{A}_j^{(23)} \\
     \mathcal{A}_j^{(30)} &   \mathcal{A}_j^{(31)} &  \mathcal{A}_j^{(32)}  &  \mathcal{A}_j^{(33)}
\end{bmatrix}
\begin{bmatrix}
  f^{(0)} \\ f^{(1)} \\ f^{(2)} \\ f^{(3)}
\end{bmatrix}
\right) = \frac{1}{\kn}
\begin{bmatrix}
  0 & 0 & 0 & 0 \\
  0 &  \mathcal{L}^{(11)} &  \mathcal{L}^{(12)}  &  \mathcal{L}^{(13)}  \\
    0 &  \mathcal{L}^{(21)} &  \mathcal{L}^{(22)}  &  \mathcal{L}^{(23)} \\
     0 &   \mathcal{L}^{(31)} &  \mathcal{L}^{(32)}  &  \mathcal{L}^{(33)}
\end{bmatrix}
\begin{bmatrix}
  f^{(0)} \\ f^{(1)} \\ f^{(2)} \\ f^{(3)}
\end{bmatrix},
\end{equation}
where
\begin{displaymath}
\mathcal{A}_j^{(kl)} = \mathcal{P}^{(k)} \xi_j \mathcal{P}^{(l)}, \qquad
\mathcal{L}^{(kl)} = \mathcal{P}^{(k)} \mathcal{L} \mathcal{P}^{(l)}.
\end{displaymath}
Note that the zero operators on the right-hand side of \eqref{eq:pBoltz} comes from the conservation laws, and on the left-hand side, we have written the sum over $j$ explicitly for clearness. The self-adjointness of $\mathcal{P}_{\mathbb{V}} \xi_j \mathcal{P}_{\mathbb{V}}$ and $\mathcal{P}_{\mathbb{V}} \mathcal{L}\mathcal{P}_{\mathbb{V}}$ implies that $\mathcal{A}_j^{(kl)} = [\mathcal{A}_j^{(lk)}]^{\dagger}$ and $\mathcal{L}_j^{(kl)} = [\mathcal{L}_j^{(lk)}]^{\dagger}$. For simplicity, below we will use the definition
\begin{displaymath}
\mathcal{A}^{(kl)} = \sum_{j=1}^3 \frac{\partial}{\partial x_j} \mathcal{A}_j^{(kl)},
\end{displaymath}
so that the equations \eqref{eq:pBoltz} become 
\begin{equation}\label{eq:pBoltzSimplified}
\boldsymbol{\mathcal{A}} \boldsymbol{f} = \frac{1}{\kn} \boldsymbol{\mathcal{L}}\boldsymbol{f},
\end{equation}
where
\begin{equation} 
\boldsymbol{\mathcal{A}} =
\begin{bmatrix}
 \mathcal{A}^{(00)} &  \mathcal{A}^{(01)} &  \mathcal{A}^{(02)} &  \mathcal{A}^{(03)} \\
  \mathcal{A}^{(10)} &  \mathcal{A}^{(11)} &  \mathcal{A}^{(12)}  &  \mathcal{A}^{(13)}  \\
    \mathcal{A}^{(20)} &  \mathcal{A}^{(21)} &  \mathcal{A}^{(22)}  &  \mathcal{A}^{(23)} \\
     \mathcal{A}^{(30)} &   \mathcal{A}^{(31)} &  \mathcal{A}^{(32)}  &  \mathcal{A}^{(33)}
\end{bmatrix}, \quad
\boldsymbol{\mathcal{L}} = 
\begin{bmatrix}
  0 & 0 & 0 & 0 \\
  0 &  \mathcal{L}^{(11)} &  \mathcal{L}^{(12)}  &  \mathcal{L}^{(13)}  \\
    0 &  \mathcal{L}^{(21)} &  \mathcal{L}^{(22)}  &  \mathcal{L}^{(23)} \\
     0 &   \mathcal{L}^{(31)} &  \mathcal{L}^{(32)}  &  \mathcal{L}^{(33)}
\end{bmatrix}, \quad
\boldsymbol{f} = \begin{bmatrix}
  f^{(0)} \\ f^{(1)} \\ f^{(2)} \\ f^{(3)}
\end{bmatrix}.
\end{equation}

In the rest part of this section, we will derive a simplified version of the equations \eqref{eq:pBoltzSimplified} in the following form:
\begin{equation} \label{eq:R13}
\begin{bmatrix}
 \mathcal{A}^{(00)} &  \mathcal{A}^{(01)} &  \mathcal{A}^{(02)} \\
  \mathcal{A}^{(10)} &  \mathcal{A}^{(11)} &  \mathcal{A}^{(12)} \\
    \mathcal{A}^{(20)} &  \mathcal{A}^{(21)} &  \bar{\mathcal{A}}^{(22)}
\end{bmatrix}
\begin{bmatrix}
  f^{(0)} \\ f^{(1)} \\ f^{(2)}
\end{bmatrix}
= \frac{1}{\kn} \begin{bmatrix}
  0 & 0 & 0 \\
  0 &  \mathcal{L}^{(11)} &  \mathcal{L}^{(12)}\\
    0 &  \mathcal{L}^{(21)} &  \mathcal{L}^{(22)} \\
\end{bmatrix}
\begin{bmatrix}
  f^{(0)} \\ f^{(1)} \\ f^{(2)}
\end{bmatrix}.
\end{equation}
We require that the super-Burnett equations can also be derived from these equations. The abstract system \eqref{eq:R13} will be further formulated as R13 equations in Section \ref{sec:R13_explicit}. To specify the operator $\bar{\mathcal{A}}^{(22)}$ and clarify why the final equations hold the form \eqref{eq:R13}, we need three steps given in the three subsections below.

\subsubsection{Step 1: Diagonalization of the right-hand side}
The first step of our derivation is to reformulate \eqref{eq:pBoltzSimplified} into an equivalent form where the right-hand side contains only a diagonal matrix of operators. Such a form will make it easier for us to spot high-order terms that can be dropped in the final form of R13 equations. To this end, we define
\begin{equation}\label{Q}
\boldsymbol{\mathcal{Q}} = 
\begin{bmatrix}
  \mathcal{I} & 0 & 0 & 0 \\
  0 & \mathcal{I} & 0 & 0 \\
  0 & 0 & \mathcal{I} & 0 \\
  0 & 0 & -\mathcal{B}^{(32)} & \mathcal{I}
\end{bmatrix} \begin{bmatrix}
  \mathcal{I} & 0 & 0 & 0 \\
  0 & \mathcal{I} & 0 & 0 \\
  0 & -\mathcal{B}^{(21)} & \mathcal{I} & 0 \\
  0 & -\mathcal{B}^{(31)} & 0 & \mathcal{I}
\end{bmatrix},
\end{equation}
where
\begin{gather*}
\mathcal{B}^{(21)} = \mathcal{L}^{(21)} [\mathcal{L}^{(11)}]^{-1}, \qquad
\mathcal{B}^{(31)} = \mathcal{L}^{(31)} [\mathcal{L}^{(11)}]^{-1}, \\
\mathcal{B}^{(32)} =[\mathcal{L}^{(32)} - \mathcal{B}^{(31)} \mathcal{L}^{(12)}][\mathcal{L}^{(22)} - \mathcal{B}^{(21)} \mathcal{L}^{(12)}]^{-1}.
\end{gather*}
This procedure is similar to Gaussian elimination of $\boldsymbol{\mathcal{L}}$, and the invertibility of the operators in the definitions of $\mathcal{B}^{(kl)}$ can be guaranteed by the fact that the linearized Boltzmann collision operator $\mathcal{L}$ is negative semidefinite with its nullspace being $\mathbb{V}^{(0)}$. Let $\boldsymbol{\mathcal{Q}}^{\dagger}$ be the adjoint transpose of the operator matrix $\boldsymbol{\mathcal{Q}}$. It can be seen that the matrix $\boldsymbol{\mathcal{Q}} \boldsymbol{\mathcal{L}}\boldsymbol{\mathcal{Q}}^{\dagger}$ is diagonal.

We now apply this transformation to the projected Boltzmann equation \eqref{eq:pBoltzSimplified}. By introducing $\tilde{\boldsymbol{f}} = \boldsymbol{\mathcal{Q}}^{-\dagger} \boldsymbol{f}$, we can multiply both sides of \eqref{eq:pBoltzSimplified} by $\boldsymbol{\mathcal{Q}}$ and write the result as
\begin{equation} \label{eq:transformed}
\boldsymbol{\mathcal{Q}} \boldsymbol{\mathcal{A}}\boldsymbol{\mathcal{Q}}^{\dagger} \tilde{\boldsymbol{f}} = \frac{1}{\kn} \boldsymbol{\mathcal{Q}} \boldsymbol{\mathcal{L}}\boldsymbol{\mathcal{Q}}^{\dagger} \tilde{\boldsymbol{f}}.
\end{equation}
We claim that the operator matrices $\boldsymbol{\mathcal{Q}} \boldsymbol{\mathcal{A}}\boldsymbol{\mathcal{Q}}^{\dagger}$ and $\boldsymbol{\mathcal{Q}} \boldsymbol{\mathcal{L}}\boldsymbol{\mathcal{Q}}^{\dagger}$ have the following structures:
\begin{equation} \label{eq:QAQ_QLQ}
\boldsymbol{\mathcal{Q}} \boldsymbol{\mathcal{A}}\boldsymbol{\mathcal{Q}}^{\dagger} =\begin{bmatrix}
 \mathcal{A}^{(00)} &  \mathcal{A}^{(01)} & 0 & 0 \\
  \mathcal{A}^{(10)} &  \mathcal{A}^{(11)} &  \tilde{\mathcal{A}}^{(12)}  &  0 \\
    0 &  \tilde{\mathcal{A}}^{(21)} &  \tilde{\mathcal{A}}^{(22)}  &  \tilde{\mathcal{A}}^{(23)} \\
    0 & 0 &  \tilde{\mathcal{A}}^{(32)}  &  \tilde{\mathcal{A}}^{(33)}
\end{bmatrix}, \qquad
\boldsymbol{\mathcal{Q}} \boldsymbol{\mathcal{L}}\boldsymbol{\mathcal{Q}}^{\dagger} =\begin{bmatrix}
 0 & 0 & 0 & 0 \\
  0 & \mathcal{L}^{(11)} &  0 &  0 \\
    0 & 0 &  \tilde{\mathcal{L}}^{(22)}  & 0 \\
    0 & 0 & 0 & \tilde{\mathcal{L}}^{(33)}
\end{bmatrix}.
\end{equation}
The diagonal structure of $\boldsymbol{\mathcal{Q}} \boldsymbol{\mathcal{L}}\boldsymbol{\mathcal{Q}}^{\dagger}$ has been clarified in the construction of the matrix $\boldsymbol{\mathcal{Q}}$. To explain why $\boldsymbol{\mathcal{Q}} \boldsymbol{\mathcal{A}}\boldsymbol{\mathcal{Q}}^{\dagger}$ is tridiagonal, we need the following facts:
\begin{itemize}
\item By straightforward calculation, we have
\begin{displaymath}
\tilde{\boldsymbol{f}} := \begin{bmatrix} \tilde{f}^{(0)} \\ \tilde{f}^{(1)} \\ \tilde{f}^{(2)} \\ \tilde{f}^{(3)} \end{bmatrix}
= \begin{bmatrix}
f^{(0)} \\
f^{(1)} + [\mathcal{B}^{(21)}]^{\dagger} f^{(2)} + [\mathcal{B}^{(31)} + \mathcal{B}^{(32)} \mathcal{B}^{(21)}]^{\dagger} f^{(3)} \\
f^{(2)} + [\mathcal{B}^{(32)}]^{\dagger} f^{(3)} \\
f^{(3)}
\end{bmatrix},
\end{displaymath}
which shows that the last two components of $\boldsymbol{\mathcal{Q}} \boldsymbol{\mathcal{L}}\boldsymbol{\mathcal{Q}}^{\dagger} \tilde{\boldsymbol{f}}$ have magnitudes $O(\kn^2)$ and $O(\kn^3)$, respectively.
\item By \eqref{eq:transformed}, the last two components of $\boldsymbol{\mathcal{Q}} \boldsymbol{\mathcal{A}}\boldsymbol{\mathcal{Q}}^{\dagger} \tilde{\boldsymbol{f}}$ should have magnitudes $O(\kn)$ and $O(\kn^2)$, respectively. Therefore, the three operators below the subdiagonal of $\boldsymbol{\mathcal{Q}} \boldsymbol{\mathcal{A}}\boldsymbol{\mathcal{Q}}^{\dagger}$ can only be zero operators.
\item Due to the symmetric structure of $\boldsymbol{\mathcal{Q}} \boldsymbol{\mathcal{A}}\boldsymbol{\mathcal{Q}}^{\dagger}$, the three operators above its superdiagonal must also be zero operators.
\end{itemize}

\subsubsection{Step 2: Dropping high-order terms}
In this step, our purpose is to drop as many terms in \eqref{eq:transformed} as possible while retaining the super-Burnett order of these equations. One possible approach to obtain super-Burnett equations is to perform the following Maxwell iteration based on \eqref{eq:transformed}:
\begin{align*}
\tilde{f}_{k+1}^{(1)} &= \kn [\mathcal{L}^{(11)}]^{-1} \left( \mathcal{A}^{(10)} f^{(0)} + \mathcal{A}^{(11)} \tilde{f}_k^{(1)} + \tilde{\mathcal{A}}^{(12)} \tilde{f}_k^{(2)} \right), & \tilde{f}^{(1)}_0 &= 0, \\
\tilde{f}_{k+1}^{(2)} &= \kn [\tilde{\mathcal{L}}^{(22)}]^{-1}\left( \tilde{\mathcal{A}}^{(21)} f_k^{(1)} + \tilde{\mathcal{A}}^{(22)} \tilde{f}_k^{(2)} + \tilde{\mathcal{A}}^{(23)} \tilde{f}_k^{(3)} \right), & \tilde{f}^{(2)}_0 &= 0, \\
\tilde{f}_{k+1}^{(3)} &= \kn [\tilde{\mathcal{L}}^{(33)}]^{-1}\left( \tilde{\mathcal{A}}^{(32)} \tilde{f}_k^{(2)} + \tilde{\mathcal{A}}^{(33)} \tilde{f}_k^{(3)}\right), & \tilde{f}^{(3)}_0 &= 0.
\end{align*}
Note that the local equilibrium $f^{(0)}$ does not attend the iteration. The super-Burnett equations can be written as
\begin{displaymath}
\mathcal{A}^{(00)} f^{(0)} + \mathcal{A}^{(01)} \tilde{f}^{(1)}_3 = 0,
\end{displaymath}
where $\tilde{f}^{(1)}_3$ is from the result of three Maxwell iterations. Straightforward calculation yields 
\begin{displaymath}
\begin{split}
\tilde{f}^{(1)}_3
&=\kn [\mathcal{L}^{(11)}]^{-1} \mathcal{A}^{(10)} f^{(0)} + \kn^2 [\mathcal{L}^{(11)}]^{-1} \mathcal{A}^{(11)} [\mathcal{L}^{(11)}]^{-1} \mathcal{A}^{(10)} f^{(0)} \\
& \quad {} + \kn^3  ([\mathcal{L}^{(11)}]^{-1}  \mathcal{A}^{(11)})^2 [\mathcal{L}^{(11)}]^{-1} \mathcal{A}^{(10)} f^{(0)} \\
& \quad {} + \kn^3 [\mathcal{L}^{(11)}]^{-1} \tilde{\mathcal{A}}^{(12)} [\mathcal{L}^{(22)}]^{-1} \tilde{\mathcal{A}}^{(21)} 
 [\mathcal{L}^{(11)}]^{-1} \mathcal{A}^{(10)} f^{(0)}.
\end{split}
\end{displaymath}
which has nothing to do with the operators $\tilde{\mathcal{A}}^{(22)}$, $\tilde{\mathcal{A}}^{(23)}$, $\tilde{\mathcal{A}}^{(32)}$ and $\tilde{\mathcal{A}}^{(33)}$. Therefore, we can set these four operators in \eqref{eq:QAQ_QLQ} to be zero and claim that the resulting equations
\begin{equation} \label{eq:R13_transformed}
\begin{bmatrix}
 \mathcal{A}^{(00)} &  \mathcal{A}^{(01)} & 0 & 0 \\
  \mathcal{A}^{(10)} &  \mathcal{A}^{(11)} &  \tilde{\mathcal{A}}^{(12)}  &  0 \\
    0 &  \tilde{\mathcal{A}}^{(21)} & 0 & 0 \\
    0 & 0 & 0 & 0
\end{bmatrix}
\begin{bmatrix}
f^{(0)} \\
\tilde{f}^{(1)} \\
\tilde{f}^{(2)} \\
f^{(3)}
\end{bmatrix} = \frac{1}{\kn} \begin{bmatrix}
 0 & 0 & 0 & 0 \\
  0 & \mathcal{L}^{(11)} &  0 &  0 \\
    0 & 0 &  \tilde{\mathcal{L}}^{(22)}  & 0 \\
    0 & 0 & 0 & \tilde{\mathcal{L}}^{(33)}
\end{bmatrix}
\begin{bmatrix}
f^{(0)} \\
\tilde{f}^{(1)} \\
\tilde{f}^{(2)} \\
f^{(3)}
\end{bmatrix}
\end{equation}
still have the super-Burnett order.

\subsubsection{Step 3: Applying the inverse transformation} The equations \eqref{eq:R13_transformed} already give us the abstract form of the R13 equations. However, due to the transformation introduced by $\boldsymbol{\mathcal{Q}}$, the left-hand sides of these equations no longer represent the approximation of $\xi_j \partial_{x_j}f$, and the right-hand sides no longer represent the approximation of $\mathcal{L}f$. Recovering such straightforward correspondence requires applying the inverse transformation $\boldsymbol{\mathcal{Q}}^{-1}$. By multiplying both sides of \eqref{eq:R13_transformed} by $\boldsymbol{\mathcal{Q}}^{-1}$ and using $\tilde{\boldsymbol{f}} = \boldsymbol{\mathcal{Q}}^{-\dagger} \boldsymbol{f}$, we obtain
\begin{displaymath}
\begin{bmatrix}
 \mathcal{A}^{(00)} &  \mathcal{A}^{(01)} &  \mathcal{A}^{(02)} &  \mathcal{A}^{(03)} \\
  \mathcal{A}^{(10)} &  \mathcal{A}^{(11)} &  \mathcal{A}^{(12)}  &  \mathcal{A}^{(13)}  \\
    \mathcal{A}^{(20)} &  \mathcal{A}^{(21)} &  \bar{\mathcal{A}}^{(22)} &  \bar{\mathcal{A}}^{(23)} \\
     \mathcal{A}^{(30)} &   \mathcal{A}^{(31)} &  \bar{\mathcal{A}}^{(32)} &  \bar{\mathcal{A}}^{(33)}
\end{bmatrix}
\begin{bmatrix}
  f^{(0)} \\ f^{(1)} \\ f^{(2)} \\ f^{(3)}
\end{bmatrix}
= \frac{1}{\kn} \begin{bmatrix}
  0 & 0 & 0 & 0 \\
  0 &  \mathcal{L}^{(11)} &  \mathcal{L}^{(12)}  &  \mathcal{L}^{(13)}  \\
    0 &  \mathcal{L}^{(21)} &  \mathcal{L}^{(22)}  &  \mathcal{L}^{(23)} \\
     0 &   \mathcal{L}^{(31)} &  \mathcal{L}^{(32)}  &  \mathcal{L}^{(33)}
\end{bmatrix}
\begin{bmatrix}
  f^{(0)} \\ f^{(1)} \\ f^{(2)} \\ f^{(3)}
\end{bmatrix},
\end{displaymath}
where 
\begin{equation}\label{barA22}
\bar{\mathcal{A}}^{(22)} = \mathcal{B}^{(21)}\mathcal{A}^{(12)} + \mathcal{A}^{(21)} [\mathcal{B}^{(21)}]^{\dagger} - \mathcal{B}^{(21)}\mathcal{A}^{(11)}[\mathcal{B}^{(21)}]^{\dagger},
\end{equation}
and the operators $\bar{\mathcal{A}}^{(23)}$, $\bar{\mathcal{A}}^{(32)}$ and $\bar{\mathcal{A}}^{(33)}$ are unimportant since \eqref{eq:R13_transformed} already shows $f^{(3)} = 0$. We can actually remove $f^{(3)}$ completely from the system to get the final form \eqref{eq:R13}.

Since $f^{(3)}$ is no longer present in the final equations, they can be reformulated using a smaller function space
\begin{equation} \label{eq:barV}
\overline{\mathbb{V}} := \mathbb{V}^{(0)} \oplus \mathbb{V}^{(1)} \oplus \mathbb{V}^{(2)}.
\end{equation}
Let $\mathcal{P}_{\overline{\mathbb{V}}}$ be the projection operator onto $\overline{\mathbb{V}}$. The equations can then be written as
\begin{displaymath}
\frac{\partial}{\partial x_j} (\bar{\mathcal{A}}_j \bar{f}) = \frac{1}{\kn} \mathcal{P}_{\overline{\mathbb{V}}} \mathcal{L} \bar{f},
\end{displaymath}
where $\bar{f}$ is the unknown function in $\overline{\mathbb{V}}$, and
\begin{equation} \label{eq:Aj}
\bar{\mathcal{A}}_j = \mathcal{P}_{\overline{\mathbb{V}}} \xi_j - \left(\mathcal{A}_j^{(22)} - \mathcal{B}^{(21)}\mathcal{A}_j^{(12)} - \mathcal{A}_j^{(21)} [\mathcal{B}^{(21)}]^{\dagger} + \mathcal{B}^{(21)}\mathcal{A}_j^{(11)}[\mathcal{B}^{(21)}]^{\dagger}\right) \mathcal{P}^{(2)},
\end{equation}
which is an operator on $\overline{\mathbb{V}}$ approximating the multiplication of a function by the velocity component $\xi_j$.
It is straightforward to verify that $\bar{\mathcal{A}}_j$ is a self-adjoint operator on $\overline{\mathbb{V}}$.

So far, we have obtained a linear system with desired symmetric structure, allowing us to further derive stable boundary conditions. This will be discussed in Section \ref{sec:bc}. In what follows, we will first provide the explicit forms of these equations.

\subsection{Derivation of R13 moment equations}
\label{sec:R13_explicit}
We will now derive the R13 equations presented in Section \ref{sec:symR13}.
According to definitions of the function spaces $\overline{\mathbb{V}}^{(0)}$, $\overline{\mathbb{V}}^{(1)}$ and $\overline{\mathbb{V}}^{(2)}$ (see \eqref{eq:V0}\eqref{eq:V1}\eqref{eq:V2}), we can express the projection $\bar{f} := \mathcal{P}_{\overline{\mathbb{V}}} f$ explicitly as
\begin{equation} \label{f u expansion}
  \bar{f} = \underbrace{w^0 \psi^{0} + w^1 \psi^{1} + 3w_i^0 \psi_i^{0}}_{\text{zeroth order}} + \underbrace{3u_i^{(1)} \phi_i^{(1)} + \frac{15}{2} u_{ij}^{(1)} \phi_{ij}^{(1)}}_{\text{first order}} + \underbrace{u^{(2)} \phi^{(2)} + 3u_i^{(2)} \phi_i^{(2)} + \frac{15}{2}u_{ij}^{(2)} \phi_{ij}^{(2)} + \frac{35}{2}u_{ijk}^{(2)} \phi_{ijk}^{(2)}}_{\text{second order}}.
\end{equation}
For any $\bar{f} \in \overline{\mathbb{V}}$ with the above expression, we have
\begin{gather*}
w^0 = \langle  \bar{f} , \psi^{(0)} \rangle, \qquad w^1 = \langle \bar{f},\psi^{(1)} \rangle, \qquad w_i^0 = \langle  \bar{f} , \psi_i^{(0)} \rangle, \\
u_i^{(1)} =  \langle \bar{f} , \phi^{(1)}_{i} \rangle   , \qquad u_{ij}^{(1)} =  \langle \bar{f} , \phi^{(1)}_{ij} \rangle  , \\
u^{(2)} =   \langle \bar{f} , \phi^{(2)} \rangle  , \qquad u^{(2)}_i =  \langle \bar{f} , \phi^{(2)}_{i} \rangle, \qquad u^{(2)}_{ij} =  \langle \bar{f} , \phi^{(2)}_{ij} \rangle , \qquad  u^{(2)}_{ijk} = \langle \bar{f} , \phi^{(2)}_{ijk} \rangle. 
\end{gather*}
Our purpose is to find the expressions of $\partial_{x_j} (\bar{\mathcal{A}}_j \bar{f})$ and $\mathcal{P}_{\overline{\mathbb{V}}} \mathcal{L} \bar{f}$ under such representation.

We begin with the collision term $\mathcal{P}_{\overline{\mathbb{V}}} \mathcal{L} \bar{f}$. Due to the conservation laws, the zeroth-order part of $\bar{f}$ vanishes after applying $\mathcal{L}$. Using the rotational invariance of $\mathcal{L}$, we get
\begin{equation} \label{eq:PVLf}
\mathcal{P}_{\overline{\mathbb{V}}} \mathcal{L} \bar{f} = \mathscr{L}_0^{(22)} u^{(2)} \phi^{(2)} + \sum_{m=1}^2 \sum_{n=1}^2 \mathscr{L}_1^{(mn)} u_i^{(n)} \phi_i^{(m)} + \sum_{m=1}^2 \sum_{n=1}^2 \mathscr{L}_2^{(mn)} u_{ij}^{(n)} \phi_{ij}^{(m)} + \mathscr{L}_3^{(22)} u_{ijk}^{(2)} \phi_{ijk}^{(2)},
\end{equation}
where 
\begin{equation}\label{Lmnl}
\mathscr{L}_l^{(mn)} = \frac{(2l+1)!!}{l!} \frac{\langle \phi_{i_1\cdots i_l}^{(m)}, \mathcal{L}\phi_{i_1\cdots i_l}^{(n)} \rangle}{\langle \phi_{i_1\cdots i_l}^{(m)}, \phi_{i_1\cdots i_l}^{(m)} \rangle}
 =\frac{(2l+1)!!}{l!} \sum_{m',n'} c^{(m),m'}_l c^{(n),n'}_l a_{ln'm'}
 \end{equation}
This value is independent of the choice of $i_1, \cdots, i_l$ due to the rotational invariance, and is symmetric with respect to the superscript $(mn)$. The expansion \eqref{eq:PVLf} can be formally expressed by 
\begin{equation} \label{eq:R13 collision}
\boldsymbol{\phi}^T 
\left[ \begin{array}{ccc|c@{\!\!}c|c@{\!\!}c@{\!\!}c@{\!\!}c@{}}
0 & 0 & 0 & 0 & 0 & 0 & 0 & 0 & 0 \\
0 & 0 & 0 & 0 & 0 & 0 & 0 & 0 & 0 \\
0 & 0 & 0 & 0 & 0 & 0 & 0 & 0 & 0 \\
\hline \rule{0pt}{15pt}
0 & 0 & 0 & \mathscr{L}_1^{(11)} & 0 & 0 & \mathscr{L}_1^{(12)} & 0 & 0 \\
0 & 0 & 0 & 0 & \mathscr{L}_2^{(11)} & 0 & 0 & \mathscr{L}_2^{(12)} & 0 \\[7pt]
\hline  \rule{0pt}{15pt}
0 & 0 & 0 & 0 & 0 & \mathscr{L}_0^{(22)} & 0 & 0 & 0 \\
0 & 0 & 0 & \mathscr{L}_1^{(21)} & 0 & 0 & \mathscr{L}_1^{(22)} & 0 & 0 \\
0 & 0 & 0 & 0 & \mathscr{L}_2^{(21)} & 0 & 0 & \mathscr{L}_2^{(22)} & 0 \\
0 & 0 & 0 & 0 & 0 & 0 & 0 & 0 & \mathscr{L}_3^{(22)}
\end{array} \right]
\boldsymbol{u}
\end{equation}
where 
\begin{displaymath}
    \begin{split}
\boldsymbol{\phi}^T = &  \ 
\left( \psi^0\ \psi^1\ \psi^0_i \ \vline \  \phi^{(1)}_i\ \phi^{(1)}_{ij} \ \vline \   \phi^{(2)} \ \phi^{(2)}_i\ \phi^{(2)}_{ij}\ \phi^{(2)}_{ijk} \right), \\ 
\boldsymbol{u}^T = &  \ 
\left( w^0\ w^1\ w^0_i \ \vline \  u^{(1)}_i\ u^{(1)}_{ij} \ \vline \   u^{(2)} \ u^{(2)}_i\ u^{(2)}_{ij}\  u^{(2)}_{ijk} \right),
 \end{split}
\end{displaymath}

The left-hand side $\partial_{x_j} (\bar{\mathcal{A}}_j \bar{f})$ should be calculated according to \eqref{eq:Aj}. We can first find $\partial_{x_j}(\mathcal{P}_{\overline{\mathbb{V}}} \xi_j \bar{f})$ by straightforward calculation using \eqref{psi xi}. We will again denote the result formally using the form \eqref{eq:R13 collision}
\begin{equation} \label{eq:R13 lhs raw}
 \boldsymbol{\phi}^T\resizebox{.8\hsize}{!}{$
\left[ \begin{array}{@{}c@{\!}c@{\!\!\!\!}c|c@{\!\!\!}c|c@{\!\!\!}c@{\!\!\!}c@{\!\!\!}c@{}}
     0 & 0 & \sqrt{3}\frac{\partial}{\partial x_i} & 0 & 0 & 0 & 0 & 0 & 0 \\[5pt]
     0 & 0 & -\sqrt{2}\frac{\partial}{\partial x_i} & \sqrt{5}c_1^{(1),1}\frac{\partial}{\partial x_i} & 0  & 0 & \sqrt{5}c^{(2),1}_1\frac{\partial}{\partial x_i} & 0 & 0 \\[5pt]
        \sqrt{3}\frac{\partial}{\partial x_i} &  -\sqrt{2}\frac{\partial}{\partial x_i} & 0 & 0 & 3\sqrt{5}c_2^{(1),0}\frac{\partial}{\partial x_j} & 0 & 0 & 3\sqrt{5}c^{(2),0}_2 \frac{\partial}{\partial x_j} & 0 \\[5pt] \hline  \rule{0pt}{15pt}
   0 &  \sqrt{5} c_1^{(1),1} \frac{\partial}{\partial x_i} & 0 & 0 & A_{45}\frac{\partial}{\partial x_j}  & A_{46}\frac{\partial}{\partial x_i}  & 0 & A_{48}\frac{\partial}{\partial x_j} & 0 \\
     0 &  0 & 3\sqrt{5} c_2^{(1),0}\frac{\partial_{\langle}}{\partial x_{j\rangle}} & A_{45}\frac{\partial_{\langle}}{\partial x_{j\rangle}}   & 0 & 0 & A_{57}\frac{\partial_\langle}{\partial x_{j\rangle}} & 0 & A_{59}\frac{\partial}{\partial x_k} \\[7pt]  \hline  \rule{0pt}{15pt}
     0 &  0 &  0 &  A_{46}\frac{\partial}{\partial x_i} &  0 &  0 &  A_{67}\frac{\partial}{\partial x_i} &  0 &  0 \\
     0 & \sqrt{5}c^{(2),1}_1\frac{\partial}{\partial x_i} &  0 & 0 &  A_{57}\frac{\partial}{\partial x_{j}} &  A_{67}\frac{\partial}{\partial x_{i}} &  0 &  A_{78}\frac{\partial}{\partial x_{j}} & 0 \\ 
          0 & 0 &  3\sqrt{5}c^{(2),0}_2 \frac{\partial_\langle}{\partial x_{j\rangle}} & A_{48}\frac{\partial_\langle}{\partial x_{j\rangle}} &  0 &  0 &  A_{78}\frac{\partial_\langle}{\partial x_{j\rangle}} &  0 &  A_{89}\frac{\partial}{\partial x_{k}} \\
          0 & 0 & 0 & 0 & A_{59}\frac{\partial_\langle}{\partial x_{k\rangle}} & 0 & 0 &  A_{89}\frac{\partial_\langle}{\partial x_{k\rangle}} & 0 \\[5pt] 
\end{array} \right]
$}\boldsymbol{u} 
\end{equation}
for conciseness, where the expressions of $A_{ij}$ can be found in Section \ref{app:matrices}. The derivative operator $\frac{\partial_{\langle}}{\partial x_{j \rangle}}$ in the above matrix is defined such that $\frac{\partial_{\langle}}{\partial x_{j \rangle}} u^{(d)}_{i_1\cdots i_l} = \frac{\partial u^{(d)}_{\langle i_1\cdots i_l}}{\partial x_{j\rangle}}$ where Einstein's summation should be applied when any of $i_1,\cdots,i_l$ is $j$. 

To get the final form of the R13 equations, we replace the lower-right block of the matrix on the left-hand side by an operator corresponding to $\bar{\mathcal{A}}^{(22)}$ defined in \eqref{barA22}. According to \eqref{eq:R13 collision}, the operator $\mathcal{B}^{(21)} = \mathcal{L}^{(21)} [\mathcal{L}^{(11)}]^{-1}$ can be represented by the matrix 
\begin{displaymath}
\pmb{\mathscr{B}}^{(21)} := \begin{bmatrix}
0 & 0 \\ \mathscr{L}_1^{(21)} & 0 \\
0 & \mathscr{L}_2^{(21)} \\ 0 & 0
\end{bmatrix}
\begin{bmatrix}
\mathscr{L}_1^{(11)} & 0 \\
0 & \mathscr{L}_2^{(11)}
\end{bmatrix}^{-1} =
\begin{bmatrix}
0 & 0 \\ \mathscr{L}_1^{(21)} / \mathscr{L}_1^{(11)}  & 0 \\
0 & \mathscr{L}_2^{(21)} / \mathscr{L}_2^{(11)} \\ 0 & 0
\end{bmatrix}.
\end{displaymath}
Note that we actually have
\begin{equation} \label{eq:L}
\mathscr{L}_1^{(21)} / \mathscr{L}_1^{(11)} = c_1^{(2),1} / c_1^{(1),1}, \quad \text{and} \quad \mathscr{L}_2^{(21)} / \mathscr{L}_2^{(11)} = c_2^{(2),0} / c_2^{(1),0}
\end{equation}
since $\mathcal{A}^{(20)} - \mathcal{B}^{(21)} \mathcal{A}^{(10)} = 0$ (see \eqref{eq:QAQ_QLQ}).
Thus, following \eqref{barA22}, the operator $\bar{\mathcal{A}}^{(22)}$ can be represented by
\begin{displaymath}
\begin{split}
& \pmb{\mathscr{B}}^{(21)} \!
\left[\begin{array}{@{}c@{\!\!}c@{\!\!}c@{\!\!}c@{}}
    A_{46}\frac{\partial}{\partial x_i}  & 0 & A_{48}\frac{\partial}{\partial x_j} & 0 \\
    0 & A_{57}\frac{\partial_\langle}{\partial x_{j\rangle}} & 0 & A_{59}\frac{\partial}{\partial x_k}
\end{array} \right] +
\left[ \begin{array}{@{}c@{\!\!}c@{}}
     A_{46}\frac{\partial}{\partial x_i} & 0 \\
     0 & A_{57}\frac{\partial}{\partial x_{j}} \\ 
     A_{48}\frac{\partial_\langle}{\partial x_{j\rangle}} & 0 \\
     0 & A_{59}\frac{\partial_\langle}{\partial x_{k\rangle}}
\end{array} \right] \! [\pmb{\mathscr{B}}^{(21)}]^T
- \pmb{\mathscr{B}}^{(21)} \!
\left[\begin{array}{@{}c@{\!\!}c@{}}
  0 & A_{45} \frac{\partial}{\partial x_j} \\
  A_{45} \frac{\partial_{\langle}}{\partial x_{j\rangle}} & 0
\end{array} \right] \!
[\pmb{\mathscr{B}}^{(21)}]^T \\
={} & \begin{bmatrix}
  0 & \frac{c^{(2),1}_1}{c^{(1),1}_1} A_{46}  \frac{\partial }{\partial x_i}& 0 & 0 \\
  \frac{c^{(2),1}_1}{c^{(1),1}_1} A_{46}  \frac{\partial }{\partial x_i} & 0 & \bar{A}_{78}\frac{\partial }{\partial x_j } & 0 \\
  0 & \bar{A}_{78}\frac{\partial_{\langle} }{\partial x_{j \rangle}} & 0 & \frac{c^{(2),0}_2}{c^{(1),0}_2} A_{59}  \frac{\partial }{\partial x_k} \\
  0 & 0 & \frac{c^{(2),0}_2}{c^{(1),0}_2} A_{59}  \frac{\partial_{\langle} }{\partial x_{k \rangle}} & 0
\end{bmatrix}
\end{split}
\end{displaymath}
where 
\begin{equation} \label{eq:A78}
   \bar{A}_{78} =  \frac{c^{(2),1}_1}{c^{(1),1}_1} A_{48}+ \frac{c^{(2),0}_2}{c^{(1),0}_2} A_{57}- \frac{c^{(2),1}_1 c^{(2),0}_2}{c^{(1),1}_1 c^{(1),0}_2} A_{45}. 
\end{equation}
The final form of $\partial_{x_j}(\bar{\mathcal{A}}_j \bar{f})$ can be obtained by using the matrix above to replace the lower-right block of the matrix in \eqref{eq:R13 lhs raw}. The result will then be equated to \eqref{eq:R13 collision} to get the explicit expressions of the R13 equations, which are given as follows: 
 \begin{align}
    \label{eq1}
      \sqrt{3} \frac{\partial w^0_j}{\partial x_j} & = 0,  \\ 
    \label{eq2}
     -\sqrt{2}  \frac{\partial w^0_j}{\partial x_j}  + \sqrt{5}c_1^{(1),1}\frac{\partial u^{(1)}_j}{\partial x_j} +   \sqrt{5}c^{(2),1}_1 \frac{\partial u^{(2)}_j}{\partial x_j} & = 0 ,\\ 
    \label{eq3}
     \sqrt{3} \frac{\partial w^0}{\partial x_i} - \sqrt{2}\frac{\partial w^1}{\partial x_i} + 3\sqrt{5} c_2^{(1),0} \frac{\partial u^{(1)}_{ij}}{\partial x_j}  + 3\sqrt{5}  c^{(2),0}_2 \frac{\partial u^{(2)}_{ij}}{\partial x_j} & = 0, \\ 
    \label{eq4}
      \sqrt{5}c_1^{(1),1} \frac{\partial w^1}{\partial x_i} + A_{45} \frac{\partial u^{(1)}_{ij}}{\partial x_j} + A_{46} \frac{\partial u^{(2)}}{\partial x_i} + A_{48}\frac{\partial u^{(2)}_{ij}}{\partial x_j} & = \frac{1}{\kn}(\mathscr{L}^{(11)}_1 u^{(1)}_{i}+\mathscr{L}^{(12)}_1 u^{(2)}_{i}),\\ 
      \label{eq5}
      3\sqrt{5} c_2^{(1),0} \frac{\partial w^{0}_{\langle i}}{\partial x_{j \rangle}} + A_{45}\frac{\partial u^{(1)}_{\langle i}}{\partial x_{j \rangle}} + A_{57}\frac{\partial u^{(2)}_{\langle i}}{\partial x_{j \rangle}} + A_{59} \frac{\partial u^{(2)}_{ijk}}{\partial x_k} & = \frac{1}{\kn}(\mathscr{L}^{(11)}_2 u^{(1)}_{ij}+\mathscr{L}^{(12)}_2 u^{(2)}_{ij}),\\ 
      \label{eq6}
      A_{46} \frac{\partial u^{(1)}_j}{\partial x_j}  + \frac{c^{(2),1}_1}{c^{(1),1}_1} A_{46}  \frac{\partial u^{(2)}_j}{\partial x_j} & =  \frac{1}{\kn}\mathscr{L}^{(22)}_0 u^{(2)},\\ 
      \label{eq7}
            \sqrt{5} c^{(2),1}_1 \frac{\partial w^1}{\partial x_i} + A_{57} \frac{\partial u^{(1)}_{ij}}{\partial x_j} + \frac{c^{(2),1}_1}{c^{(1),1}_1} A_{46} \frac{\partial u^{(2)}}{\partial x_i} + \bar{A}_{78}\frac{\partial u^{(2)}_{ij}}{\partial x_j} & = \frac{1}{\kn}(\mathscr{L}^{(21)}_1 u^{(1)}_{i}+\mathscr{L}^{(22)}_1 u^{(2)}_{i}),\\ 
        \label{eq8}
        3\sqrt{5} c^{(2),0}_2 \frac{\partial w^{0}_{\langle i}}{\partial x_{j \rangle}} + A_{48}\frac{\partial u^{(1)}_{\langle i}}{\partial x_{j \rangle}} + \bar{A}_{78}\frac{\partial u^{(2)}_{\langle i}}{\partial x_{j \rangle}} + \frac{c^{(2),0}_2}{c^{(1),0}_2} A_{59} \frac{\partial u^{(2)}_{ijk}}{\partial x_k} & = \frac{1}{\kn}(\mathscr{L}^{(21)}_2 u^{(1)}_{ij}+\mathscr{L}^{(22)}_2 u^{(2)}_{ij}),\\ 
       \label{eq9}
       A_{59} \frac{\partial u^{(1)}_{\langle ij}}{\partial x_{k \rangle}} +  \frac{c^{(2),0}_2}{c^{(1),0}_2} A_{59}\frac{\partial u^{(2)}_{\langle ij}}{\partial x_{k \rangle}} & = \frac{1}{\kn} \mathscr{L}^{(22)}_3 u^{(2)}_{ijk}. 
 \end{align}

It may be more interesting to write down these equations using the variables $w_{ij}^0$ and $w_i^1$. For the distribution function \eqref{f u expansion}, the moments $w_{ij}^0$ can be related to the coefficients by
\begin{equation} \label{eq:wij}
 w_{ij}^0 = \langle \psi_{ij}^0, \bar{f} \rangle = u_{i'j'}^{(1)} \langle \psi_{ij}^0, \phi_{i'j'}^{(1)} \rangle + u_{i'j'}^{(2)} \langle \psi_{ij}^0, \phi_{i'j'}^{(2)} \rangle = c_2^{(1),0} u_{ij}^{(1)} + c_2^{(2),0} u_{ij}^{(2)} .
\end{equation}
Similarly, we have
\begin{equation} \label{eq:wi}
w_i^1 = c_1^{(1),1} u_i^{(1)} + c_1^{(2),1} u_i^{(2)}.
\end{equation}
Using these variables, equations \eqref{eq6} and \eqref{eq9} become much neater: 
\begin{equation} \label{eq:u0 u3}
   u^{(2)} = \kn \beta'_0 \frac{\partial w^1_j}{\partial x_j}, \qquad u^{(2)}_{ijk} = \kn \beta'_3 \frac{\partial w^0_{\langle ij}}{\partial x_{k \rangle}}
\end{equation}
where 
\begin{displaymath}
   \beta_0' = \frac{A_{46}}{c^{(1),1}_1\mathscr{L}_0^{(22)}}, \qquad \beta_3' = \, \frac{A_{59}}{c^{(1),0}_2\mathscr{L}_3^{(22)}}.
\end{displaymath}
The other two second-order variables $u_i^{(2)}$ and $u_{ij}^{(2)}$ can also be represented using derivatives of $w_{ij}^0$ or $w_i^1$. To get $u_i^{(2)}$, we need to multiply \eqref{eq4} by $c_1^{(2),1}/c_1^{(1),1}$ and subtract the result by \eqref{eq7}. During the calculation, we need to use \eqref{eq:L} and \eqref{eq:A78} to get
\begin{equation} \label{eq:u1}
   u^{(2)}_i = \kn \beta'_1 \frac{\partial w^0_{ij}}{\partial x_j}, \qquad \beta_1'  = \frac{c^{(1),1}_1 A_{57} - c^{(2),1}_1A_{45}}{c^{(1),0}_2(c^{(1),1}_1\mathscr{L}_1^{(22)} - c^{(2),1}_1\mathscr{L}_1^{(12)})}.
\end{equation}
Similarly, we can use \eqref{eq5} and \eqref{eq8} to find
\begin{equation} \label{eq:u2}
   u^{(2)}_{ij} = \kn \beta'_2 \frac{\partial w^1_{\langle i}}{\partial x_{j \rangle}}, \qquad\beta_2' = \frac{c^{(1),0}_2 A_{48} - c^{(2),0}_2A_{45}}{c^{(1),1}_1(c^{(1),0}_2\mathscr{L}_2^{(22)} - c^{(2),0}_2 \mathscr{L}_2^{(12)})}.
\end{equation}
Finally, we can plug \eqref{eq:wij}--\eqref{eq:u2} into \eqref{eq1}--\eqref{eq5} to get a linear system written completely in the variables $w^0, w^1, w_i^0, w_i^1, w_{ij}^0$. The linear system in Section \ref{sec:symR13} can then be obtained by applying the relationship \eqref{physics quantities}.
\begin{remark}
For Maxwell molecules, due to the lack of three dimensions in $\mathbb{V}^{(2)}$, the variables  $u^{(2)}_i$ do not exist. As a result, one needs to remove \eqref{eq7} from the system. Also, the expression of $\beta'_1$ in \eqref{eq:u1} becomes $\frac{0}{0}$ for Maxwell molecules, and we need to set $\beta'_1$ to be zero to obtain correct equations.
\end{remark}


\section{Derivation of Onsager boundary conditions for R13 equations}
\label{sec:bc}
We are now ready to derive wall boundary conditions for the R13 equations. In this work, we will focus on Maxwell's accommodation model \cite{Maxwell1878}, which considers the interaction between gas molecules and the solid wall as a combination of specular reflection and diffusive reflection. In the derivation of boundary conditions for R13 equations, we will again write an abstract form using operators on function spaces to avoid lengthy formulas, and then convert it to its concrete form. Before starting our derivation, we will first briefly review Maxwell's boundary conditions for the linearized Boltzmann equation.

\subsection{Maxwell's boundary condition}
Consider the boundary point at which the outer normal unit vector is $\boldsymbol{n} = (n_1, n_2, n_3)^T$. For simplicity, we adopt the coordinate system with basis vectors $\boldsymbol{n}$, $\boldsymbol{t}_1$ and $\boldsymbol{t}_2$ with $\boldsymbol{t}_1$ and $\boldsymbol{t}_2$ being two orthogonal tangent vectors. Thus, the distribution function can now be presented by $f(\xi_n, \xi_{t_1}, \xi_{t_2})$, where $\xi_n = \bxi \cdot \boldsymbol{n}$ and $\xi_{t_1}$ and $\xi_{t_2}$ are similarly defined. For all other vectors and tensors, the indices will also be changed from $1,2,3$ to $n,t_1,t_2$ in this section. Thus, the Boltzmann equation can be written as 
\begin{displaymath}
\xi_n \frac{\partial f}{\partial x_n} +
\xi_{t_1} \frac{\partial f}{\partial x_{t_1}} +
\xi_{t_2} \frac{\partial f}{\partial x_{t_2}} = \frac{1}{\kn}\mathcal{L}f.
\end{displaymath}

For hyperbolic equations, boundary conditions are needed only for incoming characteristics. Maxwell \cite{Maxwell1878} proposed the following wall boundary condition for the distribution function:
\begin{equation} \label{eq:Maxwell bc}
f(\xi_n, \xi_{t_1}, \xi_{t_2}) = \chi f_W(\xi_n, \xi_{t_1}, \xi_{t_2}) + (1-\chi) f(-\xi_n, \xi_{t_1}, \xi_{t_2}), \qquad \xi_n < 0,
\end{equation}
and $\chi \in [0,1]$ is the accommodation coefficient denoting the proportion of the diffusive reflection. Assume that the solid wall has temperature $\theta^W$ and only moves in the tangential direction with velocity $(0, v_{t_1}^W, v_{t_2}^W)$. Then the ``wall Maxwellian'' $f_W(\bxi)$ has the expression 
\begin{displaymath}
   f_W(\xi_n, \xi_{t_1}, \xi_{t_2}) = w^{0,W} \psi^0 + w^{1,W} \psi^1 + 3w_{t_1}^{0,W} \psi^0_{t_1} + 3w_{t_2}^{0,W} \psi^0_{t_2}.
\end{displaymath}
Here $w^{1,W} = -\sqrt{\frac{3}{2}}\theta^W$, $w^{0,W}_{t_i} = \frac{\sqrt{3}}{3} v^W_{t_i}$, and $w^{0,W}$ are chosen such that
\begin{equation} \label{eq:zero_flux}
\int_{\mathbb{R}^3} \xi_n f(\xi_n, \xi_{t_1}, \xi_{t_2}) \,\mathrm{d}\xi_n \,\mathrm{d}\xi_{t_1} \,\mathrm{d}\xi_{t_2} = 0,
\end{equation}
meaning that the normal mass flow is zero on the solid wall.

To derive boundary conditions for moment equations, it is more convenient to rewrite the \eqref{eq:Maxwell bc} in the following equivalent form:
\begin{equation} \label{eq:odd_even}
f_{\mathrm{odd}}(\xi_n, \xi_{t_1}, \xi_{t_2}) = \frac{\chi}{2-\chi} [f_W(\xi_n, \xi_{t_1}, \xi_{t_2}) - f_{\mathrm{even}}(\xi_n, \xi_{t_1}, \xi_{t_2})], \qquad \xi_n < 0,
\end{equation}
where $f_{\mathrm{odd}}$ and $f_{\mathrm{even}}$ refer to the odd and even parts of $f$:
\begin{align*}
f_{\mathrm{odd}}(\xi_n, \xi_{t_1}, \xi_{t_2}) &= \frac{f(\xi_n, \xi_{t_1}, \xi_{t_2}) - f(-\xi_n, \xi_{t_1}, \xi_{t_2})}{2}, \\
f_{\mathrm{even}}(\xi_n, \xi_{t_1}, \xi_{t_2}) &= \frac{f(\xi_n, \xi_{t_1}, \xi_{t_2}) + f(-\xi_n, \xi_{t_1}, \xi_{t_2})}{2}.
\end{align*}
Alternatively, we can define $f_{\mathrm{odd}}$ and $f_{\mathrm{even}}$ by introducing the odd and even function spaces:
\begin{align*}
L_{\mathrm{odd}}^2(\mathbb{R}^3, [f_M(\bxi)]^{-1} \,\mathrm{d}\bxi) &= \left\{f \in L^2(\mathbb{R}^3, [f_M(\bxi)]^{-1} \,\mathrm{d}\bxi) \mid f(-\xi_n, \xi_{t_1}, \xi_{t_2}) = -f(\xi_n, \xi_{t_1}, \xi_{t_2}) \right\}, \\
L_{\mathrm{even}}^2(\mathbb{R}^3, [f_M(\bxi)]^{-1} \,\mathrm{d}\bxi) &= \left\{f \in L^2(\mathbb{R}^3, [f_M(\bxi)]^{-1} \,\mathrm{d}\bxi) \mid f(-\xi_n, \xi_{t_1}, \xi_{t_2}) = f(\xi_n, \xi_{t_1}, \xi_{t_2}) \right\},
\end{align*}
and let
\begin{displaymath}
f_{\mathrm{odd}} = \mathcal{P}_{\mathrm{odd}} f, \qquad 
f_{\mathrm{even}} = \mathcal{P}_{\mathrm{even}} f.
\end{displaymath}
Here $\mathcal{P}_{\mathrm{odd}}$ and $\mathcal{P}_{\mathrm{even}}$ are projection operators onto the function spaces $L_{\mathrm{odd}}^2(\mathbb{R}^3, [f_M(\bxi)]^{-1} \,\mathrm{d}\bxi)$ and $L_{\mathrm{even}}^2(\mathbb{R}^3, [f_M(\bxi)]^{-1} \,\mathrm{d}\bxi)$, respectively. Thus, a more compact way to write down Maxwell's boundary condition \eqref{eq:odd_even} is 
\begin{displaymath}
\mathcal{C} \mathcal{P}_{\mathrm{odd}} f = \frac{\chi}{2-\chi} \mathcal{C}(f_W - \mathcal{P}_{\mathrm{even}} f),
\end{displaymath}
where $\mathcal{C}$ is half-space operator defined by
\begin{equation} \label{eq:bc op}
\mathcal{C}g(\xi_n, \xi_{t_1}, \xi_{t_2}) = \begin{cases}
  g(\xi_n, \xi_{t_1}, \xi_{t_2}), & \text{if } \xi_n < 0, \\
  0, & \text{if } \xi_n > 0.
\end{cases}
\end{equation}
Noting that $\mathcal{P}_{\mathrm{odd}} = 2\mathcal{P}_{\mathrm{odd}} \mathcal{C}\mathcal{P}_{\mathrm{odd}}$, we can apply $\mathcal{P}_{\mathrm{odd}}$ to both sides of \eqref{eq:bc op} to get
\begin{equation} \label{eq:bc}
\mathcal{P}_{\mathrm{odd}} f = \frac{2\chi}{2-\chi} \mathcal{P}_{\mathrm{odd}} \mathcal{C}(f_W - \mathcal{P}_{\mathrm{even}} f).
\end{equation}
To determine $w^{0,W}$ in $f_W$, we need to use \eqref{eq:zero_flux}, which can also be formulated as
\begin{equation} \label{eq:zero_vel}
\langle \psi_n^0, f \rangle = 0.
\end{equation}
The complete Maxwell's boundary condition includes both \eqref{eq:bc} and \eqref{eq:zero_vel}.

\subsection{A first attempt to formulate boundary conditions for R13 equations}
Following the formulation of Maxwell's boundary conditions \eqref{eq:bc}, we will also split the function space for R13 equations into an odd part and an even part:
\begin{displaymath}
\overline{\mathbb{V}} = \overline{\mathbb{V}}_{\mathrm{odd}}\oplus\overline{\mathbb{V}}_{\mathrm{even}},
\end{displaymath}
where
\begin{displaymath}
\overline{\mathbb{V}}_{\mathrm{odd}} = \overline{\mathbb{V}} \cap L_{\mathrm{odd}}^2(\mathbb{R}^3, [f_M(\bxi)]^{-1} \,\mathrm{d}\bxi), \qquad
\overline{\mathbb{V}}_{\mathrm{even}} = \overline{\mathbb{V}} \cap L_{\mathrm{even}}^2(\mathbb{R}^3, [f_M(\bxi)]^{-1} \,\mathrm{d}\bxi).
\end{displaymath}
Grad's work \cite{Grad1949} states that the boundary conditions of moment equations should be formulated by taking only odd moments of the kinetic boundary conditions. Therefore, we expect that the boundary conditions for R13 equations are written as a map from $\overline{\mathbb{V}}_{\mathrm{even}}$ to $\overline{\mathbb{V}}_{\mathrm{odd}}$, which is analogous to \eqref{eq:bc}.

Furthermore, to get Onsager boundary conditions satisfying the $L^2$ stability, we write the operator $\bar{\mathcal{A}}_n = n_j \bar{\mathcal{A}}_j$ in the following form:
\begin{equation} \label{eq:oe_decomp}
\bar{\mathcal{A}}_n = \bar{\mathcal{A}}_{\mathrm{eo}} \mathcal{P}_{\mathrm{odd}}|_{\overline{\mathbb{V}}} + \bar{\mathcal{A}}_{\mathrm{oe}}\mathcal{P}_{\mathrm{even}}|_{\overline{\mathbb{V}}}.
\end{equation}
Here $\mathcal{P}_{\mathrm{odd}}|_{\overline{\mathbb{V}}}$ and $\mathcal{P}_{\mathrm{even}}|_{\overline{\mathbb{V}}}$ are the restrictions of $\mathcal{P}_{\mathrm{odd}}$ and $\mathcal{P}_{\mathrm{even}}$ on $\overline{\mathbb{V}}$, respectively.
Such a form of $\bar{\mathcal{A}}_n$ can be observed from the equations \eqref{eq1}--\eqref{eq9}, and it corresponds to the equation \eqref{eq:PAP} when written in the matrix form. The Onsager boundary conditions should hold the form
\begin{equation} \label{eq:Onsager}
\mathcal{P}_{\mathrm{odd}}|_{\overline{\mathbb{V}}} \, \bar{f} = \mathcal{S} \bar{\mathcal{A}}_{\mathrm{oe}}(f_W - \mathcal{P}_{\mathrm{even}}|_{\overline{\mathbb{V}}} \,\bar{f}),
\end{equation}
where $\mathcal{S}$ is a self-adjoint and negative semidefinite operator on $\overline{\mathbb{V}}_{\mathrm{odd}}$. Note that \eqref{eq:Onsager} is the operator form of the boundary conditions \eqref{eq:stableBC}. The operator $\mathcal{S}$ can be figured out by comparing \eqref{eq:Onsager} with \eqref{eq:bc}: since $\bar{\mathcal{A}}_{\mathrm{oe}}$ approximates the operation that multiplies an even function by $\xi_n$, the operator $\mathcal{S}$ should approximate the operator
\begin{displaymath}
\frac{2\chi}{2-\chi} \mathcal{P}_{\mathrm{odd}} \mathcal{C} \xi_n^{-1},
\end{displaymath}
so that \eqref{eq:Onsager} can be regarded as a discretization of \eqref{eq:bc}. Thus, a natural choice is
\begin{equation} \label{eq:S}
\mathcal{S} = \frac{2\chi}{2-\chi} \mathcal{P}_{\overline{\mathbb{V}}}\mathcal{P}_{\mathrm{odd}} \mathcal{C} \xi_n^{-1}.
\end{equation}
Note that $\mathcal{P}_{\overline{\mathbb{V}}}\mathcal{P}_{\mathrm{odd}}$ is the projection operator onto $\overline{\mathbb{V}}_{\mathrm{odd}}$, and therefore $\mathcal{S}$ is an operator on $\overline{\mathbb{V}}_{\mathrm{odd}}$. The proposition below shows some desired properties of $\mathcal{S}$:

\begin{proposition}
The operator $\mathcal{S}$ is self-adjoint and negative semidefinite.
\end{proposition}
\begin{proof}
For any $\bar{f}_{\mathrm{o}},\bar{g}_{\mathrm{o}} \in \overline{\mathbb{V}}_{\mathrm{odd}}$, we have
\begin{displaymath}
\begin{split}
\langle \bar{f}_{\mathrm{o}}, \mathcal{S}\bar{g_{\mathrm{o}}} \rangle &= \frac{2\chi}{2-\chi} \langle\bar{f}_{\mathrm{o}}, \mathcal{P}_{\overline{\mathbb{V}}}\mathcal{P}_{\mathrm{odd}} \mathcal{C} \xi_n^{-1} \bar{g}_{\mathrm{o}}\rangle = \langle\mathcal{P}_{\overline{\mathbb{V}}}\mathcal{P}_{\mathrm{odd}}\bar{f}_{\mathrm{o}}, \mathcal{C} \xi_n^{-1} \bar{g}_{\mathrm{o}}\rangle = \langle \bar{f}_{\mathrm{o}}, \mathcal{C} \xi_n^{-1} \bar{g}_{\mathrm{o}}\rangle \\
&= \int_{-\infty}^0 \int_{-\infty}^{+\infty} \int_{-\infty}^{+\infty} \bar{f}_{\mathrm{o}}(\xi_n, \xi_{t_1}, \xi_{t_2}) \xi_n^{-1} \bar{g}_{\mathrm{o}}(\xi_n, \xi_{t_1}, \xi_{t_2}) \,\mathrm{d} \xi_{t_2} \,\mathrm{d} \xi_{t_1} \,\mathrm{d} \xi_n \\
&= \langle \mathcal{C} \xi_n^{-1} \bar{f}_{\mathrm{o}}, \bar{g}_{\mathrm{o}}\rangle = \langle \mathcal{C} \xi_n^{-1} \bar{f}_{\mathrm{o}}, \mathcal{P}_{\overline{\mathbb{V}}}\mathcal{P}_{\mathrm{odd}}\bar{g}_{\mathrm{o}}\rangle = \langle \mathcal{P}_{\overline{\mathbb{V}}}\mathcal{P}_{\mathrm{odd}} \mathcal{C} \xi_n^{-1} \bar{f}_{\mathrm{o}}, \bar{g}_{\mathrm{o}}\rangle = \langle \mathcal{S} \bar{f}_{\mathrm{o}}, \bar{g}_{\mathrm{o}} \rangle, \\
\langle \bar{f}_{\mathrm{o}}, \mathcal{S}\bar{f}_{\mathrm{o}} \rangle &= \int_{-\infty}^0 \int_{-\infty}^{+\infty} \int_{-\infty}^{+\infty} \xi_n^{-1} [\bar{f}_{\mathrm{o}}(\xi_n, \xi_{t_1}, \xi_{t_2})]^2 \,\mathrm{d} \xi_{t_2} \,\mathrm{d} \xi_{t_1} \,\mathrm{d} \xi_n \leq 0,
\end{split}
\end{displaymath}
which shows both properties of $\mathcal{S}$.
\end{proof}

The equation \eqref{eq:Onsager} with $\mathcal{S}$ defined by \eqref{eq:S} is still incomplete since the condition for mass conservation \eqref{eq:zero_vel} has not been considered to determine $w^{0,W}$ in the definition of $f_W$. Therefore, we further require that $\langle \psi_n^0, \bar{f} \rangle = 0$, or equivalently,
\begin{equation} \label{eq:Puf}
\mathcal{P}_u \bar{f} = 0,
\end{equation}
where $\mathcal{P}_u$ is the projection operator from $\overline{\mathbb{V}}$ onto the following subspace:
\begin{displaymath}
\overline{\mathbb{V}}_u = \operatorname{span} \{ \psi_n^0 \}.
\end{displaymath}
The proposition below shows how we can combine \eqref{eq:Onsager} and \eqref{eq:Puf} into one equation:
\begin{proposition}\label{prop}
Let $\mathcal{S}_{uu}$ be a map from $\overline{\mathbb{V}}_u$ to $\overline{\mathbb{V}}_u$ defined by $\mathcal{S}_{uu} = \mathcal{P}_u \mathcal{S}|_{\overline{\mathbb{V}}_u}$. Then for any $\bar{f}$ satisfying
\begin{equation} \label{eq:finalBC}
\mathcal{P}_{\mathrm{odd}}|_{\overline{\mathbb{V}}} \, \bar{f} = \bar{\mathcal{S}} \bar{\mathcal{A}}_{\mathrm{oe}}(f_W - \mathcal{P}_{\mathrm{even}}|_{\overline{\mathbb{V}}} \,\bar{f})
\end{equation}
with
\begin{displaymath}
\bar{\mathcal{S}} = (\mathcal{I} - \mathcal{P}_u) (\mathcal{I} - \mathcal{S} \mathcal{S}_{uu}^{-1} \mathcal{P}_u) \mathcal{S},
\end{displaymath}
both \eqref{eq:Onsager} and \eqref{eq:Puf} hold for an appropriately chosen $\rho_W$. Meanwhile, the operator $\bar{\mathcal{S}}$ is a self-adjoint and positive semidefinite operator on $\overline{\mathbb{V}}_{\mathrm{odd}}$, and it satisfies
\begin{equation} \label{eq:SA}
\bar{\mathcal{S}} \bar{\mathcal{A}}_{\mathrm{oe}} = \bar{\mathcal{S}} \bar{\mathcal{A}}_{\mathrm{oe}} (\mathcal{I} - \mathcal{P}_{\rho}),
\end{equation}
where $\mathcal{P}_{\rho}: \overline{\mathbb{V}}_{\mathrm{even}} \rightarrow \overline{\mathbb{V}}_{\mathrm{even}}$ is a projection operator defined by $\mathcal{P}_{\rho} \bar{f}_{\mathrm{e}} = \langle \bar{f}_{\mathrm{e}}, \psi^0\rangle \psi^0$.
\end{proposition}

The rigorous proof of this proposition can be found in Section \ref{sm:prop} of supplementary material. This proposition shows that we can use \eqref{eq:finalBC} as the Onsager boundary conditions for the R13 equations, and the property \eqref{eq:SA} implies that $\rho_W$ actually does not appear in \eqref{eq:finalBC} since
\begin{displaymath}
(\mathcal{I} - \mathcal{P}_{\rho}) f_W = w^{1,W} \psi^1 + 3 w_{t_1}^{0,W} \psi_{t_1}^0 + 3 w_{t_2}^{0,W} \psi_{t_2}^0,
\end{displaymath}
which does not involve $\rho_W$. Following the abstract form \eqref{eq:finalBC}, we can write down the boundary conditions explicitly using the series expansion of the distribution functions. Since
\begin{displaymath}
\left\{ \psi_n^0, \phi_n^{(1)}, \phi_{nt_1}^{(1)}, \phi_{nt_2}^{(1)}, \phi_n^{(2)}, \phi_{nt_1}^{(2)}, \phi_{nt_2}^{(2)}, \phi_{nnn}^{(2)}, \phi_{nt_1 t_2}^{(2)},  \phi_{nt_1 t_1}^{(2)} + \frac{1}{2} \phi_{nnn}^{(2)} \right\}
\end{displaymath}
forms an orthogonal basis of $\overline{\mathbb{V}}_{\mathrm{odd}}$ (note that $\phi_{nt_2 t_2}^{(2)} = -\phi_{nnn}^{(2)} - \phi_{nt_1 t_1}^{(2)}$),
ten boundary conditions are to be prescribed for each boundary point. The calculation of the boundary conditions is tedious but straightforward. The results are 
\begin{align}
      \label{eq:bc1} w_n^0 &= 0, \\
      \label{eq:bc2} u^{(1)}_{n} & = \frac{2\chi}{2-\chi} \left[ \lambda'_{11} (w^1 -w^{1,W}) + \lambda'_{12} u^{(1)}_{nn} +  \lambda'_{13} u^{(2)} +  \lambda'_{14} u^{(2)}_{nn}\right], \\ 
    \label{eq:bc3} u^{(2)}_{n} & = \frac{2\chi}{2-\chi} \left[ \lambda'_{21} (w^1 -w^{1,W}) + \lambda'_{22}u^{(1)}_{nn} +  \lambda'_{23} u^{(2)} +  \lambda'_{24} u^{(2)}_{nn}\right], \\
    \label{eq:bc4}   u^{(1)}_{t_i n} & = \frac{2\chi}{2-\chi} \left[\lambda'_{31} (w^0_{t_i} -w^{0,W}_{t_i}) + \lambda'_{32}u^{(1)}_{t_i} + \lambda'_{33} u^{(2)}_{t_i } +  \lambda'_{34}  u^{(2)}_{t_i nn} \right], \\ 
              \label{eq:bc5}  u^{(2)}_{t_in} & = \frac{2\chi}{2-\chi} \left[ \lambda'_{41} (w^0_{t_i} -w^{0,W}_{t_i}) + \lambda'_{42}u^{(1)}_{t_i} + \lambda'_{43} u^{(2)}_{t_i} + \lambda'_{44}  u^{(2)}_{t_i nn} \right], \\ 
 \label{eq:bc6}  u^{(2)}_{nnn} & = \frac{2\chi}{2-\chi} \left[\lambda'_{51} (w^1 -w^{1,W}) + \lambda'_{52}u^{(1)}_{nn} +  \lambda'_{53} u^{(2)} +  \lambda'_{54} u^{(2)}_{nn}\right], \\ 
 \label{eq:bc7} u^{(2)}_{t_i t_i n} + \frac{1}{2} u^{(2)}_{nnn}& =  \frac{2\chi}{2-\chi} \left[ \lambda'_{61}\left(u^{(1)}_{t_i t_i} + \frac{c^{(2),0}_2}{c^{(1),0}_2}u^{(2)}_{t_it_i} \right) + \lambda'_{62}\left(u^{(1)}_{nn} + \frac{c^{(2),0}_2}{c^{(1),0}_2}u^{(2)}_{nn} \right) \right], \\
 \label{eq:bc8}  u^{(2)}_{t_1 t_2 n} & = \frac{2\chi}{2-\chi} \lambda'_{71}\left(u^{(1)}_{t_1 t_2} + \frac{c^{(2),0}_2}{c^{(1),0}_2}u^{(2)}_{t_1t_2} \right).
\end{align}
The detailed derivation and the expressions of the coefficient $\lambda'_{ij}$ are given in Section \ref{app:bc}.

However, these equations give too many boundary conditions for the R13 equations. The general theory of hyperbolic conditions requires that the number of boundary conditions equal the number of negative eigenvalues of $\bar{\mathcal{A}}_n$. In our case, $\bar{\mathcal{A}}_n$ has only 9 negative eigenvalues. In other words, the operator $\bar{\mathcal{A}}_{\mathrm{oe}}$ is not surjective.
As mentioned in Section \ref{sec:intro}, this is because our selection of function spaces $\overline{\mathbb{V}}_{\mathrm{odd}}$ and $\overline{\mathbb{V}}_{\mathrm{even}}$ are purely based on the symmetry of functions, regardless of the structure of equations. Detailed explanations and the fix of these boundary conditions will be shown in the next subsection.

\subsection{Fixing the boundary conditions}
To show why $\bar{\mathcal{A}}_{\mathrm{oe}}$ is not surjective, we multiply the equation \eqref{eq4} by $c_1^{(2),1}/c_1^{(1),1}$, subtract the result from \eqref{eq7}, and set the index $i$ to be the normal direction $n$, yielding the equation
\begin{displaymath}
\begin{split}
\left(A_{57} - \frac{c_1^{(2),1}}{c_1^{(1),1}} A_{45}\right) \Bigg[ & \frac{\partial}{\partial x_n} \left(u_{nn}^{(1)} + \frac{c_2^{(2),0}}{c_2^{(1),0}} u_{nn}^{(2)} \right) + \frac{\partial}{\partial x_{t_1}} \left(u_{nt_1}^{(1)} + \frac{c_2^{(2),0}}{c_2^{(1),0}} u_{nt_1}^{(2)} \right) \\
& + \frac{\partial}{\partial x_{t_2}} \left(u_{nt_2}^{(1)} + \frac{c_2^{(2),0}}{c_2^{(1),0}} u_{nt_2}^{(2)} \right) \Bigg] 
= \frac{1}{\kn} \left(\mathscr{L}_1^{(22)} - \frac{c_1^{(2),1}}{c_1^{(1),1}} \mathscr{L}_1^{(12)}\right) u_n^{(2)}.
\end{split}
\end{displaymath}
Then we set $i,j,k$ to be $n$ in \eqref{eq9}: 
\begin{displaymath}
\begin{split}
\frac{3}{5} A_{59} \frac{\partial}{\partial x_n} \left(u_{nn}^{(1)} + \frac{c_2^{(2),0}}{c_2^{(1),0}} u_{nn}^{(2)} \right) - \frac{2}{5} A_{59} \frac{\partial}{\partial x_{t_1}} \left(u_{nt_1}^{(1)} + \frac{c_2^{(2),0}}{c_2^{(1),0}} u_{nt_1}^{(2)} \right) - \frac{2}{5} A_{59} \frac{\partial}{\partial x_{t_2}} \left(u_{nt_2}^{(1)} + \frac{c_2^{(2),0}}{c_2^{(1),0}} u_{nt_2}^{(2)} \right) \\
= \frac{1}{\kn} \mathscr{L}_3^{(22)} u_{nnn}^{(2)}.
\end{split}
\end{displaymath}
The two equations above show that if we perform the linear combination
\begin{displaymath}
\frac{3}{5} A_{59} \times \left[\eqref{eq7}\big\vert_{i=n} - \frac{c_1^{(2),1}}{c_1^{(1),1}} \times \eqref{eq4}\big\vert_{i=n} \right] - \left(A_{57} - \frac{c_1^{(2),1}}{c_1^{(1),1}} A_{45} \right) \times \eqref{eq9}\big\vert_{i,j,k=n},
\end{displaymath}
the derivatives with respect to $x_n$ will all cancel out on the left-hand side. An equivalent statement is
\begin{displaymath}
\left\langle
  \phi_0, \bar{\mathcal{A}}_{\mathrm{oe}} \bar{f}
\right\rangle = 0, \qquad \forall \bar{f} \in \overline{\mathbb{V}},
\end{displaymath}
where
\begin{equation} \label{eq:phi0}
  \phi_0 = \frac{9}{5} A_{59} \left(\phi_n^{(2)} - \frac{c_1^{(2),1}}{c_1^{(1),1}} \phi_n^{(1)} \right) - \frac{35}{2} \left(A_{57} - \frac{c_1^{(2),1}}{c_1^{(1),1}} A_{45} \right) \phi_{nnn}^{(2)}.
\end{equation}
Let $\overline{\mathbb{V}}_0 = \operatorname{span} \{\phi_0\}$ and $\mathcal{P}_0$ be the associated projection operator. We then have
\begin{equation} \label{eq:P0Aoe}
\mathcal{P}_0 \bar{\mathcal{A}}_{\mathrm{oe}} = 0.
\end{equation}

Therefore, to fix the boundary conditions, the subspace $\overline{\mathbb{V}}_0$ should be removed form $\overline{\mathbb{V}}_{\mathrm{odd}}$. This requires replacing the operator $\mathcal{S}$ with $(\mathcal{I} - \mathcal{P}_0) \mathcal{S} (\mathcal{I} - \mathcal{P}_0)$. Thus, the operator $\bar{\mathcal{S}}$ should be replaced with
\begin{displaymath}
\begin{split}
& (\mathcal{I} - \mathcal{P}_u) [\mathcal{I} - (\mathcal{I} - \mathcal{P}_0)\mathcal{S}(\mathcal{I} - \mathcal{P}_0) \mathcal{S}_{uu}^{-1} \mathcal{P}_u] (\mathcal{I} - \mathcal{P}_0)\mathcal{S}(\mathcal{I} - \mathcal{P}_0) \\
={} & (\mathcal{I} - \mathcal{P}_u)(\mathcal{I} - \mathcal{P}_0) [\mathcal{I} - \mathcal{S}(\mathcal{I} - \mathcal{P}_0) \mathcal{S}_{uu}^{-1} \mathcal{P}_u(\mathcal{I} - \mathcal{P}_0)] \mathcal{S}(\mathcal{I} - \mathcal{P}_0) \\
={} & (\mathcal{I} - \mathcal{P}_0)(\mathcal{I} - \mathcal{P}_u) (\mathcal{I} - \mathcal{S} \mathcal{S}_{uu}^{-1} \mathcal{P}_u) \mathcal{S}(\mathcal{I} - \mathcal{P}_0) \hspace{50pt} [\text{Since } \mathcal{P}_0 \mathcal{P}_u = \mathcal{P}_u \mathcal{P}_0 = 0]\\
={} & (\mathcal{I} - \mathcal{P}_0) \bar{\mathcal{S}}(\mathcal{I} - \mathcal{P}_0).
\end{split}
\end{displaymath}
According to \eqref{eq:finalBC}, the final boundary conditions should be
\begin{displaymath}
(\mathcal{P}_{\mathrm{odd}}|_{\overline{\mathbb{V}}} - \mathcal{P}_0) \bar{f} = (\mathcal{I} - \mathcal{P}_0) \bar{\mathcal{S}} (\mathcal{I} - \mathcal{P}_0) \bar{\mathcal{A}}_{\mathrm{oe}}(f_W - \mathcal{P}_{\mathrm{even}}|_{\overline{\mathbb{V}}} \,\bar{f}),
\end{displaymath}
which can be further simplified to
\begin{equation} \label{eq:newBC}
(\mathcal{P}_{\mathrm{odd}}|_{\overline{\mathbb{V}}} - \mathcal{P}_0)  \bar{f} = (\mathcal{I} - \mathcal{P}_0) \bar{\mathcal{S}} \bar{\mathcal{A}}_{\mathrm{oe}}(f_W - \mathcal{P}_{\mathrm{even}}|_{\overline{\mathbb{V}}} \,\bar{f})
\end{equation}
due to \eqref{eq:P0Aoe}.

The equation above indicates that the explicit form of the new boundary conditions \eqref{eq:newBC} can be derived by linear combinations of the boundary conditions \eqref{eq:bc1}--\eqref{eq:bc8}. Using the expression of $\phi_0$ \eqref{eq:phi0}, we can find the following orthogonal basis of the orthogonal complement of $\overline{\mathbb{V}}_0$ in $\overline{\mathbb{V}}$:
\begin{equation} \label{eq:orth basis}
\left\{ \psi_n^0, \phi_{nt_1}^{(1)}, \phi_{nt_2}^{(1)}, \phi_{nt_1}^{(2)}, \phi_{nt_2}^{(2)}, \phi_{nt_1 t_2}^{(2)},  \phi_{nt_1 t_1}^{(2)} + \frac{1}{2} \phi_{nnn}^{(2)},  \phi_n^{(1)} + \frac{c^{(2),1}_1}{c^{(1),1}_1} \phi_n^{(2)},\mu_1 \phi_{nnn}^{(2)} + \mu_2 \left( \phi_n^{(2)} - \frac{c^{(2),1}_1}{c^{(1),1}_1} \phi_n^{(1)} \right)\right\},
\end{equation}
where 
\begin{equation}\label{mu}
\mu_1 = 3 A_{59} \left[ 1 + \left( \frac{c_1^{(2),1}}{c_1^{(1),1}} \right)^2 \right], \qquad
\mu_2 = 2\left( A_{57} - \frac{c_1^{(2),1}}{c_1^{(1),1}} A_{45} \right).
\end{equation}
Therefore, in the new boundary conditions, the equations \eqref{eq:bc1}\eqref{eq:bc4}\eqref{eq:bc5}\eqref{eq:bc7}\eqref{eq:bc8} are preserved, and the two boundary conditions associated with the last two basis functions in \eqref{eq:orth basis} are
\begin{align*}
  u^{(1)}_n + \frac{c_1^{(2),1}}{c_1^{(1),1}} u^{(2)}_n  & = \frac{2\chi}{2-\chi} \left[ \kappa_{11} (w^1 -w^{1,W}) + \kappa_{12} u^{(1)}_{nn} +  \kappa_{13} u^{(2)} +  \kappa_{14} u^{(2)}_{nn}\right], \\ 
   \mu_1 u_{nnn}^{(2)} + \mu_2 u^{(2)}_n - \frac{c_1^{(2),1}}{c_1^{(1),1}} \mu_2 u^{(1)}_n  &= \frac{2\chi}{2-\chi} \left[\kappa_{21} (w^1 -w^{1,W}) + \kappa_{22} u^{(1)}_{nn} +  \kappa_{23} u^{(2)} +  \kappa_{24} u^{(2)}_{nn}\right],
\end{align*}
where
\begin{equation}\label{kappa}
\kappa_{1j} = \lambda'_{1j} + \frac{c_1^{(2),1}}{c_1^{(1),1}} \lambda'_{2j}, \quad
\kappa_{2j} = \mu_1 \lambda'_{5j} + \mu_2 \left( \lambda'_{2j} - \frac{c_1^{(2),1}}{c_1^{(1),1}} \lambda'_{1j} \right), \qquad j = 1,2,3,4.
\end{equation}
Our final boundary conditions presented in Section \ref{sec:obc} are obtained by rewriting these boundary conditions using the physical quantities $\rho, \theta, v_i, q_i$ and $\sigma_{ij}$, whose derivation needs \eqref{eq:wij}--\eqref{eq:u2} and \eqref{physics quantities} for conversion.

\begin{remark}
For Maxwell molecules, due to the nonexistence of the variable $u_n^{(2)}$, the boundary condition \eqref{eq:bc3} does not exist. Therefore, such a fix is unnecessary.
\end{remark}

\section{Results of one-dimensional channel flows}
 \label{sec:results 1d}
Compared with the boundary conditions of R13 equations proposed in \cite{Hu2020}, our boundary conditions have a nicer structure that agrees with the form \eqref{eq:stableBC}. However, it is unclear whether the acquirement of this structure will harm the accuracy of the model. In this section, we will reuse the examples for one-dimensional channel flows in \cite{Hu2020} to test our models.

\subsection{Problem settings and reduced moment equations}
We consider the gas flow between two infinitely large parallel plates in steady state (see Figure \ref{fig:flow}). The distance between the two plates is $L$, and both plates are perpendicular to the $x_2$-axis. The temperatures of the left and right plates are respectively given by $\theta_W^l$ and $\theta_W^r$. Both plates can move inside their own plane, and we choose the reference frame and the coordinates such that both velocities are parallel to the $x_1$-axis. Under such settings, all the moments are functions of $x_2$ only. Due to mass conservation, $v_2=0$. In addition, we have the symmetry $f(\xi_1,\xi_2,\xi_3) = f(\xi_1,\xi_2,-\xi_3)$ since the plates only move along $x_1$, and thus all the moments that are odd in $\xi_3$ vanish. In specific, we have 
\begin{displaymath}
    v_3 = q_3 = \sigma_{13} = \sigma_{23} = 0.
\end{displaymath}
Furthermore, since $\sigma_{ij}$ is trace-free, $\sigma_{33}$ is automatically obtained once $\sigma_{11}$ and $\sigma_{22}$ are known. Therefore, the 13 moments in our problem can be reduced to eight variables including
\begin{itemize}
\item Equilibrium variables including density $\rho$, temperature $\theta$, and the velocity component parallel to the plates $v_1$;
\item Components of the stress tensor including the parallel stress $\sigma_{11}$, normal stress $\sigma_{22}$ and the shear stress $\sigma_{12}$;
\item Heat fluxes including the parallel heat flux $q_1$ and the normal heat flux $q_2$.
\end{itemize}

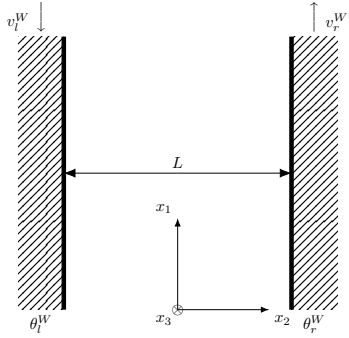
\begin{figure}[h]
\centering \resizebox{.3\hsize}{!}{ $
\begin{tikzpicture}

\draw[-{Latex[scale=1.5]}] (2.5,0) -- (5,0);
\draw[-{Latex[scale=1.5]}] (2.5,0) -- (0,0);
\node[above] at (2.5,0) {$L$};

\draw [line width=1mm] (0,-3) -- (0,3);\draw [line width=1mm] (5,-3) -- (5,3);
\fill [pattern=north east lines] (0,-3) rectangle (-1,3);
\fill [pattern=north east lines] (5,-3) rectangle (6,3);
\node[below] at (-0.5,-3) {$\theta_l^W$};\node[below] at (5.5,-3) {$\theta_r^W$};
\node[above] at (-0.5,3) {$\Big\downarrow$};\node[above] at (5.5,3) {$\Big\uparrow$};
\node[above] at (-1,3) {$v_l^W$};\node[above] at (6,3) {$v_r^W$};

\node at (2.5,-3) {$\otimes$}; 
\draw[-{Latex[scale=1]}] (2.5,-3) -- (4.5,-3);
\draw[-{Latex[scale=1]}] (2.5,-3) -- (2.5,-1);
\node[below left] at (2.5,-3) {$x_3$};
\node[below right] at (4.5,-3) {$x_2$};
\node[above left] at (2.5,-1) {$x_1$};


\end{tikzpicture} $}
\caption{The two parallel solid walls of the channel are moving along $x_1$ directions, and they may have different temperatures.}
  \label{fig:flow}
\end{figure}
 
The 13-moment equations \eqref{r13eq:density}--\eqref{r13eq:stress tensor} for the one-dimensional channel flows can be immediately obtained by dropping all partial derivatives with respective to $x_1$ and $x_3$. As a result, the conservation laws become
\begin{equation}
       \frac{\dd \sigma_{12}}{\dd x_2} = 0, \qquad
       \frac{\dd \rho}{\dd x_2} +  \frac{\dd \theta}{\dd x_2} +  \frac{\dd \sigma_{22}}{\dd x_2} = 0, \qquad
       \frac{\dd q_2}{\dd x_2} = 0. \label{1dr13eq_theta}
\end{equation}
The equation of stress tensor $\sigma_{11}$, $\sigma_{22}$ and $\sigma_{12}$ are 
\begin{gather}
    \frac{1}{\kn} \mathscr{L}^{(11)}_{2} \sigma_{11} + 5 (c^{(1),0}_2)^2 \beta_4  \frac{\dd q_2}{\dd x_2} - \kn \frac{c^{(1),0}_2}{3}\beta_3 \frac{\dd^2 \sigma_{11}}{\dd x^2_2} + \kn \frac{c^{(1),0}_2}{3}(\beta_1 + \frac{2}{5}\beta_3)  \frac{\dd^2 \sigma_{22}}{\dd x^2_2} = 0,\label{1dr13eq_sigma11}\\
     \frac{1}{\kn} \mathscr{L}_{2}^{(11)} \sigma_{22} - 10 (c^{(1),0}_2)^2 \beta_4 \frac{\dd q_2}{\dd x_2}- \kn c^{(1),0}_2 (\frac{2}{3}\beta_1 + \frac{3}{5}\beta_3)  \frac{\dd^2 \sigma_{22}}{\dd x^2_2} = 0,\label{1dr13eq_sigma22}\\
       \frac{1}{\kn} \mathscr{L}^{(11)}_2 \sigma_{12} - \frac{15}{2}(c^{(1),0}_2)^2\frac{\dd v_1}{\dd x_2} - \frac{15}{2} (c^{(1),0}_2)^2 \beta_4 \frac{\dd q_1}{\dd x_2}- \kn c^{(1),0}_2 (\frac{1}{2}\beta_1 + \frac{8}{15}\beta_3)  \frac{\dd^2 \sigma_{12}}{\dd x^2_2} = 0.\label{1dr13eq_sigma12}
\end{gather}
The equations of heat flux $q_1$ and $q_2$ are
\begin{gather}
      \frac{1}{\kn} \mathscr{L}^{(11)}_1 q_1 - \frac{15}{2} (c^{(1),1}_1)^2 \beta_4  \frac{\dd \sigma_{12}}{\dd x_2} - \kn \frac{c^{(1),1}_1}{2} \beta_2 \frac{\dd^2 q_1}{\dd x^2_2} = 0,\label{1dr13eq_q1}\\
      \frac{1}{\kn}  \mathscr{L}^{(11)}_1 q_2 - \frac{15}{2} (c^{(1),1}_1)^2 \beta_4 \frac{\dd \sigma_{22}}{\dd x_2} - \kn c^{(1),1}_1 (\beta_0  + \frac{2}{3} \beta_2) \frac{\dd^2 q_2}{\dd x^2_2} - \frac{15}{2}(c^{(1),1}_1)^2 \frac{\dd \theta}{\dd x_2} = 0. \label{1dr13eq_q2}
\end{gather}

The general solution of this linear system can be found analytically. By the conservation laws, we see that $\sigma_{12}$, $\rho + \theta + \sigma_{22}$ and $q_2$ are all constants. Then we can solve $q_1$ and $v_1$ from \eqref{1dr13eq_sigma12}\eqref{1dr13eq_q1} and solve $q_2$ and $\theta$ from \eqref{1dr13eq_sigma22}\eqref{1dr13eq_q2}.
Afterwards, $\rho$ and $\sigma_{11}$ can be immediately obtained by solving \eqref{1dr13eq_sigma11} once other quantities are known. The general solution will include 11 constants to be determined. One of the constants depends on the average density between the plates: upon setting the range of $x_2$ to be $[-1/2, 1/2]$ (so that $L = 1$), we assign the average density as  
\begin{equation} \label{eq:zero_mass}
\int_{-1/2}^{1/2} \rho(x_2) \,\mathrm{d}x_2 = 0.
\end{equation}
The other 10 constants will be determined by 10 boundary conditions, of which each boundary has five given by \eqref{bc phys 2}--\eqref{bc phys 6} in the one-dimensional setting. Below we provide the boundary conditions on the right wall:
\begin{align}
    q_2 = & \  \frac{\chi_r}{2 - \chi_r} \left[ 2\lambda_{11} (\theta - \theta^W_r) + 2\lambda_{12}\sigma_{22} + \kn(2\lambda_{13} + \frac{4}{3}\lambda_{14}) \frac{\dd q_2}{\dd x_2}  \right], \label{1dbc1}\\
    \sigma_{12} = & \  \frac{\chi_r}{2 - \chi_r}\left[  2\lambda_{21}(v_{1} - v^W_{r}) + 2\lambda_{22} q_{1} + \kn(2\lambda_{23} + \frac{16}{15}\lambda_{24} )\frac{\dd \sigma_{12}}{\dd x_2} \right],  \label{1dbc2} \\
    \kn \frac{\dd q_{1}}{\dd x_2} = & \  \frac{\chi_r}{2 - \chi_r} \left[ 4\lambda_{31} (v_{1} - v^W_{r}) + 4\lambda_{32}q_{1} + \kn(4\lambda_{33} + \frac{32}{15}\lambda_{34} )\frac{\dd \sigma_{1 2 }}{\dd x_2}  \right],  \label{1dbc3} \\
     (\lambda_{45}+ \frac{3}{5})\kn \frac{\dd \sigma_{22}}{\dd x_2} = & \  \frac{\chi_r}{2 - \chi_r} \left[ 2\lambda_{41} (\theta - \theta^W_r) + 2\lambda_{42}\sigma_{22} + \kn(2\lambda_{43} + \frac{4}{3}\lambda_{44}) \frac{\dd q_2}{\dd x_2}  \right], \label{1dbc4}\\
     \kn\left(2 \frac{\dd  \sigma_{ 11}}{\dd x_{2 }  } +  \frac{\dd \sigma_{ 22}}{\dd x_{2 }  }\right) = & \   \frac{\chi_r}{2-\chi_r}\left(12\lambda_{51} \sigma_{11} + 12\lambda_{52} \sigma_{22}\right). \label{1dbc5}
\end{align}
For the boundary conditions on the left solid wall, one only needs to make the following replacement of parameters:
\begin{displaymath}
q_2 \rightarrow -q_2, \quad \sigma_{12} \rightarrow -\sigma_{12}, \quad x_2 \rightarrow -x_2, \quad \chi_r \rightarrow \chi_l, \quad \theta^W_r \rightarrow \theta^W_l, \quad v^W_r \rightarrow v^W_l.
\end{displaymath}
In the next two subsections, two special cases will be considered to verify our model.

\subsection{Results}
In this section, we illustrate the analytic solutions of R13 equations \eqref{1dr13eq_theta}--\eqref{1dr13eq_q2} with the Onsager boundary conditions \eqref{1dbc1}--\eqref{1dbc5} using the examples of one-dimensional Couette and Fourier flows in \cite{Hu2020}. The results will be compared with reference solutions by DSMC simulations which are obtained by Bird's code \cite{Bird1994}.  

Following the examples in \cite{Hu2020}, we consider the inverse-power-law model which assumes that the force $F$ between two molecules is proportional to an inverse power of the distance $r$ between them as $F = \kappa r^{-\eta}$, where $\eta$ and $\kappa$ are positive parameters. After nondimensionalization as done in \cite{Hu2020}, the parameter $\kappa$ will be integrated into the Knudsen number $\kn$, and we choose $\kn=0.05$ and $0.1$ in our tests. For the parameter $\eta$, we take three values $\eta =5,10$ and $\infty$. In particular, when $\eta =5$, the inverse-power-law model becomes Maxwell molecules and the corresponding moment equations are identical to those derived in \cite{Taheri2009b}; the choice $\eta = 10$ is often used in the simulation of the argon gas; when $\eta \rightarrow \infty$, the model reduces to the hard-sphere model \cite{Struchtrup2013}. Different choices of $\eta$ will result in different values of $a_{lmn}$ in the collision operator \eqref{eq:rot_inv}. Thus, the coefficients appearing in the moment equations such as $\mathscr{L}^{(mn)}_l$ and $\beta_{i}$ as well as $\lambda_{ij}$ in the boundary conditions all rely on the choice of $\eta$. We list the numerical values of these coefficients in the supplementary material for reference. In addition, both plates are assumed to be completely diffusive, that is, the accommodation coefficients are $\chi_l = \chi_r =1$. 

The results are plotted in Figure \ref{fig:couette} and \ref{fig:fourier} for the Couette flow and the Fourier flow. The horizontal axis in each subfigure is $x_2$, and the vertical axis is the value of the variables we are interested in for the one-dimensional channel problems with Knudsen numbers $\kn =0.05$ and $0.1$. The R13 results are plotted as the blue, yellow and green solid lines for $\eta =5,10$ and $\infty$, and the corresponding reference solutions by DSMC are given by the dotted lines with the same colors. Detailed descriptions of these figures are given in the following subsections.

\subsubsection{Results of the Couette flow}
In the planar Couette flow, the two plates have the same temperature and move in opposite directions (see Figure \ref{fig:flow}). In our test, we select $\theta_l^W = \theta_r^W = 0$ and $v_l^W = -v_r^W = 0.2$. In Figure \ref{fig:couette}, we plot the results of $v_1$, $\sigma_{12}$ and $q_1$ for Couette flow. As can be seen, the solid curves by R13 solutions generally agree with the dotted curves by the DSMC reference solutions, and they match better for smaller $\kn$, which implies the validity of our models. Note that due to the lack of nonlinearity, the other five variables ($\rho$, $\theta$, $\sigma_{11}$, $\sigma_{22}$ and $q_2$) are all zero in the analytical solutions.
\begin{figure}[!htb]
  \centering
  {\includegraphics[width=\textwidth]{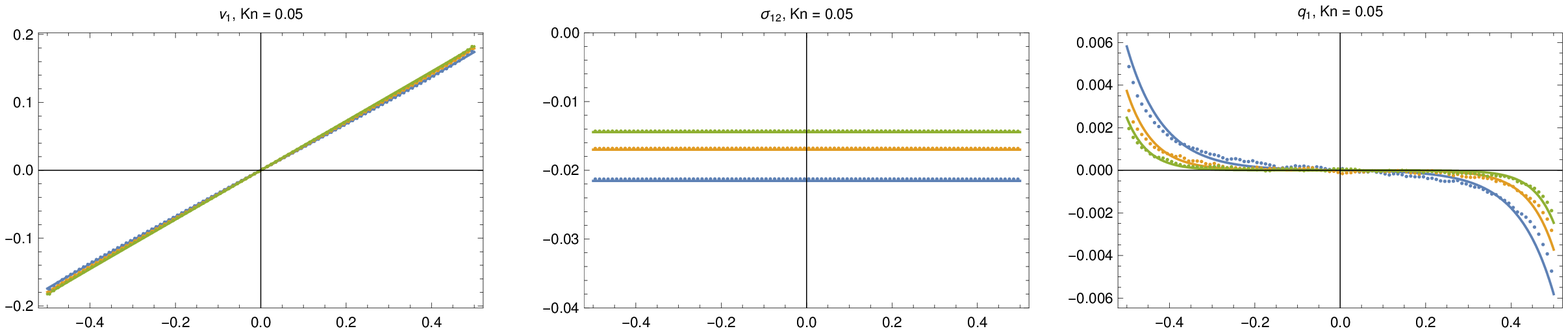}}\\[10pt]
  {\includegraphics[width=\textwidth]{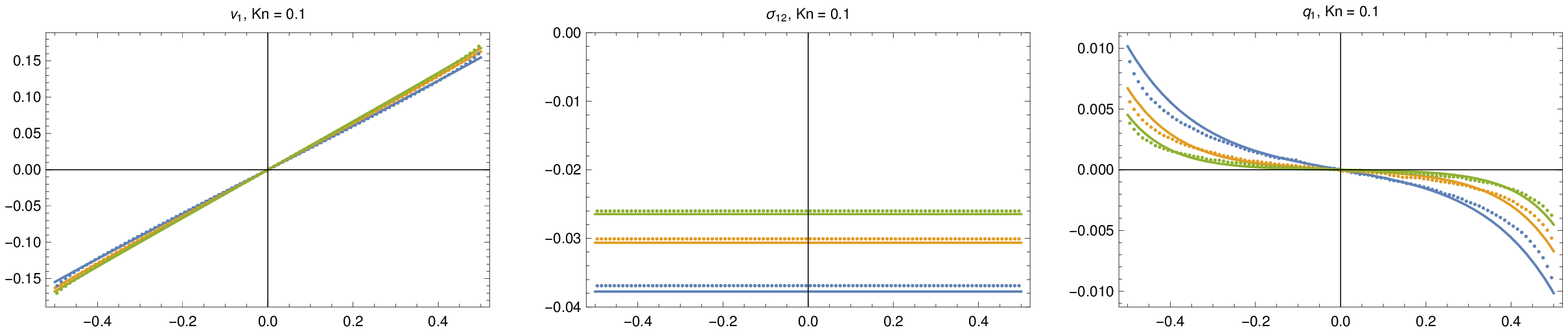}}
  \caption{Results of the Couette flow.}
   \label{fig:couette}
\end{figure}

\subsubsection{Results of the Fourier flow}
For the Fourier flow, the two plates are stationary, i.e, $v^W_r = v^W_l = 0$, and the flow is driven solely by the difference of plate temperatures. In our simulations, we set $\theta^W_l = 0$ and $\theta^W_r =0.2$. Due to the simplicity of the structure, the Fourier flow relies less on the nonlinear contributions and we may look at some other variables that we have not shown in Figure \ref{fig:couette} for Couette flow. In Figure \ref{fig:flow}, we plot the results for $\theta$ , $\sigma_{22}$ and $q_2$. Again, the R13 solutions (solid lines) provide qualitatively correct results as expected, showing that our linear R13 model works for the Fourier flow as well. In general, despite the enforcement of a symmetric structure, our models provide almost equal qualities as the equations in \cite{Hu2020}.

\begin{figure}[!htb]
  \centering
  {\includegraphics[width=\textwidth]{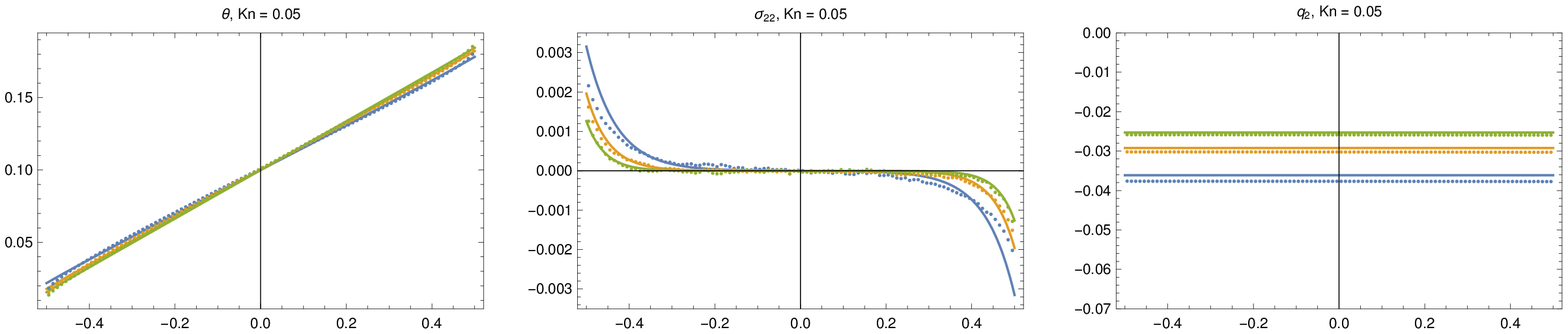}}\\[10pt]
  {\includegraphics[width=\textwidth]{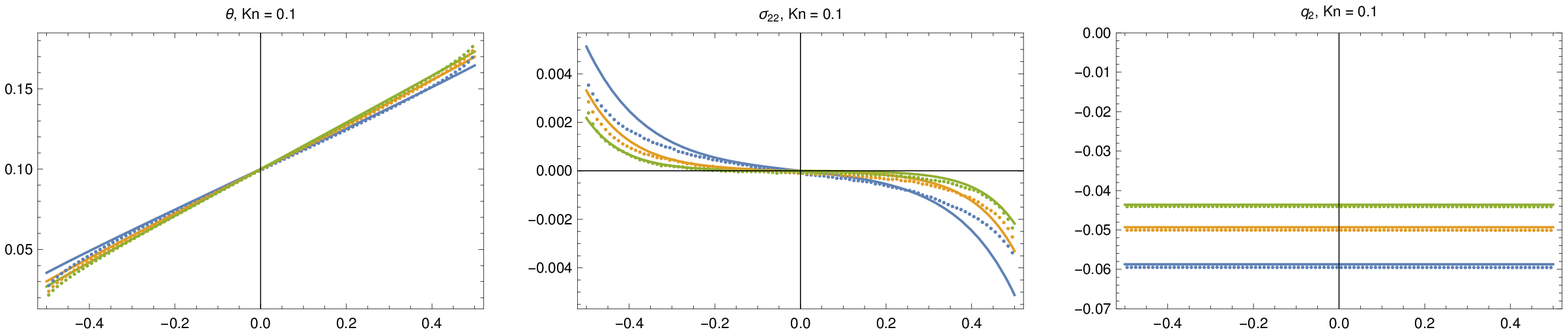}}
  \caption{Results of the Fourier flow.}
   \label{fig:fourier}
\end{figure}

\section{Conclusion}
\label{sec:conclusion}
This paper performs a re-derivation of the linear steady-state regularized 13-moment equations of general gas molecules from a novel point of view. Our derivation is based on a decomposition of the function space according to the order of accuracy, which provides a clearer picture of the entire procedure and gains a symmetric structure in the final system. This also allows a straightforward approach to obtaining Onsager boundary conditions. As an ongoing work, we are trying to apply this approach to time-dependent equations. Extensions to other kinetic equations such as the radiative transfer equation will also be considered in our future work.

\bibliographystyle{plain}
\bibliography{article}

\begin{thebibliography}{10}

\bibitem{Beckmann2018}
Alexander~Felix Beckmann, Anirudh~Singh Rana, Manuel Torrilhon, and Henning
  Struchtrup.
\newblock Evaporation boundary conditions for the linear r13 equations based on
  the {O}nsager theory.
\newblock {\em Entropy}, 20(9):680, 2018.

\bibitem{Bird1994}
G.~A. Bird.
\newblock {\em Molecular Gas Dynamics and the Direct Simulation of Gas Flows}.
\newblock Oxford: Clarendon Press, 1994.

\bibitem{Bohmer2020}
Niclas B{\"o}hmer and Manuel Torrilhon.
\newblock Entropic quadrature for moment approximations of the
  {B}oltzmann-{BGK} equation.
\newblock {\em Journal of Computational Physics}, 401:108992, 2020.

\bibitem{Bunger2023}
Jonas B{\"u}nger, Edilbert Christhuraj, Andrea Hanke, and Manuel Torrilhon.
\newblock Structured derivation of moment equations and stable boundary
  conditions with an introduction to symmetric, trace-free tensors.
\newblock {\em Kinetic and Related Models}, 16(3):458--494, 2023.

\bibitem{Cai2014}
Z.~Cai, Y.~Fan, and R.~Li.
\newblock On hyperbolicity of 13-moment system.
\newblock {\em Kinet. Relat. Mod.}, 7(3):415--432, 2014.

\bibitem{Cai2015}
Z.~Cai and M.~Torrilhon.
\newblock Approximation of the linearized {B}oltzmann collision operator for
  hard-sphere and inverse-power-law models.
\newblock {\em J. Comput. Phys.}, 295:617--643, 2015.

\bibitem{Cai2020regularized}
Z.~Cai and Y.~Wang.
\newblock Regularized 13-moment equations for inverse power law models.
\newblock {\em J. Fluid Mech.}, 894:A12, 2020.

\bibitem{Cai2019}
Zhenning Cai and Manuel Torrilhon.
\newblock On the {H}olway-{W}eiss debate: Convergence of the
  {G}rad-moment-expansion in kinetic gas theory.
\newblock {\em Physics of Fluids}, 31(12):126105, 2019.

\bibitem{Claydon2017}
Rory Claydon, Abhay Shrestha, Anirudh~S. Rana, James~E. Sprittles, and
  Duncan~A. Lockerby.
\newblock Fundamental solutions to the regularised 13-moment equations:
  efficient computation of three-dimensional kinetic effects.
\newblock {\em Journal of Fluid Mechanics}, 833:R4, 2017.

\bibitem{DeFraja2022}
Thomas~C. De~Fraja, Anirudh~S. Rana, Ryan Enright, Laura~J. Cooper, Duncan~A.
  Lockerby, and James~E. Sprittles.
\newblock Efficient moment method for modeling nanoporous evaporation.
\newblock {\em Phys. Rev. Fluids}, 7(2):024201, 2022.

\bibitem{Dimarco2018}
Giacomo Dimarco, Rapha{\"e}l Loub{\`e}re, Jacek Narski, and Thomas Rey.
\newblock An efficient numerical method for solving the {B}oltzmann equation in
  multidimensions.
\newblock {\em Journal of Computational Physics}, 353:46--81, 2018.

\bibitem{Fox2022}
Rodney~O. Fox and Fr{\'e}d{\'e}rique Laurent.
\newblock Hyperbolic quadrature method of moments for the one-dimensional
  kinetic equation.
\newblock {\em SIAM Journal on Applied Mathematics}, 82(2):750--771, 2022.

\bibitem{Grad1949}
H.~Grad.
\newblock On the kinetic theory of rarefied gases.
\newblock {\em Comm. Pure Appl. Math.}, 2(4):331--407, 1949.

\bibitem{Hu2020fast}
J.~Hu and K.~Qi.
\newblock A fast {F}ourier spectral method for the homogeneous {B}oltzmann
  equation with non-cutoff collision kernels.
\newblock {\em J. Comput. Phys.}, 423:109806, 2020.

\bibitem{Hu2020}
Z.~Hu, S.~Yang, and Z.~Cai.
\newblock Flows between parallel plates: Analytical solutions of regularized
  13-moment equations for inverse-power-law models.
\newblock {\em Phys. Fluids}, 32(12):122007, 2020.

\bibitem{Ivanov2013}
I.~E. Ivanov, I.~A. Kryukov, and M.~Yu. Timokhin.
\newblock Application of moment equations to the mathematical simulation of gas
  microflows.
\newblock {\em Comput. Math. Math. phys.}, 53:1534--1550, 2013.

\bibitem{Liu2020}
Chang Liu, Yajun Zhu, and Kun Xu.
\newblock Unified gas-kinetic wave-particle methods {I}: Continuum and rarefied
  gas flow.
\newblock {\em Journal of Computational Physics}, 401:108977, 2020.

\bibitem{Maxwell1878}
J.~C. Maxwell.
\newblock On stresses in rarefied gases arising from inequalities of
  temperature.
\newblock {\em Proc. R. Soc. Lond.}, 27(185--189):304--308, 1878.

\bibitem{Muller}
I.~M{\"u}ller and T.~Ruggeri.
\newblock {\em Rational Extended Thermodynamics, Second Edition}, volume~37 of
  {\em Springer tracts in natural philosophy}.
\newblock Springer-Verlag, New York, 1998.

\bibitem{Pareschi2022}
Lorenzo Pareschi and Thomas Rey.
\newblock Moment preserving {F}ourier-{G}alerkin spectral methods and
  application to the {B}oltzmann equation.
\newblock {\em SIAM Journal on Numerical Analysis}, 60(6):3216--3240, 2022.

\bibitem{Pichard2022}
Teddy Pichard.
\newblock A moment closure based on a projection on the boundary of the
  realizability domain: Extension and analysis.
\newblock {\em Kinetic and Related Models}, 15(5):793--822, 2022.

\bibitem{Rana2018}
A.~S. Rana, V.~K. Gupta, and H.~Struchtrup.
\newblock Coupled constitutive relations: a second law based higher-order
  closure for hydrodynamics.
\newblock {\em Proc. Roy. Soc. A}, 474:20180323, 2018.

\bibitem{Sarna2017}
N.~Sarna and M.~Torrilhon.
\newblock On stable wall boundary conditions for the {H}ermite discretization
  of the linearised {B}oltzmann equation.
\newblock {\em J. Stat. Phys.}, 170:101--126, 2018.

\bibitem{Sarna2021}
Neeraj Sarna.
\newblock A positive and stable {L2}-minimization based moment method for the
  boltzmann equation of gas dynamics.
\newblock {\em Journal of Computational Physics}, 440:110428, 2021.

\bibitem{Struchtrup2004}
H.~Struchtrup.
\newblock Stable transport equations for rarefied gases at high orders in the
  {K}nudsen number.
\newblock {\em Phys. Fluids}, 16(11):3921--3934, 2004.

\bibitem{Struchtrup2003}
H.~Struchtrup and M.~Torrilhon.
\newblock Regularization of {G}rad's 13 moment equations: Derivation and linear
  analysis.
\newblock {\em Phys. Fluids}, 15(9):2668--2680, 2003.

\bibitem{Struchtrup2013}
H.~Struchtrup and M.~Torrilhon.
\newblock Regularized 13 moment equations for hard sphere molecules: Linear
  bulk equations.
\newblock {\em Phys. Fluids}, 25:052001, 2013.

\bibitem{Struchtrup2005macroscopic}
Henning Struchtrup.
\newblock Macroscopic transport equations for rarefied gas flows.
\newblock In {\em Macroscopic transport equations for rarefied gas flows},
  pages 145--160. Springer, 2005.

\bibitem{Su2020}
Wei Su, Lianhua Zhu, and Lei Wu.
\newblock Fast convergence and asymptotic preserving of the general synthetic
  iterative scheme.
\newblock {\em SIAM Journal on Scientific Computing}, 42(6):B1517--B1540, 2020.

\bibitem{Taheri2009a}
Peyman Taheri and Henning Struchtrup.
\newblock Effects of rarefaction in microflows between coaxial cylinders.
\newblock {\em Phys. Rev. E}, 80(6):066317, 2009.

\bibitem{Taheri2009b}
Peyman Taheri, Manuel Torrilhon, and Henning Struchtrup.
\newblock {C}ouette and {P}oiseuille microflows: Analytical solutions for
  regularized 13-moment equations.
\newblock {\em Phys. Fluids}, 21(1):017102, 2009.

\bibitem{Theisen2021}
Lambert Theisen and Manuel Torrilhon.
\newblock {FenicsR13}: A tensorial mixed finite element solver for the linear
  r13 equations using the {FEniCS} computing platform.
\newblock {\em ACM Trans. Math. Softw.}, 47(2), 2021.

\bibitem{Timokhin2017}
M.~Yu. Timokhin, H.~Struchtrup, A.~A. Kokhanchik, and Ye.~A. Bondar.
\newblock Different variants of {R13} moment equations applied to the
  shock-wave structure.
\newblock {\em Physics of Fluids}, 29(3):037105, 2017.

\bibitem{Torrilhon2018}
M.~Torrilhon.
\newblock Convergence study of moment approximations for boundary value
  problems of the {B}oltzmann-{BGK} equation.
\newblock {\em Comm. Comput. Phys.}, 18(3):529--557, 2018.

\bibitem{Torrilhon2017}
Manuel Torrilhon and Neeraj Sarna.
\newblock Hierarchical boltzmann simulations and model error estimation.
\newblock {\em Journal of Computational Physics}, 342:66--84, 2017.

\end{thebibliography}

\pagebreak

\begin{center}
\textbf{\large Supplementary materials: Linear regularized 13-moment equations with Onsager boundary conditions for general gas molecules}
\end{center}

\let\theequationWithoutS\theequation 
\renewcommand\theequation{SM-\theequationWithoutS}
\let\thefigureWithoutS\thefigure 
\renewcommand\thefigure{SM-\thefigureWithoutS}

\renewcommand{\thesection}{SM\arabic{section}}

\section{Derivation of moment equations \eqref{eq:mnt_eq}}
\label{app:w mnt eq}
Given any tensors $A^m_{i_1\cdots i_l}$ for all $l,m = 0,1,\cdots$, the basis functions $\psi^{m}_{i_1,\cdots,i_l}$ satisfy the following relation \cite[Appendix A.2.3]{Struchtrup2005macroscopic}:
\begin{equation}\label{psi orthorgonality}
    \sum_{m,l=0}^{+\infty} A^m_{i_1\cdots i_l} \langle  f_M \psi^m_{i_1\cdots i_l},\psi^n_{j_1\cdots j_{l'}} \rangle = \frac{ l'!}{(2l'+1)!!} A^n_{\langle j_1\cdots j_{l'} \rangle}.
\end{equation}
Using this property, we can derive \eqref{eq:mnt_eq} by multiplying the Boltzmann equation \eqref{boltzmann eq} by $\psi_{i_1\cdots i_l}^n$ and then integrating with respect to $\bxi$. The calculation of the right-hand side is the as follows:
\begin{equation} \label{eq:rhs}
 \begin{split}
 \frac{1}{\kn} \langle \mathcal{L}[f], \psi_{i_1,\cdots,i_l}^n \rangle &=
\frac{1}{\kn} \sum_{m,k=0}^{+\infty} \frac{(2k+1)!!}{k!}w_{j_1,\cdots,j_k}^m  \langle \mathcal{L}[\psi^{m}_{j_1,\cdots,j_k} f_M], \psi^n_{i_1,\cdots,i_l}\rangle\\
& = \frac{1}{\kn} \sum_{m,k=0}^{+\infty} \sum_{n'=0}^{+\infty} a_{kmn'}  \frac{(2k+1)!!}{k!}w_{j_1,\cdots,j_k}^m \langle f_M  \psi^{n'}_{j_1,\cdots,j_k}, \psi^n_{i_1,\cdots,i_l}\rangle\\
& = \frac{1}{\kn} \sum_{n'=0}^{+\infty} a_{lnn'} w_{i_1,\cdots, i_l}^{n'},
 \end{split}
\end{equation}
where we have used \eqref{eq:rot_inv} to expand $\mathcal{L}[\psi^{m}_{j_1,\cdots,j_k}]$.
For the left-hand side, we need another property from \cite[Appendix A.2.3]{Struchtrup2005macroscopic}:
\begin{equation}\label{psi xi}
  \begin{split}
  &\psi^n_{i_1\cdots i_l} \xi_j = \left( \sqrt{2(n+l)+3}  \psi^n_{i_1 \cdots i_l j} - \sqrt{2n}\psi^{n-1}_{i_1 \cdots i_l j}    \right)+\\
   & \qquad \qquad  + \frac{l}{2l+1} \left( \sqrt{2(n+l)+1} \psi^{n}_{\langle i_1 \cdots i_{l-1} } \delta_{i_{l} \rangle j}   -\sqrt{2(n+1)} \psi^{n+1}_{\langle  i_1 \cdots i_{l-1}  } \delta_{i_{l} \rangle j} \right),
    \end{split}
\end{equation}
using which we have
\begin{equation} \label{eq:lhs}
\begin{split}
  \left\langle \xi_j \frac{\partial f}{\partial x_j}, \psi_{i_1,\cdots,i_l}^n \right\rangle &=
  \left\langle \frac{\partial f}{\partial x_j}, \psi_{i_1,\cdots,i_l}^n \xi_j \right\rangle
  = \sum_{m,k=0}^{+\infty} \frac{(2k+1)!!}{k!} \frac{\partial w_{j_1,\cdots,j_k}^m }{\partial x_j}  \langle f_M \psi^{m}_{j_1,\cdots,j_k},  \psi^n_{i_1,\cdots,i_l} \xi_j \rangle \\
  &= \left( \sqrt{2(n+l)+3} \frac{\partial w_{i_1\cdots i_l j}^n}{\partial x_j} - \sqrt{2n} \frac{\partial w_{i_1\cdots i_l j}^{n-1}}{\partial x_j} \right) \\
  & \qquad + \frac{l}{2l+1} \left( \sqrt{2(n+l)+1} \frac{\partial w_{\langle i_1\cdots i_{l-1}}^n}{\partial x_{i_l \rangle}} - \sqrt{2(n+1)} \frac{\partial w_{\langle i_1\cdots i_{l-1}}^{n+1}}{\partial x_{i_l\rangle}} \right).
\end{split}
\end{equation} 
Equating \eqref{eq:lhs} and \eqref{eq:rhs} gives the moment equations \eqref{eq:mnt_eq}.

\section{Proof of \textbf{(O0)}-\textbf{(O4)}}
Similar as \eqref{eq:w_asymp}, $T_{i_1\cdots i_l}^n$ can also be expanded as 
\begin{equation}\label{eq:T_asymp}
   T_{i_1\cdots i_l}^n =  T_{i_1\cdots i_l}^{n|0} + \kn T_{i_1\cdots i_l}^{n|1} + \kn^2 T_{i_1\cdots i_l}^{n|2} + \kn^3 T_{i_1 \cdots i_l}^{n|3} + \cdots.
\end{equation}
\label{app:proof of lemma}
\subsection{Zeroth order}
To find the zeroth-order terms in the asymptotic expansion \eqref{eq:w_asymp}, we insert \eqref{eq:w_asymp}\eqref{eq:T_asymp} into \eqref{eq:mnt_eq} and balance the $O(\kn^{-1})$ terms on both sides. The result is
\begin{equation} \label{eq:zeroth}
\sum_{n'=0}^{+\infty} a_{lnn'} w_{i_1\cdots i_l}^{n'|0} = 0
\end{equation}
For $l \geq 2$, we can multiply both sides by $b_{ln_1n}^{(0)}$ and take the sum over $n$, which gives us
\begin{displaymath}
w_{i_1\cdots i_l}^{n_1|0} = 0, \qquad \forall n_1 \geqslant 0.
\end{displaymath}
When $l = 0$ and $l = 1$, due to the relations \eqref{eq:a_zero_value}, we can multiply both sides of \eqref{eq:zeroth} by $b_{0n_1n}^{(2)}$ or $b_{1n_1n}^{(1)}$ and then take the sum over $n$, to obtain
\begin{displaymath}
w^{n|0} = 0, \quad \forall n \geqslant 2, \qquad \text{and} \qquad w_i^{n|0} = 0, \quad \forall n \geqslant 1.
   \end{displaymath}
The values of $w^{0|0}$, $w_i^{0|0}$ and $w^{1|0}$ are allowed to be nonzero, and thus we reach to the conclusion $\textbf{(O0)}$. Following the idea of the Chapman-Enskog expansion, we assume that
\begin{equation}\label{assump zeroth order}
w^{0|k} = w^{1|k} = w_i^{0|k} = 0, \qquad \forall k \geqslant 1.
\end{equation}
Using these results, we can find the following zeroth-order terms for $T_{i_1 \cdots i_l}^n$:
\begin{gather}
 T^{n|0} = 0, \qquad n \geqslant 2, \\
\label{eq:T_i^10}
T_i^{1|0} = \frac{\sqrt{5}}{3} \frac{\partial w^1}{\partial x_i}, \quad
T_i^{n|0} = 0, \qquad n \geqslant 2, \\
\label{eq:T_ij^00}
T_{ij}^{0|0} =  \frac{2}{\sqrt{5}} \frac{\partial w_{\langle i}^0}{\partial x_{j\rangle}}, \quad
T_{ij}^{n|0} = 0, \qquad n \geqslant 1,\\
T_{i_1 \cdots i_l}^{n|0} = 0, \qquad l \geqslant 3, \quad n \geqslant 0.
\end{gather}
These results are to be used in the derivation of first-order terms.

\subsection{First order}
By inserting \eqref{eq:w_asymp}\eqref{eq:T_asymp} into \eqref{eq:mnt_eq} and balancing the $O(1)$ terms on both sides, we get
\begin{equation*}
T_{i_1\cdots i_l}^{n|0} = \sum_{n'=0}^{+\infty} a_{lnn'} w^{n'|1}_{i_1\cdots i_l}.
\end{equation*}

\paragraph{Expressions of $\boldsymbol{w_{i_1\cdots i_l}^{n|1}}$ for $\boldsymbol{l = 0}$ and $\boldsymbol{l \geqslant 3}$}
When $l = 0$ and $n \geqslant 2$, the left-hand side of the above equation is zero, and therefore
\begin{equation} \label{eq:w_0^n1}
w^{n|1} = 0, \qquad n \geqslant 2.
\end{equation}
Similarly,
\begin{equation} \label{eq:w_3^n1}
w_{i_1\cdots i_l}^{n|1} = 0, \qquad \text{if } l \geqslant 3 \text{ and } n \geqslant 0.
\end{equation}

\paragraph{Expressions of $\boldsymbol{w_i^{n|1}}$}
Using \eqref{eq:T_i^10}, we have
\begin{displaymath}
   \frac{\sqrt{5}}{3} \frac{\partial w^1}{\partial x_i}  = \sum_{n'=1}^{+\infty} a_{11n'} w_i^{n'|1} \text{~and~}
 \sum_{n'=1}^{+\infty} a_{1nn'} w_i^{n'|1} = 0, \text{~for~} n \geqslant 2.
\end{displaymath}
Hence,
\begin{equation*} 
w_i^{n|1} = \frac{\sqrt{5}}{3} b_{11n}^{(1)} \frac{\partial w^1}{\partial x_i}
\end{equation*}
which produces the expression \eqref{eq:w_i^n1}. Similar to the Chapman-Enskog expansion, here we would like to assume that
\begin{equation}\label{assump wi1k}
w_i^{1|k} = 0, \qquad \text{for all } k \geqslant 2.
\end{equation}

\paragraph{Expressions of $\boldsymbol{w_{ij}^{n|1}}$}
Following the similar argument for the case $l = 1$ upon using \eqref{eq:T_ij^00}, we have
\begin{equation} \label{eq:f_2mn^1}
w_{ij}^{n|1} = \frac{2}{\sqrt{5}} b_{20n}^{(0)} \frac{\partial w_{\langle i}^0}{\partial x_{j\rangle}},
\end{equation}
which yields the relation \eqref{eq:w_ij^n1}. Again, since $w_{ij}^0$ appears in the final equations, instead of \eqref{eq:f_2mn^1}, we are going to assume
\begin{equation}\label{assump wij0k}
w_{ij}^{0|k} = 0, \qquad \text{for all } k \geqslant 2.
\end{equation}


By \eqref{eq:w_0^n1} and \eqref{eq:w_3^n1}, we conclude that $\{w^n_i\}_{n=1}^\infty$, $\{w^n_{ij}\}_{n=0}^\infty$ are only first order moments in the expansion of distribution function and thus we arrive at the statement \textbf{(O1)}. By inserting the above results \eqref{eq:w_0^n1}\eqref{eq:w_3^n1}\eqref{eq:w_i^n1}\eqref{eq:w_ij^n1} into \eqref{eq:T}, we can obtain the expressions of $T_{i_1\cdots i_l}^{n|1}$. When $l \geqslant 4$, it is easy to see from \eqref{eq:T} that $T_{i_1\cdots i_l}^{n|1} = 0$. For $l = 0,1,2,3$, these quantities are given as follows:

\paragraph{Expressions of $\boldsymbol{T^{n|1}}$}
By the relation \eqref{eq:w_i^n1}, it holds that
\begin{equation}\label{Tn1}
T^{n|1} = \sqrt{2n+3} \frac{\partial w_j^{n|1}}{\partial x_j} - \sqrt{2n} \frac{\partial w_j^{n-1|1}}{\partial x_j} 
= \frac{\sqrt{2n+3} b_{11n}^{(1)} - \sqrt{2n} b_{11,n-1}^{(1)}}{b_{111}^{(1)}} \frac{\partial w_j^{1|1}}{\partial x_j}
\end{equation}
for $n \geqslant 1$. Note that when $n = 1$, the coefficient $b_{110}^{(1)}$ should be regarded as zero since $w_j^{0|1} = 0$ by the assumption \eqref{assump zeroth order}.

\paragraph{Expressions of $\boldsymbol{T_i^{n|1}}$}
By the relation \eqref{eq:w_ij^n1}, we have
\begin{equation} \label{eq:T_1mn^1}
T_i^{n|1} = \sqrt{2n+5} \frac{\partial w_{ij}^{n|1}}{\partial x_j} - \sqrt{2n} \frac{\partial w_{ij}^{n-1|1}}{\partial x_j} 
= \frac{\sqrt{2n+5} b_{20n}^{(0)} - \sqrt{2n} b_{20,n-1}^{(0)}}{b_{200}^{(0)}} \frac{\partial w_{ij}^{0|1}}{\partial x_j}
\end{equation}
for $n\geqslant 1$.

\paragraph{Expressions of $\boldsymbol{T_{ij}^{n|1}}$}
By the relation \eqref{eq:w_i^n1}, we have 
\begin{equation} \label{eq:T_2mn^1}
T_{ij}^{n|1} = \frac{2}{5} \left( \sqrt{2n+5} \frac{\partial w_{\langle i}^{n|1}}{\partial x_{j\rangle}} - \sqrt{2(n+1)} \frac{\partial w_{\langle i}^{n+1|1}}{\partial x_{j\rangle}} \right) 
= \frac{2}{5} \frac{\sqrt{2n+5} b_{11n}^{(1)} - \sqrt{2(n+1)} b_{11,n+1}^{(1)}}{b_{111}^{(1)}} \frac{\partial w_{\langle i}^{1|1}}{\partial x_{j \rangle}}
\end{equation}
for $n \geqslant 0$, where the coefficient $b_{110}^{(1)}$ should be regarded as zero when $n=0$ again by the assumption \eqref{assump zeroth order}.

\paragraph{Expressions of $\boldsymbol{T_{ijk}^{n|1}}$}
Since $w_{ijkl}^{n|1} = 0$, by the relation \eqref{eq:w_ij^n1} we have
\begin{equation}\label{eq:T_3mn^1}
T_{ijk}^{n|1} = \frac{3}{7} \left( \sqrt{2n+7} \frac{\partial w_{\langle ij}^{n|1}}{\partial x_{k\rangle} } - \sqrt{2(n+1)} \frac{\partial w_{\langle ij}^{n+1|1}}{\partial x_{k\rangle}} \right) 
= \frac{3}{7} \frac{\sqrt{2n+7} b_{20n}^{(0)} - \sqrt{2(n+1)} b_{20,n+1}^{(0)}}{b_{200}^{(0)}} \frac{\partial w_{\langle ij}^{0|1}}{\partial x_{k \rangle}}
\end{equation}

\subsection{Second order}
We now find the second-order moments in the expansion of distribution function based on 
\begin{equation*}
T_{i_1\cdots i_l}^{n|1} = \sum_{n'=0}^{+\infty} a_{lnn'} w^{n'|2}_{i_1\cdots i_l}.
\end{equation*}
\paragraph{Expressions of $\boldsymbol{w^{n|2}}$}
By \eqref{Tn1}, we have
\begin{equation*}
w^{n|2} = \sum_{n'=2}^{+\infty} b_{0nn'}^{(2)} T^{n'|1} 
= \gamma^{(2),n}_0 \frac{\partial w_j^{1|1}}{\partial x_j}
\end{equation*}
for $n\geqslant 2$, where $\gamma^{(2),n}_0$ is defined in \eqref{w^n}. Note that the summation begins from $n'=2$ due to the assumption \eqref{assump zeroth order}.

\paragraph{Expressions of $\boldsymbol{w_i^{n|2}}$}
By \eqref{eq:T_1mn^1}, we have
\begin{equation} \label{eq:f_1mn^2_intermediate}
w_i^{n|2} = \sum_{n'=2}^{+\infty} b_{1nn'}^{(2)} T_i^{n'|1} 
=\gamma^{(1),n}_{1} \frac{\partial w_{ij}^{0|1}}{\partial x_j}
\end{equation}
for $n\geqslant 2$, where $\gamma^{(1),n}_{1}$ is defined in \eqref{eq:w_i^n2}. Note that the summation begins from $n'=2$ due to the assumption \eqref{assump wi1k}.

\paragraph{Expressions of $\boldsymbol{w_{ij}^{n|2}}$}
By \eqref{eq:T_2mn^1}, we have 
\begin{equation} \label{eq:f_2mn^2_intermediate}
w_{ij}^{n|2} = \sum_{n'=1}^{+\infty} b_{2nn'}^{(1)} T_{ij}^{n'|1} 
=  \gamma^{(1),n}_2 \frac{\partial w_{\langle i}^{1|1}}{\partial x_{j\rangle}},
\end{equation}
for $n \geqslant 1$, where $\gamma^{(1),n}_2$ is defined in \eqref{eq:w_ij^n2}. Note that the summation begins from $n'=1$ due to the assumption \eqref{assump wij0k}.

\paragraph{Expressions of $\boldsymbol{w_{ijk}^{n|2}}$}
By \eqref{eq:T_3mn^1}, we have
\begin{equation}
w_{ijk}^{n|2} = \sum_{n'=0}^{+\infty} b_{3nn'}^{(0)} T_{ijk}^{n'|1} =
\gamma^{(2),n}_3 \frac{\partial w_{\langle ij}^{0|1}}{\partial x_{k\rangle}}
\end{equation}
for any non-negative $n$, where $\gamma^{(2),n}_3$ is defined in \eqref{wijk}.

\paragraph{Expressions of $\boldsymbol{w_{i_1\cdots i_l}^{n|2}}$ for $\boldsymbol{l \geqslant 4}$}
When $l \geqslant 4$, since $T_{i_1\cdots i_l}^{n|1} = 0$, we have
\begin{displaymath}
w_{i_1\cdots i_l}^{n|2} = 0.
\end{displaymath}

At this point, we may conclude that $\{w^n\}_{n=2}^\infty$ and $\{w^n_{ijk}\}_{n=0}^\infty$ are the all second-order moments in the expansion of distribution function, and thus we arrive at \textbf{(O2)}. Finally, the statement \textbf{(O3)} is the direct conclusion from 
\begin{displaymath}
    w_{ijkl}^{n|3} = \sum_{n'=0}^{+\infty} b_{4nn'}^{(0)} T_{ijkl}^{n'|2} \sim O(\kn^3),
\end{displaymath}
while \textbf{(O4)} comes from 
\begin{displaymath}
   w_{i_1\cdots i_l}^{n|3} = \sum_{n'=0}^{+\infty} b_{lnn'}^{(0)} T_{i_1\cdots i_l}^{n'|2} = 0 
\end{displaymath}
as $T_{i_1\cdots i_l}^{n|2} = 0$ when $l \geqslant 5$.





\section{Proof of Proposition \ref{prop}}\label{sm:prop}
We first claim that $\bar{\mathcal{S}} = \bar{\mathcal{S}} (\mathcal{I} - \mathcal{P}_u)$, or equivalently, $\bar{\mathcal{S}} \mathcal{P}_u = 0$. This can be shown by direct calculation:
\begin{equation} \label{eq:SPu}
\begin{split}
\bar{\mathcal{S}} \mathcal{P}_u &= (\mathcal{I} - \mathcal{P}_u) (\mathcal{I} - \mathcal{S} \mathcal{S}_{uu}^{-1} \mathcal{P}_u) \mathcal{S} \mathcal{P}_u = (\mathcal{I} - \mathcal{P}_u) \mathcal{S} \mathcal{P}_u - (\mathcal{I} - \mathcal{P}_u) \mathcal{S} \mathcal{S}_{uu}^{-1} \mathcal{P}_u \mathcal{S} \mathcal{P}_u \\
&= (\mathcal{I} - \mathcal{P}_u) \mathcal{S} \mathcal{P}_u - (\mathcal{I} - \mathcal{P}_u) \mathcal{S} \mathcal{P}_u = 0.
\end{split}
\end{equation}
Since $\overline{\mathbb{V}}_{\rho} \subset \overline{\mathbb{V}}^{(0)}$ and $\overline{\mathbb{V}}_{\rho} \subset \overline{\mathbb{V}}_{\mathrm{even}}$, we can see from \eqref{eq:Aj} and \eqref{eq:oe_decomp} that
\begin{displaymath}
\bar{\mathcal{A}}_{\mathrm{oe}} \mathcal{P}_{\rho} = \bar{\mathcal{A}}_n \mathcal{P}_{\rho} = \mathcal{P}_{\overline{\mathbb{V}}} \xi_n \mathcal{P}_{\rho}.
\end{displaymath}
Thus, by $\xi_n \psi^0 = \psi_n^0 \in \overline{\mathbb{V}}_u$, we get $\bar{\mathcal{A}}_{\mathrm{oe}} \mathcal{P}_{\rho} = \mathcal{P}_u \xi_n \mathcal{P}_{\rho}$. Therefore, one can derive from \eqref{eq:SPu} that
\begin{displaymath}
\bar{\mathcal{S}} \bar{\mathcal{A}}_{\mathrm{oe}} \mathcal{P}_{\rho} = \bar{\mathcal{S}} \mathcal{P}_u \xi_n \mathcal{P}_{\rho} = 0,
\end{displaymath}
which proves \eqref{eq:SA}.

The proof of \eqref{eq:Puf} is straightforward. Using $\mathcal{P}_u \bar{\mathcal{S}} = \mathcal{P}_u (\mathcal{I} - \mathcal{P}_u) \cdots = 0$, we can find that
\begin{displaymath}
\mathcal{P}_u \bar{f} = \mathcal{P}_u \mathcal{P}_{\mathrm{odd}} |_{\overline{\mathbb{V}}} \bar{f} = \mathcal{P}_u \bar{\mathcal{S}} \bar{\mathcal{A}}_{\mathrm{oe}}(f_W - \mathcal{P}_{\mathrm{even}}|_{\overline{\mathbb{V}}} \,\bar{f}) = 0.
\end{displaymath}
To show \eqref{eq:Onsager}, we need to choose $\rho_W$ such that
\begin{displaymath}
\mathcal{P}_u \mathcal{S} \bar{\mathcal{A}}_{\mathrm{oe}} (f_W - \mathcal{P}_{\mathrm{even}}|_{\overline{\mathbb{V}}} \bar{f}) = 0.
\end{displaymath}
When \eqref{eq:finalBC} holds, one can use the two equalities above to obtain
\begin{displaymath}
\begin{split}
\mathcal{P}_{\mathrm{odd}}|_{\overline{\mathbb{V}}} \, \bar{f} &= \mathcal{P}_u \bar{f} + (\mathcal{I} - \mathcal{P}_u) \mathcal{P}_{\mathrm{odd}}|_{\overline{\mathbb{V}}} \, \bar{f} = 0 + (\mathcal{I} - \mathcal{P}_u) \bar{\mathcal{S}} \bar{\mathcal{A}}_{\mathrm{oe}}(f_W - \mathcal{P}_{\mathrm{even}}|_{\overline{\mathbb{V}}} \,\bar{f}) \\
&= (\mathcal{I} - \mathcal{P}_u) (\mathcal{I} - \mathcal{S} \mathcal{S}_{uu}^{-1} \mathcal{P}_u)\mathcal{S} \bar{\mathcal{A}}_{\mathrm{oe}}(f_W - \mathcal{P}_{\mathrm{even}}|_{\overline{\mathbb{V}}} \,\bar{f}) \\
&= (\mathcal{I} - \mathcal{P}_u) \mathcal{S} \bar{\mathcal{A}}_{\mathrm{oe}}(f_W - \mathcal{P}_{\mathrm{even}}|_{\overline{\mathbb{V}}} \,\bar{f}) - (\mathcal{I} - \mathcal{P}_u)\mathcal{S} \mathcal{S}_{uu}^{-1} \mathcal{P}_u \mathcal{S} \bar{\mathcal{A}}_{\mathrm{oe}}(f_W - \mathcal{P}_{\mathrm{even}}|_{\overline{\mathbb{V}}} \,\bar{f}) \\
&= (\mathcal{I} - \mathcal{P}_u) \mathcal{S} \bar{\mathcal{A}}_{\mathrm{oe}}(f_W - \mathcal{P}_{\mathrm{even}}|_{\overline{\mathbb{V}}} \,\bar{f}) = \mathcal{S} \bar{\mathcal{A}}_{\mathrm{oe}}(f_W - \mathcal{P}_{\mathrm{even}}|_{\overline{\mathbb{V}}} \,\bar{f}).
\end{split}
\end{displaymath}
This proves \eqref{eq:Onsager}.

It remains to show that $\bar{\mathcal{S}}$ is self-adjoint and positive semidefinite. Its self-adjointness can be seen by rewriting the expression of $\bar{\mathcal{S}}$ as
\begin{displaymath}
\bar{\mathcal{S}} = \bar{\mathcal{S}} (\mathcal{I} - \mathcal{P}_u) = (\mathcal{I} - \mathcal{P}_u)(\mathcal{S} - \mathcal{S} \mathcal{P}_u \mathcal{S}_{uu}^{-1} \mathcal{P}_u \mathcal{S}) (\mathcal{I} - \mathcal{P}_u).
\end{displaymath}
For any $\bar{g} \in \overline{\mathbb{V}}_{\mathrm{odd}}$, we can define $g = (\mathcal{I} - \mathcal{S}_{uu}^{-1} \mathcal{P}_u \mathcal{S})(\mathcal{I} - \mathcal{P}_u) \bar{g}$ and verify that
\begin{displaymath}
\langle \bar{g}, \bar{\mathcal{S}} \bar{g} \rangle = \langle g, \mathcal{S} g \rangle.
\end{displaymath}
Therefore, the positive semidefiniteness of $\bar{\mathcal{S}}$ follows the positive semidefiniteness of $\mathcal{S}$.

\section{Expressions of coefficients}\label{app:coefficients}
\subsection{Expressions of $c^{(d),n}_{l}$}
\label{app:cdl}
\begin{itemize}
    \item \textbf{Coefficients in the first-order moments:} The expressions of $c^{(1),n}_{1}$ and $c^{(1),n}_{2}$ are given by 
    \begin{align*}
        c^{(1),n}_{1} = & \ b_{11n}^{(1)}\bigg/\sqrt{\sum^{+\infty}_{n=1} \left(b^{(1)}_{11n}\right)^2 }, \text{~for~} n\geqslant 1\\
        c^{(1),n}_{2} = & \ b_{20n}^{(0)}\bigg/ \sqrt{\sum^{+\infty}_{n=0} \left(b^{(0)}_{20n}\right)^2 }, \text{~for~} n\geqslant 0,
    \end{align*}
    where $b^{n_0}_{lnn'}$ can be obtained from \eqref{sum ab}.
    \item \textbf{Coefficients in the first-order moments:} The expressions of $c^{(2),n}_{0}$ and $c^{(2),n}_{3}$ are given by 
        \begin{align*}
        c^{(2),n}_{0} = & \ \gamma^{(2),n}_0 \bigg/ \sqrt{\sum^{+\infty}_{n=2} \left( \gamma^{(2),n}_0\right)^2 }, \text{~for~} n\geqslant 2,\\
        c^{(2),n}_{3} = & \ \gamma^{(2),n}_3 \bigg/\sqrt{\sum^{+\infty}_{n=0} \left(\gamma^{(2),n}_3\right)^2 }, \text{~for~} n\geqslant 0
    \end{align*}
  where $\gamma^{(2),n}_0$ and $\gamma^{(2),n}_3$ are respectively given by \eqref{w^n} and \eqref{wijk}. \eqref{eq:w_i^n2} and \eqref{eq:w_ij^n2}. 

For non-Maxwell molecules, the expressions of $c^{(2),n}_{1}$ and $c^{(2),n}_{2}$ are given by 
\begin{displaymath}
 \begin{split}
 c^{(2),n}_1 =  & \ \frac{1}{C^{(2)}_1} \times \begin{cases}
  - \sum_{n'=2}^{+\infty} \bar{c}^{n'}_1 \frac{b^{(1)}_{11n'}}{b^{(1)}_{111}}, & \text{~if~}   n=1\\
    \bar{c}^n_1 , & \text{~if~}   n\geqslant 2\\
 \end{cases} , \\ 
  c^{(2),n}_{2} = & \ \frac{1}{C^{(2)}_2} \times \begin{cases}
  - \sum_{n'=1}^{+\infty} \bar{c}^{n'}_{2} \frac{b^{(0)}_{20n'}}{b^{(0)}_{200}}, & \text{~if~}   n=0\\
    \bar{c}^n_{2} , & \text{~if~}   n \geqslant 1\\
 \end{cases}
  \end{split}
\end{displaymath}
where $C^{(2)}_{1}$, $C^{(2)}_{1}$ are constants chosen such that the scaling 
\begin{equation*}
  \sum_{n=1}^{+\infty} \left|c_1^{(2),n} \right|^2 = \sum_{n=0}^{+\infty} \left|c_2^{(2),n} \right|^2 = 1
\end{equation*}
holds, and $\bar{c}^n_{1}$, $\bar{c}^n_{2}$ are respectively the solutions of the linear system 
\begin{equation*}
    \begin{pmatrix}
 \tau^2_3 & \tau^3_3 & \alpha_3 \frac{b^{(1)}_{114}}{b^{(1)}_{111}} &  \alpha_3 \frac{b^{(1)}_{115}}{b^{(1)}_{111}} & \cdots & \alpha_3 \frac{b^{(1)}_{11D}}{b^{(1)}_{111}} & \cdots  \\
 \tau^2_4 & \alpha_4 \frac{b^{(1)}_{113}}{b^{(1)}_{111}} & \tau^4_4 &  \alpha_4 \frac{b^{(1)}_{115}}{b^{(1)}_{111}} & \cdots & \alpha_4 \frac{b^{(1)}_{11D}}{b^{(1)}_{111}} & \cdots \\
 \vdots &  \vdots &  \vdots &  \vdots &  \ddots &  \vdots &  \ddots \\
  \tau^2_D & \alpha_D \frac{b^{(1)}_{113}}{b^{(1)}_{111}} &  \alpha_D \frac{b^{(1)}_{114}}{b^{(1)}_{111}} &  \cdots &\alpha_D \frac{b^{(1)}_{11,D-1}}{b^{(1)}_{111}} & \tau^D_D  & \cdots \\
   \vdots &  \vdots &  \vdots &  \vdots &  \ddots &  \vdots &  \ddots 
   \end{pmatrix}
   \begin{pmatrix}
    \bar{c}^2_1\\[3pt]
    \bar{c}^3_1\\[3pt]
    \vdots \\[3pt]
    \bar{c}^{D-1}_1\\[3pt]
      \bar{c}^{D}_1\\[3pt]
    \vdots
   \end{pmatrix} = \boldsymbol{0}
\end{equation*}
where 
\begin{displaymath}
  \tau^2_D = \frac{b^{(1)}_{112}}{b^{(1)}_{111}}\alpha_D - \frac{\gamma^{(1),D}_{1}}{\gamma^{(1),2}_{1}},\ \tau^D_D = \frac{b^{(1)}_{11D}}{b^{(1)}_{111}}\alpha_D + 1,  \  \alpha_D =  \frac{b^{(1)}_{11D}}{b^{(1)}_{111}} - \frac{\gamma^{(1),D}_{1}}{\gamma^{(1),2}_{1}} \frac{b^{(1)}_{112}}{b^{(1)}_{111}}  \text{~with~} \bar{c}^2_1 = 1
\end{displaymath}
and 
\begin{equation*}
    \begin{pmatrix}
 \tilde{\tau}^1_2 & \tilde{\tau}^2_2 & \tilde{\alpha}_2 \frac{b^{(0)}_{203}}{b^{(0)}_{200}} &  \tilde{\alpha}_2 \frac{b^{(0)}_{204}}{b^{(0)}_{200}} & \cdots & \tilde{\alpha}_2 \frac{b^{(0)}_{20D}}{b^{(0)}_{200}} & \cdots  \\
 \tilde{\tau}^1_3 & \tilde{\alpha}_3 \frac{b^{(0)}_{202}}{b^{(0)}_{200}} & \tilde{\tau}^3_3 &  \tilde{\alpha}_3 \frac{b^{(0)}_{204}}{b^{(0)}_{200}} & \cdots & \tilde{\alpha}_3 \frac{b^{(0)}_{20D}}{b^{(0)}_{200}} & \cdots \\
 \vdots &  \vdots &  \vdots &  \vdots &  \ddots &  \vdots &  \ddots \\
  \tilde{\tau}^1_D & \tilde{\alpha}_D \frac{b^{(0)}_{202}}{b^{(0)}_{200}} &  \tilde{\alpha}_D \frac{b^{(0)}_{203}}{b^{(0)}_{200}} &  \cdots &\tilde{\alpha}_D \frac{b^{(0)}_{20,D-1}}{b^{(0)}_{200}} & \tilde{\tau}^D_D  & \cdots \\
   \vdots &  \vdots &  \vdots &  \vdots &  \ddots &  \vdots &  \ddots 
   \end{pmatrix}
   \begin{pmatrix}
    \bar{c}^1_{2}\\[3pt]
    \bar{c}^2_{2}\\[3pt]
    \vdots \\[3pt]
    \bar{c}^{D-1}_{2}\\[3pt]
       \bar{c}^{D}_{2}\\[3pt]
    \vdots
   \end{pmatrix} = \boldsymbol{0}
\end{equation*}
where 
\begin{displaymath}
  \tilde{\tau}^1_D = \frac{b^{(0)}_{201}}{b^{(0)}_{200}}\tilde{\alpha}_D - \frac{\gamma^{(1),D}_{2}}{\gamma^{(1),1}_{2}},\  \tilde{\tau}^D_D = \frac{b^{(0)}_{20D}}{b^{(0)}_{200}}\tilde{\alpha}_D + 1, \  \tilde{\alpha}_D =  \frac{b^{(0)}_{20D}}{b^{(0)}_{200}} - \frac{\gamma^{(1),D}_{2}}{\gamma^{(1),1}_{2}} \frac{b^{(0)}_{201}}{b^{(0)}_{200}} \text{~with~} \bar{c}^2_1 = 1.
\end{displaymath}
For Maxwell molecules, we have 
\begin{displaymath}
  c^{(2),n}_1 = 0  \text{~for~} n\geqslant 1 \text{~and~} c^{(2),n}_2 = \begin{cases}
        1, & \text{~if~} n = 1\\
          0 , & \text{~if~} n=0 \text{~or~} n \geqslant 2
    \end{cases}.
\end{displaymath}

\end{itemize}

\subsection{Expressions of $A_{ij}$}
\label{app:matrices}
\begin{equation*}
    A_{ij} = \begin{cases}
         3  \sum_{n=1}^{+\infty}c^{(1),n}_1 \left(\sqrt{2n+5}c^{(1),n}_2 - \sqrt{2n}c^{(1),n-1}_2 \right), &\text{~if~}  (i,j)=(4,5)\\
         \sum_{n=1}^{+\infty}  c^{(1),n}_1 \left( \sqrt{2n+3}c^{(2),n}_0 - \sqrt{2(n+1)}c^{(2),n+1}_0 \right),&\text{~if~}  (i,j)=(4,6)\\
          3  \sum_{n=1}^{+\infty}c^{(1),n}_1 \left(\sqrt{2n+5}c^{(2),n}_2 - \sqrt{2n}c^{(2),n-1}_2 \right),&\text{~if~}  (i,j)=(4,8)\\
          3   \sum_{n=0}^{+\infty} c^{(1),n}_2 \left( \sqrt{2n+5}c^{(2),n}_1 - \sqrt{2(n+1)} c^{(2),n+1}_1 \right) ,&\text{~if~}  (i,j)=(5,7)\\
          \frac{15}{2}\sum_{n=0}^{+\infty} c^{(1),n}_2 \left( \sqrt{2n+7} c^{(2),n}_3 - \sqrt{2n}c^{(2),n-1}_3 \right) ,&\text{~if~}  (i,j)=(5,9)\\
          \sum_{n=2}^{+\infty} c^{(2),n}_0 \left( \sqrt{2n+3} c^{(2),n}_1 - \sqrt{2n}c^{(2),n-1}_1 \right),&\text{~if~}  (i,j)=(6,7) \\
          3  \sum_{n=1}^{+\infty}c^{(2),n}_1 \left(\sqrt{2n+5}c^{(2),n}_2 - \sqrt{2n}c^{(2),n-1}_2 \right), &\text{~if~}  (i,j)=(7,8)\\
          \frac{15}{2}  \sum_{n=0}^{+\infty} c^{(2),n}_2 \left( \sqrt{2n+7} c^{(2),n}_3 - \sqrt{2n}c^{(2),n-1}_3 \right)  , &\text{~if~}  (i,j)=(8,9)\\
    \end{cases}
\end{equation*}

\subsection{Expressions of $\lambda'_{ij}$}\label{app:bc}
The boundary conditions \eqref{eq:finalBC} can be written uniformly as 
\begin{equation} \label{bc with z}
   \begin{split}
& w_n^0 = 0, \\
& u_{i_1\cdots i_l}^{(p)} = \frac{2\chi}{2-\chi}\sum_m c^{(p),m}_l \Bigg[
     3\sum_{p'=1,2} \sum^{+\infty}_{m'=1}  c^{(p'),m'}_1  \left( \ms{A}^{m'}_{m;i_1\cdots i_l} - \frac{\ms{A}^{m'}_{0;n}I^0_{m;i_1 \cdots i_l} }{I^0_{0;n}} \right) z^{(p')}_{n} \\
&  + 15 \sum_{t=t_1,t_2} \sum_{p'=1,2} \sum^{+\infty}_{m'=0}   c^{(p'),m'}_2  \ms{B}^{m';t}_{m;i_1 \cdots i_l} z^{(p')}_{tn} + \frac{35}{2}\sum^{+\infty}_{m'=0}  c^{(2),m'}_3\left( \ms{D}^{m'}_{m;i_1\cdots i_l} - \frac{\ms{D}^{m'}_{0;n}I^0_{m;i_1 \cdots i_l} }{I^0_{0;n}} \right) z^{(2)}_{nnn}  \\
&\hspace{120pt} + \frac{105}{2}  \sum_{t',t''=t_1,t_2} \sum^{+\infty}_{m'=0}  c^{(2),m'}_3 \left(  \ms{C}^{m';t't''}_{m;i_1\cdots i_l} -\frac{\ms{C}^{m';t't''}_{0;n}I^0_{m;i_1 \cdots i_l} }{I^0_{0;n}}  \right) z^{(2)}_{t't''n}  \Bigg],
   \end{split}
\end{equation}
where $u_{i_1\cdots i_l}^{(p)}$ is taken in the set $\{u^{(1)}_{n}, u^{(2)}_{n}, u^{(1)}_{t_1 n}, u^{(1)}_{t_2 n}, u^{(2)}_{t_1 n}, u^{(2)}_{t_2 n}, u^{(2)}_{nnn}, u^{(2)}_{t_1 t_1 n}, u^{(2)}_{t_2 t_2 n}, u^{(2)}_{t_1 t_2 n}\}$, and the coefficients are given by
\begin{align*}
 I_{m;i_1 \cdots i_l}^{m';j_1 \cdots j_k} 
         & =   \int_{\bxi\cdot \bn < 0} f_M(\bxi) \bar{L}_{m'}^{(k+1/2)}\left( \frac{|\bxi|^2}{2} \right)\bar{L}_m^{(l+1/2)}\left( \frac{|\bxi|^2}{2} \right) \xi_{\langle i_{ 1}}\cdots \xi_{i_{l}\rangle}\xi_{\langle j_{ 1}}\cdots \xi_{j_{k}\rangle} \,\dd \bxi,\\
     \ms{A}^{m'}_{m;i_1 \cdots i_l} & =  \int_{\bxi\cdot \bn < 0} f_M(\bxi) \bar{L}_{m'}^{(3/2)}\left( \frac{|\bxi|^2}{2} \right)\bar{L}_m^{(l+1/2)}\left( \frac{|\bxi|^2}{2} \right) \xi_{\langle i_{ 1}}\cdots \xi_{i_{l}\rangle}  \,\dd \bxi,\\
      \ms{B}^{m';t_j}_{m;i_1 \cdots i_l} & = \int_{\bxi\cdot \bn < 0} f_M(\bxi) \bar{L}_{m'}^{(5/2)}\left( \frac{|\bxi|^2}{2} \right)\bar{L}_m^{(l+1/2)}\left( \frac{|\bxi|^2}{2} \right) \xi_{\langle i_{ 1}}\cdots \xi_{i_{l}\rangle} \xi_{t_j} \,\dd \bxi, \\ 
       \ms{C}^{m';t' t''}_{m;i_1 \cdots i_l} & = \int_{\bxi\cdot \bn < 0} f_M(\bxi) \bar{L}_{m'}^{(7/2)}\left( \frac{|\bxi|^2}{2} \right)\bar{L}_m^{(l+1/2)}\left( \frac{|\bxi|^2}{2} \right) \xi_{\langle i_{ 1}}\cdots \xi_{i_{l}\rangle} \cdot \frac{\xi_{\langle t'}\xi_{t''}\xi_{n\rangle}}{\xi_{n}}  \,\dd \bxi, \\
        \ms{D}^{m'}_{m;i_1 \cdots i_l} & = \int_{\bxi\cdot \bn < 0} f_M(\bxi) \bar{L}_{m'}^{(7/2)}\left( \frac{|\bxi|^2}{2} \right)\bar{L}_m^{(l+1/2)}\left( \frac{|\bxi|^2}{2} \right) \xi_{\langle i_{ 1}}\cdots \xi_{i_{l}\rangle} \cdot \frac{\xi_{\langle n}\xi_{n}\xi_{n\rangle}}{\xi_{n}}  \,\dd \bxi
\end{align*}
where $z_{i_1\cdots i_l}^{(k)}$ are the coefficients in the expansion of $\bar{\mathcal{A}}_{\mathrm{oe}} (f_W - \mathcal{P}_{\mathrm{even}}|_{\overline{\mathbb{V}}} \bar{f})$:
\begin{equation}\label{z to u}
 \begin{split}
    z^{(1)}_{n} = &\  \frac{1}{3}\left(\sqrt{5}c^{(1),1}_1(w^{1,W}-w^1) -A_{45}u^{(1)}_{nn}-A_{46}u^{(2)}-A_{48}u^{(2)}_{nn}\right),\\
    z^{(2)}_{n} = &\  \frac{1}{3}\left(\sqrt{5}c^{(2),1}_1 (w^{1,W}-w^1) - A_{57}u^{(1)}_{nn}-\frac{c^{(2),1}_1}{c^{(1),1}_1} A_{46}u^{(2)}-\bar{A}_{78}u^{(2)}_{nn} \right),\\
    z^{(1)}_{t_in} = &\  \frac{2}{15} \left( \frac{3\sqrt{5}}{2}c^{(1),0}_2 (w^{0,W}_{t_i}-w^0_{t_i}) -\frac{1}{2}A_{45} u^{(1)}_{t_i} - \frac{1}{2} A_{57}u^{(2)}_{t_i} -A_{59} u^{(2)}_{t_inn}\right), \quad i=1,2, \\
    z^{(2)}_{t_in} = &\  \frac{2}{15} \left(  \frac{3\sqrt{5}}{2}c^{(2),0}_2 (w^{0,W}_{t_i}-w^0_{t_i})  - \frac{1}{2}A_{48}u^{(1)}_{t_i}-\frac{1}{2}\bar{A}_{78}u^{(2)}_{t_i}- \frac{c^{(2),0}_2}{c^{(1),0}_2}A_{59} u^{(2)}_{t_inn} \right),\quad i=1,2,\\
    z^{(2)}_{t_it_jn} = &\ \frac{2}{35} A_{59}\left[ \frac{2}{15} \delta_{ij} \left(u^{(1)}_{nn}+\frac{c^{(2),0}_2}{c^{(1),0}_2}u^{(2)}_{nn}\right) - \frac{1}{3}\left(u^{(1)}_{t_it_j}+\frac{c^{(2),0}_2}{c^{(1),0}_2}u^{(2)}_{t_it_j}\right)\right], \quad i,j=1,2, \\
    z^{(2)}_{nnn} = &\ \frac{2}{175}A_{59}\left[\left(u^{(1)}_{t_1t_1}+\frac{c^{(2),0}_2}{c^{(1),0}_2}u^{(2)}_{t_1t_1}\right) + \left(u^{(1)}_{t_2t_2}+\frac{c^{(2),0}_2}{c^{(1),0}_2}u^{(2)}_{t_2t_2}\right) -2\left(u^{(1)}_{nn}+\frac{c^{(2),0}_2}{c^{(1),0}_2}u^{(2)}_{nn}\right)\right].
    \end{split}
\end{equation}
Combining \eqref{bc with z} with \eqref{z to u}, we obtain the expressions of the coefficients $\lambda'_{ij}$ as 
\begin{displaymath}
\lambda'_{ij} = \begin{cases}
        -\sqrt{5}(c^{(1),1}_1\gamma'_{i1}  +c^{(2),1}_1 \gamma'_{i2}), & \text{if~} j = 1\\
        -A_{45}\gamma'_{i1} - A_{57}\gamma'_{i2} + \frac{3}{5}A_{59}(3 \gamma'_{i3} - \gamma'_{i4}) ,& \text{if~} j = 2\\
        -A_{46}(\gamma'_{i1} + \frac{c^{(2),1}_1}{c^{(1),1}_1}\gamma'_{i2}) ,& \text{if~} j = 3\\
        -A_{48}\gamma'_{i1} - \bar{A}_{78} \gamma'_{i2} + \frac{c^{(2),0}_2 A_{59}}{5c^{(1),0}_2}(4 \gamma'_{i3} - 3 \gamma'_{i4}),& \text{if~} j = 4
    \end{cases}, \text{~for~} i = 1,2,5;
\end{displaymath}
\begin{displaymath}
\lambda'_{ij} = \begin{cases}
        -3\sqrt{5}(c^{(1),0}_2\gamma'_{i1}  +c^{(2),0}_2 \gamma'_{i2}), & \text{if~} j = 1\\
        -(A_{45}\gamma'_{i1} + A_{48}\gamma'_{i2}) ,& \text{if~} j = 2\\
         -(A_{57}\gamma'_{i1} + \bar{A}_{78}\gamma'_{i2})  ,& \text{if~} j = 3\\
        -2A_{59}(\gamma'_{i1}  +\frac{c^{(2),0}_2}{c^{(1),0}_2} \gamma'_{i2}),& \text{if~} j = 4
    \end{cases}, \text{~for~} i = 3,4;
\end{displaymath}
\begin{displaymath}
\lambda'_{ij} = \begin{cases}
        -A_{59}(\gamma'_{63}+\gamma'_{64}), & \text{if~} (i,j) = (6,1)\\
      \frac{1}{5}A_{59}(\frac{9}{2}\gamma'_{53}+2\gamma'_{63}+7\gamma'_{64}), & \text{if~} (i,j) = (6,2)\\
          -2A_{59}\gamma'_{71}, & \text{if~} (i,j) = (7,1)
    \end{cases}.
\end{displaymath}
In the expressions above, $\gamma'_{ij}$ are given as 
\begin{displaymath}
\gamma'_{ij} = \begin{cases}
       \sum_{m=1}^\infty \sum_{m'=1}^\infty c^{(i),m}_1 c^{(1),m'}_1 (\mathscr{A}^{m'}_{m;n} - \frac{\mathscr{A}^{m'}_{0;n}I^0_{m;n}}{I^0_{0;n}}) , & \text{if~} j = 1\\
        \sum_{m=1}^\infty \sum_{m'=1}^\infty c^{(i),m}_1 c^{(2),m'}_1 (\mathscr{A}^{m'}_{m;n} - \frac{\mathscr{A}^{m'}_{0;n}I^0_{m;n}}{I^0_{0;n}}) ,& \text{if~} j = 2\\
        3\sum_{m=1}^\infty \sum_{m'=0}^\infty c^{(i),m}_1 c^{(2),m'}_3 (\mathscr{C}^{m';t_1t_1}_{m;n} - \frac{\mathscr{C}^{m';t_1t_1}_{0;n}I^0_{m;n}}{I^0_{0;n}}) ,& \text{if~} j = 3\\
        \sum_{m=1}^\infty \sum_{m'=0}^\infty c^{(i),m}_1 c^{(2),m'}_3 (\mathscr{D}^{m'}_{m;n} - \frac{\mathscr{D}^{m'}_{0;n}I^0_{m;n}}{I^0_{0;n}}),& \text{if~} j = 4
    \end{cases}, \text{~for~} i = 1,2;
\end{displaymath}
\begin{displaymath}
\gamma'_{ij} = \begin{cases}
       2\sum_{m=0}^\infty \sum_{m'=0}^\infty c^{(i-2),m}_2 c^{(1),m'}_2 \mathscr{B}^{m';t_1}_{m;t_1n}  , & \text{if~} j = 1\\
        2\sum_{m=0}^\infty \sum_{m'=0}^\infty c^{(i-2),m}_2 c^{(2),m'}_2 \mathscr{B}^{m';t_1}_{m;t_1n} ,& \text{if~} j = 2
    \end{cases}, \text{~for~} i = 3,4;
\end{displaymath}
\begin{displaymath}
\gamma'_{ij} = \begin{cases}
       \sum_{m=0}^\infty \sum_{m'=1}^\infty c^{(2),m}_3 c^{(1),m'}_1 (\mathscr{A}^{m'}_{m;nnn} - \frac{\mathscr{A}^{m'}_{0;n}I^0_{m;nnn}}{I^0_{0;n}}) , & \text{if~} j = 1\\
        \sum_{m=0}^\infty \sum_{m'=1}^\infty c^{(2),m}_3 c^{(2),m'}_1 (\mathscr{A}^{m'}_{m;nnn} - \frac{\mathscr{A}^{m'}_{0;n}I^0_{m;nnn}}{I^0_{0;n}}) ,& \text{if~} j = 2\\
        3\sum_{m=0}^\infty \sum_{m'=0}^\infty c^{(2),m}_3 c^{(2),m'}_3 (\mathscr{C}^{m';t_1t_1}_{m;nnn} - \frac{\mathscr{C}^{m';t_1t_1}_{0;n}I^0_{m;nnn}}{I^0_{0;n}}) ,& \text{if~} j = 3\\
        \sum_{m=0}^\infty \sum_{m'=0}^\infty c^{(2),m}_3 c^{(2),m'}_3 (\mathscr{D}^{m'}_{m;nnn} - \frac{\mathscr{D}^{m'}_{0;n}I^0_{m;nnn}}{I^0_{0;n}}),& \text{if~} j = 4
    \end{cases}, \text{~for~} i = 5;
\end{displaymath}
\begin{displaymath}
\gamma'_{ij} = \begin{cases}
      -\frac{1}{2}\gamma'_{51} , & \text{if~} j = 1\\
        -\frac{1}{2}\gamma'_{52}  ,& \text{if~} j = 2\\
        3\sum_{m=0}^\infty \sum_{m'=0}^\infty c^{(2),m}_3 c^{(2),m'}_3 (\mathscr{C}^{m';t_1t_1}_{m;t_1 t_1 n} - \frac{\mathscr{C}^{m';t_1t_1}_{0;n}I^0_{m;t_1 t_1 n}}{I^0_{0;n}}) ,& \text{if~} j = 3\\
         3\sum_{m=0}^\infty \sum_{m'=0}^\infty c^{(2),m}_3 c^{(2),m'}_3 (\mathscr{C}^{m';t_2t_2}_{m;t_1 t_1 n} - \frac{\mathscr{C}^{m';t_2t_2}_{0;n}I^0_{m;t_1 t_1 n}}{I^0_{0;n}}) ,& \text{if~} j = 4\\
       -\frac{1}{2}\gamma'_{54},& \text{if~} j = 5
    \end{cases}, \text{~for~} i = 6.
\end{displaymath}

\subsection{Expressions of $\lambda_{ij}$}\label{app:phys bc}
\begin{gather*}
    \lambda_{11} = \frac{3\sqrt{5}}{2}c^{(1),1}_1 \kappa_{11}, \ \lambda_{12} = -\frac{\sqrt{2}}{2}\frac{c^{(1),1}_1}{c^{(1),0}_2} \kappa_{12}, \ \lambda_{13} = c^{(1),1}_1 \beta'_0 \kappa_{13}, \ \lambda_{14} = c^{(1),1}_1\beta'_2 \left( \kappa_{14} - \frac{c^{(2),0}_2}{c^{(1),0}_2} \kappa_{12}\right),\\
    \lambda_{21} = \sqrt{5}c^{(1),0}_2 \lambda'_{31} + \sqrt{5} c^{(2),0}_2 \lambda'_{41} , \ \lambda_{22} = -\frac{\sqrt{2}c^{(1),0}_2}{c^{(1),1}_1} \lambda'_{32} - \frac{\sqrt{2}c^{(2),0}_2}{c^{(1),1}_1} \lambda'_{42}, \\ \lambda_{23} = c^{(1),0}_2 \beta'_1 \left( \lambda'_{33} - \frac{c^{(2),1}_1}{c^{(1),1}_1}\lambda'_{32}\right) + c^{(2),0}_2 \beta'_1 \left( \lambda'_{43} - \frac{c^{(2),1}_1}{c^{(1),1}_1}\lambda'_{42}\right), \ \lambda_{24} = c^{(1),0}_2 \beta'_3\lambda'_{34} + c^{(2),0}_2 \beta'_3\lambda'_{44},\\
    \lambda_{31} = -\sqrt{\frac{5}{2}}\frac{\lambda'_{41}}{\beta'_2}, \  \lambda_{32} = \frac{\lambda'_{42}}{c^{(1),1}_1 \beta'_2}, \ \lambda_{33} = -\frac{\sqrt{2}}{2}\frac{\beta'_1}{\beta'_2}\left( \lambda'_{43} - \frac{c^{(2),1}_1}{c^{(1),1}_1} \lambda'_{42} \right), \ \lambda_{34} = -\frac{\sqrt{2}}{2}\frac{\beta'_3}{\beta'_2} \lambda'_{44},\\ 
    \lambda_{41} = -\frac{3\sqrt{10}}{2\beta'_3 \mu_1} \left( \kappa_{21} + \frac{c^{(2),1}_1 \mu_2}{c^{(1),1}_1} \kappa_{11} \right), \ \lambda_{42} = \frac{1}{ c^{(1),0}_2 \beta'_3 \mu_1 } \left( \kappa_{22} + \frac{c^{(2),1}_1 \mu_2}{c^{(1),1}_1} \kappa_{12} \right), \\
    \lambda_{43} = \frac{-\sqrt{2} \beta'_0 }{\beta'_3 \mu_1}\left( \kappa_{23} + \frac{c^{(2),1}_1 \mu_2}{c^{(1),1}_1} \kappa_{13} \right),
    \lambda_{44} =  \frac{-\sqrt{2}\beta'_2}{\beta'_3 \mu_1}\left[\kappa_{24} - \frac{c^{(2),0}_2 }{c^{(1),0}_2}\kappa_{22} + \frac{c^{(2),1}_1 \mu_2}{c^{(1),1}_1}\left(\kappa_{14} - \frac{c^{(2),0}_2}{c^{(1),0}_2}\kappa_{12}\right) \right],  \\   \lambda_{45} =  \left[ 1 + \left(\frac{c^{(2),1}_1}{c^{(1),1}_1}\right)^2 \right]\frac{\beta'_1 \mu_2}{\beta'_3 \mu_1}, \ 
    \lambda_{51} = \frac{1}{c^{(1),0}_2} \frac{\lambda'_{61}}{\beta'_3}, \  \lambda_{52} = \frac{1}{c^{(1),0}_2} \frac{\lambda'_{62}}{\beta'_3},\ \lambda_{61} = \frac{1}{c^{(1),0}_2} \frac{\lambda'_{71}}{\beta'_3}. 
\end{gather*}
Here $\mu_i$ and $\kappa_{ij}$ are respectively given in \eqref{mu} and \eqref{kappa}. We remark that for Maxwell molecules $\lambda_{23} = \lambda_{33}=0$.

\section{Tables of coefficients in the inverse-power-law model}
In this section, we list the value of $\beta_i$, $\mathscr{L}^{(mn)}_l$ and $c^{(1),*}_*$ appearing in the moment equations \eqref{r13eq:heat flux}\eqref{r13eq:stress tensor} as Table \ref{tab:eq}, and $\lambda_{ij}$ in the boundary conditions \eqref{bc phys 2}--\eqref{bc phys 7} as Table \ref{tab:bc} for various choices of parameter $\eta$ in the inverse-power-law model that we considered in Section \ref{sec:results 1d}.    
\begin{table}[h] 
\centering \small
\begin{displaymath}
\begin{array}{| c || c | c | c | c | c | c | c | c | }
\hline
 \rule{0pt}{12pt}
\eta &   \mathscr{L}^{(22)}_0   & \mathscr{L}^{(11)}_1   & \mathscr{L}^{(12)}_1 & \mathscr{L}^{(22)}_1 & \mathscr{L}^{(11)}_2  & \mathscr{L}^{(12)}_2 & \mathscr{L}^{(22)}_2 & \mathscr{L}^{(22)}_3 \\
\hline \rule{0pt}{10pt}
5 & -\frac{2}{3} & -2 & 0 & 0 & -\frac{15}{2} & 0 & -\frac{35}{4} & -\frac{105}{4} \\
10 & -0.6494 & -1.9746 & 0.1416 & -2.9453 & -7.4320 & 0.4842 & -8.7125 & -25.6149 \\ 
\infty & -0.6124 & -1.9180 & 0.2492 & -2.8154 & -7.2778 & 0.8693 & -8.4235 & -24.1732 \\ 
\hline
\eta &   \beta_0   & \beta_1   & \beta_2 & \beta_3 & \beta_4 & - & - & - \\
\hline
 \rule{0pt}{10pt}
5 & -6  & 0 & -\frac{36}{5} & -15 & \frac{2}{5} & - & - & -  \\
10 & -5.0643 & -0.0256 & -6.4997 & -14.2799 & 0.3541 & - & - & -   \\ 
\infty & -4.5165 & -0.0741 & -6.1729 & -14.1228 & 0.3192 & - & - & -  \\ 
\hline
 \rule{0pt}{12pt}
\eta &   c^{(1),1}_1   &c^{(1),0}_2   &- &- & - & - & - & - \\
\hline
5 &  1 & 1 & - & - & - & - & - & -  \\
10 & 0.9974 & 0.9979 &- & - & - & - & - & -   \\ 
\infty & 0.9915 & 0.9929 &- & - & - & - & - & -  \\
\hline
\end{array}
\end{displaymath}
\caption{Coefficients in moment equations.} \label{tab:eq}
\end{table}

\begin{table}[h] 
\centering \small
\begin{displaymath}
\begin{array}{| c || c | c | c | c | c | }
\hline
\eta &   \lambda_{11}   & \lambda_{12}   & \lambda_{13} & \lambda_{14} & -  \\
\hline
5 & 0.7979 & 0.1995 & -0.6383 & -0.7660 & - \\ 
10 & 0.7968 & 0.1912 & -0.5394 & -0.6949 & -   \\ 
\infty & 0.7946 & 0.1863 & -0.4826 & -0.6633 & -  \\ 
\hline
\eta &   \lambda_{21}   & \lambda_{22}  & \lambda_{23}& \lambda_{24} & -\\
5 & 0.3989 & 0.0798 & 0 & -0.7979 & - \\
10 & 0.3992 & 0.0781 & -0.0007 & -0.7618 & - \\ 
\infty & 0.3988 & 0.0789 & -0.0020 & -0.7563 & -  \\ 
\hline
\eta &  \lambda_{31}   & \lambda_{32}  & \lambda_{33} & \lambda_{34} & -  \\
\hline
5 & 0.0831  & -0.1829 & 0 & -0.1662 & -   \\
10 & 0.0729 & -0.2071 & -0.0001 & -0.1390 & -    \\ 
\infty & 0.0583 & -0.2186 & -0.0003 & -0.1105 & -  \\ 
\hline
\eta &  \lambda_{41}   & \lambda_{42}  & \lambda_{43} & \lambda_{44} & \lambda_{45}  \\
\hline
5 & 0.0798  & -0.2793 & -0.0638 & -0.0766 & 0   \\
10 & 0.0636 & -0.2994 & -0.0430 & -0.0638 &  0.0012   \\ 
\infty & 0.0473 & -0.3051 & -0.0287 & -0.0548 & 0.0035  \\ 
\hline
\eta &  \lambda_{51}   & \lambda_{52}  & - & - & -  \\
\hline
5 & -0.1995  & -0.0997 & - & - & -   \\
10 & -0.2098 & -0.1049 & - &  -& -    \\ 
\infty & -0.2108 & -0.1054 & - & - & -  \\ 
\hline
\eta &  \lambda_{61}   & -  & - & - & -  \\
\hline
5 & -0.1995  & - & - & - & -   \\
10 & -0.2098 &  & - &  -& -    \\ 
\infty & -0.2108 &  & - & - & -  \\ 
\hline
\end{array}
\end{displaymath}
\caption{Coefficients in boundary conditions.} \label{tab:bc}
\end{table}

\end{document}